\documentclass[11pt]{article}

\usepackage{graphicx, subfigure}
\usepackage[top=1.25in,bottom=1.25in,left=1in,right=1in]{geometry}
\usepackage{amsmath,amsfonts,amsthm,amscd,amssymb,bm,epsfig,epsf,bbm,constants,amssymb, prettyref, mathabx}
\usepackage{algorithm,algorithmic}
\usepackage[usenames]{color}
 
\definecolor{darkred}{RGB}{100,0,0}
\definecolor{darkgreen}{RGB}{0,100,0}
\definecolor{darkblue}{RGB}{0,0,150}
\definecolor{red}{RGB}{255,0,0}

\usepackage{hyperref}
\hypersetup{colorlinks=true, linkcolor=darkred, citecolor=darkgreen, urlcolor=darkblue}
\usepackage{url}
\usepackage[numbers,square]{natbib}

\newtheorem{theorem}{Theorem}[section]
\newtheorem{lemma}{Lemma}[section]
\newtheorem{corollary}{Corollary}[section]
\newtheorem{assumption}{Assumption}[section]
\newtheorem{proposition}{Proposition}[section]

\theoremstyle{remark}
\newtheorem{remark}{Remark}[section]

\usepackage{booktabs} 
\usepackage{multirow}
\usepackage{longtable}

\newcommand{\wt}[1]{\widetilde{#1}}
\newcommand{\wh}[1]{\widehat{#1}}
\newcommand{\calP}{\mathcal{P}}
\newcommand{\calM}{\mathcal{M}}
\newcommand{\Poff}{P^0}
\newcommand{\fnorm}[1]{\|#1\|_{\mathrm{F}}}
\newcommand{\opnorm}[1]{\|#1\|_{\mathrm{op}}}

\usepackage{tcolorbox}
\tcbuselibrary{skins, listings, theorems, breakable, hooks}

\newcommand{\logit}{\operatorname{logit}}

\newcommand{\rank}{\operatorname{r}}

\newcommand{\row}{\mathrm{row}}
\newcommand{\col}{\mathrm{col}}

\newcommand{\srank}{\operatorname{r}_{\operatorname{stable}}}

\newcommand{\tr}{\operatorname{Tr}}

\newcommand{\E}{{\bf \mathbb{E}}}

\newcommand{\one}{{1_n}}
\newcommand{\fro}[1]{\left\| #1 \right\|_{\operatorname{F}}}
\newcommand{\op}[1]{\left\| #1 \right\|_{\operatorname{op}}}
\newcommand{\vnorm}[1]{\left\| #1 \right\|}

\newcommand{\nuc}[1]{\left\| #1 \right\|_{*}}

\newcommand{\tall}{\bTheta_{\star} }
\newcommand{\tz}{\bZ_{\star} }
\newcommand{\tg}{\bG_{\star} }
\newcommand{\ta}{\balpha_{\star} }
\newcommand{\tb}{\beta_{\star} }
\newcommand{\tx}{\bx_{\star} }

\newcommand{\dg}{\Delta_{\hbG}}
\newcommand{\da}{\Delta_{\hbalpha}}
\newcommand{\db}{\Delta_{\widehat{\beta}}}
\newcommand{\dall}{\Delta_{\hbTheta}}

\newcommand{\zt}{\bZ^t}
\newcommand{\at}{\balpha^t}
\newcommand{\bt}{\beta^t}
\newcommand{\gt}{\bG^t}
\newcommand{\allt}{\bTheta^t}
\newcommand{\xt}{x^t}

\newcommand{\dzt}{\Delta_{\bZ^t}}

\newcommand{\dzzt}{\bZ^t(\bZ^t)^\top - \tz\tz^\top}

\newcommand{\dat}{\Delta_{\balpha^t}}
\newcommand{\dbt}{\Delta_{\beta^t}}
\newcommand{\dgt}{\Delta_{\bG^t}}
\newcommand{\dallt}{\Delta_{\bTheta^t}}

\newcommand{\wl}{h}
\newcommand{\wll}{H}

\newcommand{\sam}{\nabla \wl (\bTheta)}
\newcommand{\samt}{\nabla \wl (\bTheta^t)}

\newcommand{\popt}{\nabla \wll (\bTheta^t)}
\newcommand{\popstar}{\nabla \wll (\bTheta_{\star})}

\newcommand{\sigt}{\sigma (\bTheta^t)}
\newcommand{\sigp}{\sigma (\bTheta_{\star})}

\newcommand{\etaz}{\eta_{\bZ}}
\newcommand{\etaa}{\eta_{\balpha}}
\newcommand{\etab}{\eta_{\beta}}
\newcommand{\etag}{\eta_{\bG}}

\newcommand{\param}{\mathcal{F}(n, k, M_1, M_2, \bX)}

\newcommand{\paramK}{\mathcal{F}_g(n, M_1, M_2, \bX)}

\newcommand{\lambdaP}{C_0\sqrt{\max\left\{ne^{-M_2}, \log n\right\}}}

\newcommand{\gresid}{\widebar{G}_k}
\newcommand{\pint}{\mathbb{N}_+}

\newcommand{\lambdacond}{\max\{ 2\op{\bA-\bP}, ~|\langle \bA-\bP, \bX/\fro{X} \rangle |/\sqrt{k}, ~1\}}

\newcommand{\half}{1/2}

\newcommand{\dtg}{\Delta_{\wt{G}}}
\newcommand{\dta}{\Delta_{\wt{\alpha}}}
\newcommand{\dtb}{\Delta_{\wt{\beta}}}
\newcommand{\dtall}{\Delta_{\wt{\Theta}}}
\newcommand{\wtl}{\wt{\wl}}

\newcommand{\inner}[1]{\big\langle #1 \big\rangle}

\newcommand{\argmin}{\mathop{\rm argmin}}

\newcommand{\nb}[1]{\textcolor{blue}{\texttt{[#1]}}}
\newcommand{\zm}[1]{\textcolor{red}{\texttt{[#1]}}}

\def\argmin{\mathop{\rm arg\, min}}

\newcommand{\bel}{\begin{eqnarray}\label}
\newcommand{\eel}{\end{eqnarray}}
\newcommand{\bes}{\begin{eqnarray*}}
\newcommand{\ees}{\end{eqnarray*}}
\newcommand{\bei}{\begin{itemize}}
\newcommand{\beiftnt}{\begin{itemize}\footnotesize}
\newcommand{\eei}{\end{itemize}}

\def\benu{\begin{enumerate}}
\def\eenu{\end{enumerate}}

\def\argmin{\mathop{\rm arg\, min}}

\def\S{{\mathbb{S}}}
\def\E{{\mathbb{E}}}
\def\P{{\mathbb{P}}}

\def\complex{\mathop{{\rm I}\kern-.58em\hbox{\rm C}}\nolimits}

\def\rank{\hbox{\rm rank}}

\def\mathbold{\boldsymbol} 

\def\bA{\mathbold{A}}

\def\bB{\mathbold{B}}

\def\calC{{\cal C}}

\def\calC{{\cal C}}

\def\calF{{\cal F}}

\def\bG{\mathbold{G}}
\def\hbG{{\widehat{\bG}}}

\def\bI{\mathbold{I}}

\def\bJ{\mathbold{J}}

\def\bM{\mathbold{M}}
\def\calM{{\cal M}}

\def\bP{\mathbold{P}}
\def\calP{{\cal P}}

\def\bR{\mathbold{R}}

\def\bx{\mathbold{x}}

\def\bX{\mathbold{X}}

\def\bX{{X}}
\def\by{\mathbold{y}}

\def\bY{\mathbold{Y}}

\def\bz{\mathbold{z}}

\def\bZ{\mathbold{Z}}

\def\balpha{\mathbold{\alpha}}

\def\hbalpha{{\widehat{\balpha}}}

\def\hbeta{\widehat{\beta}}

\def\bTheta{\mathbold{\Theta}}\def\hbTheta{{\widehat{\bTheta}}}

\def\0{\mathbold{0}}

\def\bx{{x}}
\def\by{{y}}
\def\bt{{\beta^t}}
\def\bX{{X}}
\def\bY{{Y}}
\def\bB{{B}}
\def\bR{{R}}
\def\bZ{{Z}}
\def\bI{{I}}
\def\bM{{M}}

\def\bA{{A}}
\def\bG{{G}}
\def\bTheta{{\Theta}}

\def\bP{{P}}
\def\bJ{{J}}
\def\balpha{{\alpha}}
\def\bz{{z}}
\def\rank{\operatorname{r}}

\title{Exploration of Large Networks with Covariates via Fast and Universal
Latent Space Model Fitting}
\author{Zhuang Ma
~and~ 
Zongming Ma\footnote{
Email: \texttt{zongming@wharton.upenn.edu}.}
\\~\\
\textit{University of Pennsylvania}
}
\date{~}

\begin{document}
\maketitle

\begin{abstract}
Latent space models are effective tools for statistical modeling and exploration of network data.
These models can effectively model real world network characteristics such as degree heterogeneity, transitivity, homophily, etc. 
Due to their close connection to generalized linear models, it is also natural to incorporate covariate information in them.
The current paper presents two universal fitting algorithms for networks with edge covariates: one based on nuclear norm penalization and the other based on projected gradient descent.
Both algorithms are motivated by maximizing likelihood for a special class of inner-product models
while working simultaneously for a wide range of different latent space models, such as distance models, which
allow latent vectors to affect edge formation in flexible ways. 
These fitting methods, especially the one based on projected gradient descent, are fast and scalable to large networks.
We obtain their rates of convergence for both inner-product models and beyond.
The effectiveness of the modeling approach and fitting algorithms is demonstrated on five real world network datasets for different statistical tasks, including community detection with and without edge covariates, and network assisted learning.

\smallskip

\end{abstract}

%
%
%
%
%



\section{Introduction}  
\label{sec:intro}
  
Network is a prevalent form of data for quantitative and qualitative analysis in a number of fields, including but not limited to sociology, computer science, neuroscience, etc.
Moreover, due to advances in science and technology, the sizes of the networks we encounter are ever increasing.
Therefore, to explore, to visualize and to utilize the information in large networks poses significant challenges to Statistics.
Unlike traditional datasets in which a number of features are recorded for each subject,  
network datasets provide information on the relation among all subjects under study, sometimes together with additional features.
In this paper, we focus on the modeling, visualization and exploration of networks in which additional features might be observed for each node pair.

On real world networks, people oftentimes observe the following characteristics.
First, the degree distributions of nodes are often right-skewed and so networks exhibit degree heterogeneity.
In addition, connections in networks often demonstrate transitivity, that is nodes with common neighbors are more likely to be connected.
Moreover, nodes that are similar in certain ways (students in the same grade, brain regions that are close physically, etc.) are more likely to form bonds.
Such a phenomenon is usually called homophily in network studies.
Furthermore, nodes in some networks exhibit clustering effect and in such cases it is desirable to partition the nodes into different communities.


An efficient way to explore network data and to extract key information is to fit appropriate statistical models on them.
To date, there have been a collection of network models proposed by researchers in various fields. 
These models aim to catch different subsets of the foregoing characteristics, and \citet{goldenberg2010survey} provides a comprehensive overview.
An important class of network models are \emph{latent space models} \citep{hoff2002latent}.
Suppose there are $n$ nodes in the observed network.
The key idea underlying latent space modeling is that each node $i$ can be represented by a vector $z_i$ in some low dimensional Euclidean space (or some other metric space of choice, see e.g.~\cite{krioukov2010hyperbolic,asta2014geometric} for latent spaces with negative curvature)
that is sometimes called the social space, 
and nodes that are ``close'' in the social space are more likely to be connected.
\citet{hoff2002latent} considered two types of latent space models: distance models and projection models. 
In both cases, the latent vectors $\{z_i\}_{i=1}^n$ were treated as fixed effects.
Later, a series of papers \citep{hoff2003random,handcock2007model,hoff2008modeling,krivitsky2009representing} generalized the original proposal in \cite{hoff2002latent} for better modeling of other characteristics of social networks, such as clustering, degree heterogeneity, etc.
In these generalizations, the $z_i$'s were treated as random effects generated from certain multivariate Gaussian mixtures.
Moreover, model fitting and inference in these models has been carried out via Markov Chain Monte Carlo, and it is difficult to scale these methodologies to handle large networks \citep{goldenberg2010survey}.
In addition, one needs to use different likelihood function based on choice of model and there is little understanding of the quality of fitting when the model is mis-specified.
Albeit these disadvantages, latent space models are attractive due to their friendliness to interpretation and visualization.

For concreteness, assume that we observe an undirected network 
represented by a symmetric adjacency matrix $\bA$ on $n$ nodes with $A_{ij} = A_{ji} = 1$ if nodes $i$ and $j$ are connected and zero otherwise. 
In addition, we may also observe a symmetric pairwise covariate matrix $\bX$ which measures certain characteristics of node pairs. 
We do not allow self-loop and so we set the diagonal elements of the matrices $\bA$ and $\bX$ to be zeros. 
The covariate $\bX_{ij}$ can be binary, such as an indicator of whether nodes $i$ and $j$ share some common attribute (e.g. gender, location, etc) or it can be continuous, such as a distance/similarity measure (e.g. difference in age, income, etc).
It is straightforward to generalize the methods and theory in this paper to a finite number of covariates.

In this paper, we aim to tackle the following two key issues in latent space modeling of network data:
\begin{itemize}
	\item First, we seek a class of latent space models that is special enough so that we can design fast fitting algorithms for them and hence be able to handle networks of very large sizes;

\item In addition, we would like to be able to fit a class of models that are flexible enough to well approximate a wide range of latent space models of interest so that fitting methods for this flexible class continue to work even when the model is mis-specified.
\end{itemize}
From a practical viewpoint,  
if one is able to find such a class of models and design fast algorithms for fitting them, then one would be able to use this class as working models and to use the associated fast algorithms to effectively explore large networks.



\subsection{Main contributions} 
We make progress on tackling the foregoing two issues simultaneously in the present paper, which we summarize as the following main contributions:

\begin{enumerate}
\item We first consider a special class of latent space models, called inner-product models, and design two fast fitting algorithms for this class. 
Let the observed $n$-by-$n$ adjacency matrix and covariate matrix be $A$ and $X$, respectively. The inner-product model assumes that for any $i< j$, 
\begin{equation}
\begin{gathered}
\bA_{ij} = \bA_{ji} \stackrel{ind.}{\sim} \text{Bernoulli} (\bP_{ij}),\quad \text{with}\\
\text{logit}(\bP_{ij})=\bTheta_{ij}=\balpha_i+\balpha_j+\beta\bX_{ij}+\bz_i^\top\bz_j,
\label{eq:projection}
\end{gathered}
\end{equation}
where for any $x\in (0,1)$, $\text{logit}(x) = \log[{x}/{(1-x)}]$.
Here, $\alpha_i$, $1\leq i\leq n$, are parameters modeling degree heterogeneity. The parameter $\beta$ is the coefficient for the observed covariate, and $\bz_i^\top\bz_j$ is the inner-product between the latent vectors. 
As we will show later in \prettyref{sec:model}, this class of models
can incorporate degree heterogeneity, transitivity and homophily explicitly.
From a matrix estimation viewpoint, the matrix $G = (G_{ij}) = (z_i^\top z_j)$ is of rank at most $k$ that can be much smaller than $n$.
Motivated by recent advances in low rank matrix estimation, we design two fast algorithms for fitting \eqref{eq:projection}.
One algorithm is based on lifting and nuclear norm penalization of the negative log-likelihood function.
The other is based on directly optimizing the negative log-likelihood function via projected gradient descent.
For both algorithms, we establish high probability error bounds for inner-product models.
The connection between model \eqref{eq:projection} and the associated algorithms and other related work in the literature will be discussed immdediately in next subsection.

\item We further show that these two fitting algorithms are ``universal'' in the sense that they can work simultaneously for a wide range of latent space models beyond the inner-product model class.
For example, they work for the distance model and the Gaussian kernel model in which the inner-product term $\bz_i^\top\bz_j$ in \eqref{eq:projection} is replaced with $-\|z_i-z_j\|$ and $c \exp(-\|z_i-z_j\|^2/\sigma^2)$, respectively. 
Thus, the class of inner-product models is flexible and can be used to approximate many other latent space models of interest.
In addition, the associated algorithms can be applied to networks generated from a wide range of mis-specified models and still yield reasonable results.
The key mathematical insight that enables such flexibility is introduced in \prettyref{sec:model} as the Schoenberg Condition \eqref{eq:h-constraint}.

\item We demonstrate the effectiveness of the model and algorithms on real data examples. 
In particular, we fit inner-product models by the proposed algorithms
on five different real network datasets for several different tasks, including visualization, clustering and network-assisted classification.
On three popular benchmark datasets for testing community detection on networks, 
a simple $k$-means clustering on the estimated latent vectors obtained by our algorithm yields as good result on one dataset and better results on the other two when compared with four state-of-the-art methods.
The same ``model fitting followed by $k$-means clustering'' approach also yields nice clustering of nodes on a social network with edge covariates.
Due to the nature of latent space models, for all datasets on which we fit the model, we obtain natural visualizations of the networks by plotting latent vectors.
Furthermore, we illustrate how network information can be incorporated in traditional learning problems using a document classification example.

%
%


\end{enumerate}

A Matlab implementation of the methodology in the present paper is available upon request.

%
%
%
%
%
%
%


\subsection{Other related works} 
The current form of the inner-product model \eqref{eq:projection} has previously appeared in \citet{hoff2003random} and \citet{hoff2005bilinear},
though the parameters were modeled as random effects rather than fixed values, and Bayesian approaches were proposed for estimating variance parameters of the random effects.
\citet{hoff2008modeling} proposed a latent eigenmodel which has a probit link function as opposed to the logistic link function in the present paper.
As in \cite{hoff2003random} and \cite{hoff2005bilinear}, parameters were modeled as random effects and model fitting was through Bayesian methods.
It was shown that the eigenmodel weakly generalizes the distance model in the sense that the order of the entries in the latent component can be preserved.
This is complementary to our results which aim to approximate the latent component directly in some matrix norm.
An advantage of the eigenmodel is its ability to generalize the latent class model, whereas the inner-product model \eqref{eq:projection} and the more general model we shall consider generalize a subset of latent class models due to the constraint that the latent component (after centering) is positive semi-definite. 
We shall return to this point later in \prettyref{sec:discuss}.
\citet{young2007random} proposed a random dot product model, which can be viewed as an inner-product model with identity link function.
The authors studied a number of properties of the model, such as diameter, clustering and degree distribution.
\citet{tang2013universally} studied properties of the leading eigenvectors of the adjacency matrices of latent positions graphs (together with its implication on classification in such models) where the connection probability of two nodes is the value that some universal kernel \cite{micchelli2006universal} takes on the associated latent positions and hence generalizes the random dot product model.
This work is close in spirit to the present paper. 
However, there are several important differences. 
First, the focus here is model fitting/parameter estimation as opposed to classification in \cite{tang2013universally}.
In addition, any universal kernel considered in \cite{tang2013universally} satisfies the Schoenberg condition \eqref{eq:h-constraint} and thus is covered by the methods and theory of the present paper, and so we cover a broader range of models that inner-product models can approximate.
This is also partially due to the different inference goals.
Furthermore, we allow the presence of observed covariates while \cite{tang2013universally} did not.



When fitting a network model, we are essentially modeling and estimating the edge probability matrix. 
From this viewpoint, the present paper is related to the literautre on graphon estimation and edge probability matrix estimation for block models.
See, for instance, \cite{bickel2009nonparametric,airoldi2013stochastic,wolfe2013nonparametric,gao2015rate,klopp2015oracle,gao2016optimal} and the references therein. 
However, the block models have stronger structural assumptions than the latent space models we are going to investigate.

The algorithmic and theoretical aspects of the paper is also closely connected to the line of research on low rank matrix estimation, which plays an important role in many applications such as phase retrieval \cite{candes2015phase1,candes2015phase} and matrix completion \cite{candes2010power,keshavan2010matrix,keshavan2010matrix1,candes2012exact,koltchinskii2011nuclear}. 
Indeed, the idea of nuclear norm penalization has originated from matrix completion for both general entries \cite{candes2010power} and binary entries \cite{davenport20141}.
In particular, our convex approach can be viewed as a Lagrangian form of the proposal in \cite{davenport20141} when there is no covariate and the matrix is fully observed.
We have nonetheless decided to spell out details on both method and theory for the convex approach because the matrix completion literature typically does not take into account the potential presence of observed covariates.
On the other hand, the idea of directly optimizing a non-convex objective function involving a low rank matrix has been studied recently in a series of papers. See, for instance, \cite{ma2013sparse,burer2005local,sun2016guaranteed,tu2015low,chen2015fast,lafferty2016convergence,ge2016matrix} and the references therein.
Among these papers, the one that is the most related to the projected gradient descent algorithm we are to propose and analyze is \cite{chen2015fast} which focused on estimating a positive semi-definite matrix of exact low rank in a collection of interesting problems.
Another recent and related work \cite{levina2017} has appeared after the initial posting of the present manuscript.
However, we will obtain tighter error bounds for latent space models and we will go beyond the exact low rank scenario.

%


\subsection{Organization}
After a brief introduction of standard notation used throughout the paper, the rest of the paper is organized as follows.
\prettyref{sec:model} introduces both inner-product models and a broader class of latent space models on which our fitting methods work.
The two fitting methods are described in detail in \prettyref{sec:fitting}, followed by their theoretical guarantees in \prettyref{sec:theory} under both inner-product models and the general class.
The theoretical results are further corroborated by simulated examples in \prettyref{sec:simu}.
\prettyref{sec:data} demonstrates the competitive performance of the modeling approach and fitting methods on five different real network datasets.
We discuss interesting related problems in \prettyref{sec:discuss} and present proofs of the main results in \prettyref{sec:proof}.
Technical details justifying the initialization methods for the project gradient descent approach are deferred to the appendix.
 
\paragraph{Notation} 
For $A = (A_{ij})\in \mathbb{R}^{n\times n}$, $\mathrm{Tr}(A) = \sum_{i=1}^n A_{ii}$ stands for its trace.
For $\bX, \bY\in\mathbb{R}^{m\times n}$, 
$\inner{\bX, \bY}=\mathrm{Tr}(\bX^\top\bY)$ defines an inner product between them. 
If $m\geq n$, for any matrix $\bX$ with singular value decomposition $\bX = \sum_{i=1}^n \sigma_i u_i v_i^\top$, $\|\bX\|_* = \sum_{i=1}^n \sigma_i$, $\fnorm{\bX} = \sqrt{\sum_{i=1}^n \sigma_i^2}$ and $\opnorm{\bX} = \max_{i=1}^n \sigma_i$ stand for the nuclear norm, the Frobenius norm and the operator norm of the matrix, respectively. Moreover,
$\bX_{i*}$ and $\bX_{*j}$ denote the $i$-th row and $j$-th column of $\bX$,
and for any function $f$, $f(\bX)$ is the shorthand for applying $f(\cdot)$ element-wisely to $\bX$, that is $f(\bX)\in\mathbb{R}^{m\times n}$ and $[f(\bX)]_{ij}=f(\bX_{ij})$. 
Let $\S_+^n$ be the set of all $n\times n$ positive semidefinite matrices and $O(m, n)$ be the set of all $m\times n$ orthonormal matrices. 
For any convex set $\calC$, $\bP_{\calC}(\cdot)$ is the projection onto the $\calC$. 

\section{Latent space models}
\label{sec:model}

In this section, we first give a detailed introduction of the inner-product model \eqref{eq:projection} and conditions for its identifiability.
In addition, we introduce a more general class of latent space models that includes the inner-product model as a special case.
The methods we propose later will be motivated by the inner-product model and can also be applied to the more general class.

\subsection{Inner-product models}

Recall the inner-product model defined in \eqref{eq:projection}, i.e., for any observed $A$ and $X$ and any $i< j$, 
\begin{equation*}
\bA_{ij} = \bA_{ji} \stackrel{ind.}{\sim} \text{Bernoulli} (\bP_{ij}),\quad \text{with}\quad
\text{logit}(\bP_{ij})=\bTheta_{ij}=\balpha_i+\balpha_j+\beta\bX_{ij}+\bz_i^\top\bz_j.
\end{equation*}
Fixing all other parameters, if we increase $\alpha_i$, then node $i$ has higher chances of connecting with other nodes. 
Therefore, the $\alpha_i$'s model degree heterogeneity of nodes and we call them degree heterogeneity parameters.
Next, the regression coefficient $\beta$ moderates the contribution of covariate to edge formation.
For instance, if $\bX_{ij}$ indicates whether nodes $i$ and $j$ share some common attribute such as gender, 
then a positive $\beta$ value implies that nodes that share common attribute are more likely to connect. 
Such a phenomenon is called \emph{homophily} in the social network literaute.
Last but not least, the latent variables $\{z_i\}_{i=1}^n$ enter the model through their inner-product $z_i^\top z_j$, and hence is the name of the model.
We impose no additional structural/distributional assumptions on the latent variables for the sake of modeling flexibility.

We note that model \eqref{eq:projection} also allows the latent variables to enter the second equation in the form of $g(z_i,z_j) = -\frac{1}{2}\|z_i-z_j\|^2$. 
To see this, note that $g(z_i,z_j) = -\frac{1}{2}\|z_i\|^2 - \frac{1}{2}\|z_j\|^2 + z_i^\top z_j$, and we may re-parameterize by setting $\tilde{\alpha}_i = \alpha_i - \frac{1}{2}\|z_i\|^2$ for all $i$. 
Then we have
\begin{equation*}
\Theta_{ij} = \alpha_i + \alpha_j + \beta X_{ij} - \frac{1}{2}\|z_i-z_j\|^2 = 
\tilde{\alpha}_i + \tilde{\alpha}_j + \beta X_{ij} + z_i^\top z_j.
\end{equation*}
An important implication of this observation is that the function $g(z_i,z_j) = -\frac{1}{2}\|z_i-z_j\|^2$ directly models \emph{transitivity}, i.e., nodes with common neighbors are more likely to connect since their latent variables are more likely to be close to each other in the latent space. 
In view of the foregoing discussion, the inner-product model \eqref{eq:projection} also enjoys this nice modeling capacity.

In matrix form, we have
\begin{equation}
\bTheta=\balpha\one^\top+\one\balpha^\top+\beta\bX+\bG
\label{eq:Theta}
\end{equation}
where $\one$ is the all one vector in $\mathbb{R}^n$ and $G = \bZ \bZ^\top$ with $\bZ=(\bz_1, \cdots, \bz_n)^\top \in \mathbb{R}^{n\times k}$. 
Since there is no self-edge and $\bTheta$ is symmetric, only the upper diagonal elements of $\bTheta$ are well defined, which we denote by $\bTheta^u$.  
Nonetheless we define the diagonal element of $\Theta$ as in \eqref{eq:Theta} since it is inconsequential. 
To ensure identifiability of model parameters in \eqref{eq:projection}, we assume the latent variables
are centered,  that is 
\begin{equation}
	\label{eq:J}
\bJ\bZ=\bZ\quad \mbox{where} \quad \bJ=\bI_n-\frac{1}{n}\one\one^\top.
\end{equation}
Note that this condition uniquely identifies $Z$ up to a common orthogonal transformation of its rows while $G = ZZ^\top$ is now directly identifiable.

\subsection{A more general class and the Schoenberg condition}
\label{sec:generalmodel}

Model \eqref{eq:projection} is a special case of a more general class of latent space models, which can be defined by
\begin{equation}
\begin{gathered}
\bA_{ij} = \bA_{ji} \stackrel{ind.}{\sim} \text{Bernoulli} (\bP_{ij}),\quad \text{with}\\
\text{logit}(\bP_{ij})=\bTheta_{ij}=\tilde\balpha_i+ \tilde\balpha_j+\beta\bX_{ij}+ \ell(\bz_i,\bz_j)
\label{eq:kernel-model}
\end{gathered}
\end{equation}
where $\ell(\cdot,\cdot)$ is a smooth symmetric function on $\mathbb{R}^k\times \mathbb{R}^k$.
We shall impose an additional constraint on $\ell$ following the discussion below.
In matrix form, for $\tilde\alpha = (\tilde\alpha_1,\dots,\tilde\alpha_n)'$ and $L = (\ell(z_i,z_j))$, we can write 
\begin{equation*}
\Theta = \tilde\alpha\one^\top + \one \tilde\alpha^\top + \beta X + L.
\end{equation*}
To better connect with \eqref{eq:Theta}, let 
\begin{equation}
G = J L J,\quad \mbox{and} \quad
\alpha \one^\top + \one\alpha^\top = \tilde\alpha\one^\top + \one \tilde\alpha^\top + L - J L J.
\label{eq:G-kernel}
\end{equation}
Note that the second equality in the last display holds since the expression on its righthand side is symmetric and of rank at most two.
Then we can rewrite the second last display as 
\begin{equation}
\label{eq:Theta-kernel}
\Theta = \alpha \one^\top + \one \alpha^\top + \beta X + G
\end{equation}
which reduces to \eqref{eq:Theta} and $G$ satisfies $JG = G$. 
Our additional constraint on $\ell$ is the following Schoenberg condition:
\begin{equation}
\label{eq:h-constraint}
\begin{aligned}
&\mbox{For any positive integer $n\geq 2$ and any $z_1,\dots, z_n\in \mathbb{R}^k$,}	\\
&\mbox{$G = J L J$ is positive semi-definite for $L = (\ell(z_i,z_j))$ and $J = I_n - \frac{1}{n}\one\one^\top$.
}
\end{aligned}
\end{equation}
Condition \eqref{eq:h-constraint} may seem abstract, while the following lemma elucidates two important classes of symmetric functions for which it is satisfied.

\begin{lemma}
\label{lem:h-constraint}
Condition \eqref{eq:h-constraint} is satisfied in the following cases:
\begin{enumerate}
\item $\ell$ is a positive semi-definite kernel function on $\mathbb{R}^k\times \mathbb{R}^k$;
\item $\ell(x,y) = -\|x-y\|_p^q$ for some $p\in (0,2]$ and $q\in (0, p]$ where $\|\cdot\|_p$ is the $p$-norm (or $p$-seminorm when $p < 1$) on $\mathbb{R}^k$.
\end{enumerate}
\end{lemma} 
The first claim of \prettyref{lem:h-constraint} is a direct consequence of the definition of positive semi-definite kernel function which ensures that the matrix $L$ itself is positive semi-definite and so is $G = JLJ$ since $J$ is also positive semi-definite. 
The second claim is a direct consequence of the Hilbert space embedding result by Schoenberg \cite{schoenberg1937certain,schoenberg1938metric}. See, for instance, Theorems 1 and 2 of \cite{schoenberg1937certain}.
%
%


\section{Two model fitting methods}
\label{sec:fitting}

This section presents two methods for fitting models \eqref{eq:projection} and \eqref{eq:kernel-model}--\eqref{eq:h-constraint}.
Both methods are motivated by minimizing the negative log-likelihood function of the inner-product model, and can be regarded as pseudo-likelihood approaches for more general models.
From a methodological viewpoint,
a key advantage of these methods, in particular the projected gradient descent method, is the scalability to networks of large sizes.


\subsection{A convex approach via penalized MLE}
\label{sec:convex}

In either the inner-product model or the general model
we suppose the parameter matrix $\Theta$ in \eqref{eq:Theta} or 
\eqref{eq:Theta-kernel}
satisfies
\begin{equation}
\label{eq:theta-constraint}
-M_1\leq \bTheta_{ij}\leq -M_2~~\text{for}~~1\leq i\neq j\leq n,\quad
\text{and} \quad |\bTheta_{ii}|\leq M_1 ~~\text{for}~~ 1\leq i\leq n.
\end{equation}
where 
$M_1\geq M_2$ are non-negative. 
Then 
for any $\Theta$ satisfying \eqref{eq:theta-constraint}, the corresponding edge probabilities satisfy 
\begin{equation}
	\label{eq:M1M2}
	\frac{1}{2}e^{-M_1}\leq \frac{1}{1 + e^{M_1}} \leq P_{ij} \leq \frac{1}{1 + e^{M_2}} \leq e^{-M_2}, ~ 1\leq i\neq j\leq n. 
\end{equation}
Thus $M_1$ controls the conditioning of the problem and $M_2$ controls the sparsity of the network.


Let $\sigma(x)={1}/{(1+e^{-x})}$ be the sigmoid function, i.e.~the inverse of logit function, then for any $i\neq j$, $\bP_{ij}=\sigma(\bTheta_{ij})$ and the log-likelihood function of the inner-product model \eqref{eq:projection} can be written as
\begin{equation*}
\begin{aligned}
\sum_{i<j} \left\{ \bA_{ij} \log\Big(\sigma(\bTheta_{ij})\Big)+(1-\bA_{ij})\log\Big(1-\sigma(\bTheta_{ij})\Big)\right\}
&= \sum_{i<j} \left\{ \bA_{ij}\bTheta_{ij}+\log\Big(1-\sigma(\bTheta_{ij})\Big)\right\}.
\end{aligned}
\end{equation*} 
Recall that $\bG=\bZ\bZ^\top$ in inner-product models. 
The MLE of $\bTheta^u$ is the solution to the following rank constrained optimization problem: 
\begin{equation}
\begin{aligned}
\min_{\bTheta^u, \balpha, \beta, \bG}~~ & -\sum_{i< j}\left\{ \bA_{ij}\bTheta_{ij}+\log\Big(1-\sigma(\bTheta_{ij})\Big)\right\},\\
\text{subject to }~~& \bTheta=\balpha\one^\top+\one\balpha^\top+\beta\bX+\bG,~~ -M_1\leq \bTheta_{ij}\leq -M_2,\\
& \bG\bJ=\bG,~~ \bG\in\mathbb{S}^{n}_+,~~ \text{rank}(\bG)\leq k.
\end{aligned}
\label{eq:non-convex}
\end{equation}
This optimization problem is non-convex and generally intractable.
To overcome this difficulty, we consider a convex relaxation that replaces the rank constraint on $G$ in \eqref{eq:non-convex} with a penalty term on its nuclear norm. 
Since $G$ is positive semi-definite, its nuclear norm equals its trace.
Thus, our first model fitting scheme solves the following convex program:
\begin{equation} 
\begin{aligned}
\min_{\balpha, \beta, \bG}~~  &
-\sum_{i, j}\left\{ \bA_{ij}\bTheta_{ij}+\log\Big(1-\sigma(\bTheta_{ij})\Big)\right\}
+\lambda_n \tr(\bG)\\
\text{subject to }~~  &\bTheta=\balpha\one^\top+\one\balpha^\top+\beta\bX+\bG, \bG\bJ=\bG, ~~ \bG\in\mathbb{S}^{n}_+, ~~ -M_1\leq \bTheta_{ij}\leq -M_2.
\end{aligned}
\label{eq:convex-obj}
\end{equation}

The convex model fitting method \eqref{eq:convex-obj} is motivated by the nuclear norm penalization idea originated from the matrix completion literature. See, for instance, \cite{candes2010power}, \cite{candes2012exact}, \cite{koltchinskii2011nuclear}, \cite{davenport20141}
and the references therein.
In particular, it can be viewed as a Lagrangian form of the proposal in \cite{davenport20141} when there is no covariate and the matrix is fully observed.
However, we have decided to make this proposal and study the theoretical properties as the existing literature, such as \cite{davenport20141}, does not take in consideration the potential presence of observed covariates.
Furthermore, one can still solve \eqref{eq:convex-obj} when the true underlying model is one of the general models introduced in \prettyref{sec:generalmodel}.
The appropriate choice of $\lambda_n$ will be discussed in \prettyref{sec:theory}.

\begin{remark}
In addition to the introduction of the trace penalty, the first term in the objective function in \eqref{eq:convex-obj} now sums over all $(i,j)$ pairs.
Due to symmetry, after scaling, the difference from the sum in \eqref{eq:non-convex} lies in the inclusion of all diagonal terms in $\Theta$.
This slight modification leads to neither theoretical consequence nor noticeable difference in practice. 
However, it allows easier implementation and simplifies the theoretical investigation.
We would note that the constraint $-M_1\leq \bTheta_{ij}\leq -M_2$ is included partially for obtaining theoretical guarantees. 
In simulated examples reported in \prettyref{sec:simu}, we found that the convex program worked equally well without this constraint.
\end{remark}


\subsection{A non-convex approach via projected gradient descent}
Although the foregoing convex relaxation method is conceptually neat, state-of-the-art algorithms to solve the nuclear (trace) norm minimization problem \eqref{eq:convex-obj}, such as iterative singular value thresholding,
usually require computing a full singular value decomposition at every iteration, which can still be time consuming when fitting very large networks.

To further improve scalability of model fitting, we propose an efficient first order algorithm that directly tackles the following non-convex objective function: 
\begin{equation}
\min_{Z, \balpha, \beta}~ g(Z, \alpha, \beta) = -\sum_{i, j}\left\{ \bA_{ij}\bTheta_{ij}+\log\Big(1-\sigma(\bTheta_{ij})\Big)\right\} ~~ \text{where } \bTheta=\balpha\one^\top+\one\balpha^\top+\beta\bX+ ZZ^\top.
\label{eq:nonconvex-obj}
\end{equation}
The detailed description of the method is presented in  Algorithm~\ref{alg:non-convex}.

\newcommand{\projz}{\mathcal{C}_{Z}}
\newcommand{\proja}{\mathcal{C}_{\alpha}}
\newcommand{\projb}{\mathcal{C}_{\beta}}
\begin{algorithm}[htb]
   \caption{A projected gradient descent model fitting method. 
   } 
   \label{alg:non-convex}
\begin{algorithmic}[1]
   \STATE {\bfseries Input:} Adjacency matrix: $\bA$; covariate matrix: $\bX$; latent space dimension: $k\geq 1$; initial estimates: $\bZ^0, \balpha^0, \beta^0$; step sizes: $\etaz, \etaa, \etab$; constraint sets: $\projz, \proja, \projb$.\\
   {\bfseries Output:} $\widehat{Z}=Z^T$, $\widehat{\alpha} = \alpha^T$, $\widehat{\beta} = \beta^T$.
  \FOR{$t=0, 1, \cdots, T-1$ 
  } 
    \STATE $\widetilde{\bZ}^{t+1} = \zt - \etaz\nabla_Z g(Z,\alpha, \beta) = \zt + 2\etaz \left(\bA-\sigt\right)\zt$;
    \STATE $\widetilde{\balpha}^{t+1} = \at - \etaa\nabla_\alpha g(Z,\alpha, \beta) = \at + 2\etaa (\bA-\sigt)\one$;
    \STATE $\widetilde{\beta}^{t+1} = \bt - \etab\nabla_\beta g(Z,\alpha, \beta)= \bt + \etab\inner{ \bA-\sigt, \bX}$;
    \STATE $\bZ^{t+1} = \mathcal{P}_{\projz}(\widetilde{\bZ}^{t+1}),~ \alpha^{t+1} = \mathcal{P}_{\proja}(\widetilde{\alpha}^{t+1}),~ \beta^{t+1} = \mathcal{P}_{\projb}(\widetilde{\beta}^{t+1})$;
  \ENDFOR
\end{algorithmic}
\end{algorithm}

After initialization, Algorithm \ref{alg:non-convex} iteratively updates the estimates for the three parameters, namely $Z$, $\alpha$ and $\beta$.
In each iteration, for each parameter, the algorithm first descends along the gradient direction by a pre-specified step size.
The descent step is then followed by an additional projection step which projects the updated estimates to pre-specified constraint sets.
We propose to set the step sizes as
\begin{equation}
\label{eq:stepsize}
\etaz=\eta/\op{Z^0}^2,~~ \etaa=\eta/(2n),\quad \text{and}\quad \etab=\eta/(2\fro{X}^2)
\end{equation}
for some small numeric constant $\eta > 0$. 
To establish the desired theoretical guarantees, 
we make a specific choice of the constraint sets later in the statement of \prettyref{thm:nonconvex} and \prettyref{thm:kernel-nonconvex-upper}.
In practice, one may simply set
\begin{equation}
	\label{eq:projection-practical}
	\bZ^{t+1} = J \widetilde{\bZ}^{t+1},\quad \alpha^{t+1} = \widetilde{\alpha}^{t+1},\quad \text{and}\quad \beta^{t+1} = \widetilde{\beta}^{t+1}.
\end{equation}
Here and after, when there is no covariate, i.e.~$X=0$, we skip the update of $\beta$ in each iteration.

For each iteration, the update on the latent part is performed in the space of $Z$ (that is $\mathbb{R}^{n\times k}$) rather than the space of all $n\times n$ Gram matrices as was required in the convex approach. 
In this way, it reduces the computational cost per iteration from $O(n^3)$ to $O(n^2 k)$.
Since we are most interested in cases where $k\ll n$, such a reduction leads to improved scalability of the non-convex approach to large networks. 
To implement this non-convex algorithm, we need to specify the latent space dimension $k$, which was not needed for the convex program \eqref{eq:convex-obj}. 
We defer the discussion on the data-driven choice of $k$ to \prettyref{sec:discuss}.

We note that Algorithm~\ref{alg:non-convex} is not guaranteed to find any global minimizer, or even any local minimizer, of the objective function \eqref{eq:nonconvex-obj}. 
However, as we shall show later in \prettyref{sec:theory}, under appropriate conditions, the iterates generated by the algorithm will quickly enter a neighborhood of the true parameters ($\tz, \ta, \tb$) and any element in this neighborhood is statistically at least as good as the estimator obtained from the convex method \eqref{eq:convex-obj}. 
This approach has close connection to the investigation of various non-convex methods for other statistical and signal processing applications. See for instance \cite{candes2015phase}, \cite{chen2015fast} and the references therein.
Our theoretical analysis of the algorithm is going to provide some additional insight as we shall establish its high probability error bounds for both the exact and the approximate low rank scenarios. 
In the rest of this section, 
we discuss initialization of Algorithm \ref{alg:non-convex}.


\subsubsection{Initialization}
\label{sec:initial}

Appropriate initialization is the key to success for Algorithm \ref{alg:non-convex}.
We now present two ways to initialize it which are theoretically justifiable under different circumstances.

\paragraph{Initialization by projected gradient descent in the lifted space}

The first initialization method is summarized in Algorithm~\ref{alg:convex-init}, which is essentially running the projected gradient descent algorithm on the following regularized objective function for a small number of steps:
\begin{equation}
f (G, \alpha, \beta)=-\sum_{ i, j}\left\{ \bA_{ij}\bTheta_{ij} + \log(1-\sigma(\bTheta_{ij}))\right\}  + \lambda_n \tr(G) + \frac{\gamma_n}{2} \big(\fro{G}^2 + 2\fro{\alpha\one^\top}^2 + \fro{X\beta}^2\big).
	\label{eq:f}
\end{equation}

Except for the third term, this is the same as the objective function in \eqref{eq:convex-obj}.
However, the inclusion of the additional proximal term ensures that one obtains the desired initializers after a minimal number of steps.


\newcommand{\projg}{\mathcal{C}_{G}}
\begin{algorithm}[htb]
   \caption{Initialization of Algorithm \ref{alg:non-convex} by Projected Gradient Descent} 
   \label{alg:convex-init}
\begin{algorithmic}[1]
   \STATE {\bfseries Input:} Adjacency matrix: $\bA$; covariate matrix $\bX$; initial values: $\bG^0 = 0, \balpha^0=0, \beta^0=0$; step size: $\eta$; constraint set: $\projg, \proja, \projb$; regularization parameter: $\lambda_n, \gamma_n$; latent dimension: $k$; number of steps: $T$.
  \FOR{$t=1, 2, \cdots, T$} 
    \STATE $\widetilde{\bG}^{t+1} = \gt - \eta \nabla_G f(G,\alpha, \beta) = \gt + \eta (\bA-\sigt-\lambda_n\bI_n - \gamma_n G^t )$;
    \STATE $\widetilde{\balpha}^{t+1} = \at - \eta\nabla_\alpha f(G,\alpha, \beta)/n = \at + \eta ((\bA-\sigt)\one/2n - \gamma_n\alpha^t )$;
    \STATE $\widetilde{\beta}^{t+1} = \bt - \eta\nabla_\beta f(G,\alpha, \beta)/\fnorm{X}^2= \bt + \eta (\inner{ \bA-\sigt, \bX}/\fnorm{X}^2 - \gamma_n\beta^t )$;
    \STATE $\bG^{t+1} = \mathcal{P}_{\projg}(\widetilde{\bG}^{t+1}),~ \alpha^{t+1} = \mathcal{P}_{\proja}(\widetilde{\alpha}^{t+1}),~ \beta^{t+1} = \mathcal{P}_{\projb}(\widetilde{\beta}^{t+1})$;
  \ENDFOR
  \STATE Set $Z^{T}=U_kD_k^{\half}$ where $U_kD_kU_k^\top$ is the top-$k$ eigen components of $G^{T}$;
    \STATE {\bfseries Output:} $Z^{T}, \alpha^{T}, \beta^{T}$.
\end{algorithmic}
\end{algorithm}

The appropriate choices of $\lambda_n$ and $\gamma_n$ will be spelled out in Theorem \ref{thm:init-range-kernel} and Corollary \ref{cor:init-algo-prob}.
The step size $\eta$ in Algorithm \ref{alg:convex-init} can be set at a small positive numeric constant, e.g.~$\eta = 0.2$.
The projection sets that lead to theoretical justification will be specified later in Theorem \ref{thm:init-range-kernel} while in practice, one may simply set 
$\bG^{t+1} = J\widetilde{\bG}^{t+1}$,
$\alpha^{t+1} = \widetilde{\alpha}^{t+1}$,
and
$\beta^{t+1} = \widetilde{\beta}^{t+1}$.

\paragraph{Initialization by universal singular value thresholding}
Another way to construct the initialization is to first estimate the probability matrix $P$ by universal singular value thresholding (USVT) proposed by \cite{chatterjee2015matrix} and then compute the initial estimates of $\alpha, Z, \beta$ heuristically by inverting the logit transform.
The procedure is summarized as Algorithm~\ref{alg:SVT-init}.

\begin{algorithm}[hbt]
   \caption{Initialization of Algorithm \ref{alg:non-convex} by Singular Value Thresholding} 
   \label{alg:SVT-init}
\begin{algorithmic}[1]
   \STATE {\bfseries Input:} Adjacency matrix: $\bA$; covariate matrix $\bX$; latent dimension $k$; threshold $\tau$. 
      \STATE Let $\wt{P} = \sum_{\sigma_i\geq \tau} \sigma_i u_iv_i^\top$ where $\sum_{i=1}^n \sigma_i u_i v_i^\top$ is the SVD of $A$. 
Elementwisely project $\wt{P}$ to the interval $[\frac{1}{2} e^{-M_1}, \frac{1}{2}]$ to obtain $\wh{P}$. 
Let $\hbTheta = \text{logit}( (\wh{P} + \wh{P}^\top)/2 )$;
      \STATE Let $\alpha^0, \beta^0 = \arg\min_{\alpha, \beta} \fnorm{ \hbTheta - \left(\alpha\one^\top + \one\alpha^\top + \beta X\right)}^2$;
      \STATE Let $\wh{G} = \mathcal{P}_{\mathbb{S}_+^n}(R)$ where $R =  J (\hbTheta - (\alpha^0\one^\top + \one(\alpha^0)^\top + \beta^0 X ) )J$;
      \STATE Set $Z^0=U_kD_k^{\half}$ where $U_kD_kU_k^\top$ is the top-$k$ singular value components of $\widehat{G}$;
    \STATE {\bfseries Output:} $\alpha^{0}, Z^{0}, \beta^0$.
\end{algorithmic}
\end{algorithm}

Intuitively speaking,
the estimate of $P$ by USVT is consistent when $\|P\|_*$ is ``small''. Following the arguments in Theorems 2.6 and 2.7 of \cite{chatterjee2015matrix}, such a condition is satisfied when the covariate matrix $X = 0$ or when $X$ has ``simple'' structure. 
Such ``simple'' structure could be $X_{ij} = \kappa(x_i, x_j)$ where $x_1, \cdots, x_n\in\mathbb{R}^d$ are feature vectors associated with the $n$ nodes and $\kappa(\cdot, \cdot)$ characterizes the distance/similarity between node $i$ and node $j$. 
For instance, one could have
$X_{ij} = \mathbf{1}_{\{x_i = x_j\}}$ where $x_1, \cdots, x_n\in \{1, \cdots, K\}$ is a categorical variable such as gender, race, nationality, etc;
or 
$X_{ij} = s(|x_i - x_j|)$ where $s(\cdot)$ is a continuous monotone link function and $x_1, \cdots, x_n\in\mathbb{R}$ is a continuous node covariate such as age, income, years of education, etc.
\begin{remark}
  The least squares problem in step 2 of Algorithm~\ref{alg:SVT-init} has closed form solution and can be computed in $O(n^2)$ operations. The computational cost of Algorithm~\ref{alg:SVT-init} is dominated by matrix decompositions in step 1 and step 3. 
\end{remark}


\section{Theoretical results}
\label{sec:theory}

In this section, we first present error bounds for both fitting methods under inner-product models, followed by their generalizations to the more general models satisfying the Schoenberg condition \eqref{eq:h-constraint}.
In addition, we give theoretical justifications of the two initialization methods for Algorithm \ref{alg:non-convex}.

\subsection{Error bounds for inner-product models}
We shall establish uniform high probability error bounds for inner-product models belonging to the following parameter space:
\begin{equation}
	\label{eq:para-nonconvex}
\begin{aligned}
\param=\Big\{\bTheta| & \bTheta=\balpha\one^\top+\one\balpha^\top+\beta\bX+\bZ\bZ^\top, ~\bJ\bZ=\bZ,\\  & \max\limits_{1\leq i\leq n}\|Z_{i*}\|^2 ,~ \|\balpha\|_{\infty}, ~ |\beta| \max\limits_{1\leq i < j\leq n} |X_{ij}| \leq \frac{M_1}{3},
\max\limits_{1\leq i \neq j\leq n} \Theta_{ij}\leq -M_2 \Big\}.
\end{aligned}
\end{equation}
When $X = 0$, we replace the first inequality in \eqref{eq:para-nonconvex} with $\max_{1\leq i\leq n}\|Z_{i*}\|^2 ,~ \|\balpha\|_{\infty}\leq M_1/2$.
For the results below, $k$, $M_1$, $M_2$ and $X$ are all allowed to change with $n$.

\paragraph{Results for the convex approach}


We first present theoretical guarantees for the optimizer of \eqref{eq:convex-obj}.
When $X$ in nonzero, we make the following assumption for the identifiability of $\beta$.
\begin{assumption}
The \emph{stable rank} of the covariate matrix $\bX$ satisfies $\srank(\bX)=\fro{X}^2/\op{X}^2\geq M_0 k$ for some large enough constant $M_0$. 
 \label{assump:srank}
\end{assumption}
The linear dependence on $k$ of $\srank(X)$ is in some sense necessary for $\beta$ to be identifiable as otherwise the effect of the covariates could be absorbed into the latent component $ZZ^\top$.

Let $(\hbalpha, \hbeta, \hbG)$ be the solution to \eqref{eq:convex-obj} and $(\ta, \tb, \tg)$ be the true parameter that governs the data generating process. 
Let $\hbTheta$ and $\tall$ be defined as in \eqref{eq:Theta} but with the estimates and the true parameter values for the components respectively.
Define the error terms $\dall=\hbTheta - \tall$, $\da=\hbalpha - \ta$, $\db=\hat\beta - \tb$ and $\dg=\hbG - \tg$.
The following theorem gives both deterministic and high probability error bounds for estimating both the latent vectors and logit-transformed probability matrix.


\begin{theorem}
Under Assumption~\ref{assump:srank}, for any $\lambda_n$ satisfying $\lambda_n\geq \lambdacond$,
there exists a constant $C$ such that
\begin{equation*}
\fro{\dg}^2, \fro{\dall}^2 \leq C e^{2M_1} \lambda_n^2 k.
\end{equation*}
Specifically, setting $\lambda_n=\lambdaP$ for a large enough positive constant $C_0$, there exist positive constants $c, C$ such that uniformly over $\calF(n,k,M_1, M_2, X)$, with probability at least $1-n^{-c}$,
\begin{equation*}
\fro{\dg}^2, \fro{\dall}^2 \leq C \psi_n^2
\end{equation*}
where 
\begin{equation}
	\label{eq:psi2}
	\psi_n^2 = e^{2M_1 } nk \times \max\big\{e^{-M_2}, {\log n \over n}\big\}.
\end{equation}
\label{thm:convex}
\end{theorem}

If we turn the error metrics in \prettyref{thm:convex} to mean squared errors, namely $\fnorm{\dg}^2/{n^2}$ and $\fnorm{\dall}^2/{n^2}$, then we obtain the familiar $k/n$ rate in low rank matrix estimation problems \cite{koltchinskii2011nuclear,agarwal2012noisy,davenport20141,ma2015volume}. 
In particular, the theorem can viewed as a complementary result to the result in \cite{davenport20141} in the case where there are observed covariate and the $1-$bit matrix is fully observed.
When $e^{-M_2} \geq \frac{\log n}{n}$, the sparsity of the network affects the rate through the multiplier $e^{-M_2}$. As the network gets sparser, the multiplier will be no smaller than $\log n \over n$.

\begin{remark}
Note that the choice of the penalty parameter $\lambda_n$ depends on $e^{-M_2}$ which by \eqref{eq:M1M2} controls the sparsity of the observed network.
In practice, we do not know this quantity and we propose to estimate $M_2$ with
$\widehat{M}_2 = -\mathrm{logit}(\sum_{ij}A_{ij}/n^2)$.	
\end{remark}





\paragraph{Results for the non-convex approach} 
A key step toward establishing the statistical properties of the outputs of Algorithm \ref{alg:non-convex} is to characterize the evolution of its iterates.
To start with, we introduce an error metric that is equivalent to $\fnorm{\dallt}^2 = \fnorm{\Theta^t - \Theta_\star}^2$ while at the same time is more convenient for establishing an inequality satisfied by all iterates.
Note that the latent vectors are only identifiable up to an orthogonal transformation of $\mathbb{R}^k$, for any $\bZ_1, \bZ_2\in \mathbb{R}^{n\times k}$, we define the distance measure
\begin{equation*}
\text{dist} (\bZ_1, \bZ_2)=\min_{\bR\in O(k)} \fro{\bZ_1-\bZ_2\bR}
\end{equation*}
where $O(k)$ collects all $k\times k$ orthogonal matrices. 
Let 
$\bR^t=\argmin_{\bR\in O(k)} \fro{\zt-\tz\bR}$
\mbox{and}
$\dzt=\zt-\tz\bR^t$,
and further let $\dat = \balpha^t - \balpha_\star, \dgt = Z^t(Z^t)^\top - \tz\tz^\top$ and $\dbt = \beta^t - \beta_\star$.
Then the error metric we use for characterizing the evolution of iterates is
\begin{equation}
e_t=\op{\tz}^2\fro{\dzt}^2+2\fro{\dat\one^\top}^2+\fro{\dbt X}^2.
\label{eq:metric}
\end{equation}
Let $\kappa_{\tz}$ be the condition number of $\tz$ (i.e., the ratio of the largest to the smallest singular values). The following lemma shows that the two error metrics $e_t$ and $\fnorm{\dallt}^2$ are equivalent up to a multiplicative factor of order $\kappa_{\tz}^2$. 
\begin{lemma} 
Under Assumption~\ref{assump:srank}, there exists constant $0\leq c_0<1$ such that
\begin{equation*}
\begin{aligned}
e_t &\leq \frac{\kappa^2_{\tz}}{2(\sqrt{2}-1)}\fro{\dgt}^2+2\fro{\dat\one^\top}^2+\fro{\dbt X}^2 \leq  \frac{\kappa^2_{\tz}}{2(\sqrt{2}-1)(1-c_0)}\fro{\dallt}^2.
\end{aligned}
\end{equation*}
Moreover, if $\text{dist}(\zt, \tz)\leq c\op{\tz}$, 
\begin{equation*}
\begin{aligned}
e_t &\geq \frac{1}{(c+2)^2}\fro{\dgt}^2+2\fro{\dat\one^\top}^2+\fro{\dbt X}^2\geq  \frac{1}{(c+2)^2(1+c_0)}\fro{\dallt}^2.
\end{aligned}
\end{equation*}
\label{lem:error-metric}
\end{lemma}



In addition, our error bounds depend on the following assumption on the initializers. 
\begin{assumption}The initializers $\bZ^0, \balpha^0, \beta^0$ in Algorithm~\ref{alg:non-convex} satisfy
$e_0\leq c e^{-2M_1}\op{\tz}^4/\kappa_{\tz}^4$ for a sufficiently small positive constant $c$.
\label{assump:init}
\end{assumption}
Note that the foregoing assumption is not very restrictive. If $k\ll n$, $M_1$ is a constant and the entries of $Z_\star$ are i.i.d.~random variables with mean zero and bounded variance, then
$\|Z_\star\|_{\mathrm{op}}\asymp \sqrt{n}$ and $\kappa_{Z_\star} \asymp 1$. 
In view of Lemma \ref{lem:error-metric}, this only requires $\frac{1}{n^2} \|{\Theta}^0-\Theta_\star\|_{\mathrm{F}}^2$ to be upper bounded by some constant.
We defer verification of this assumption for initial estimates constructed by Algorithm \ref{alg:convex-init} and Algorithm \ref{alg:SVT-init} to \prettyref{sec:init-theory}.

The following theorem states that under such an initialization, errors of the iterates converge linearly till they reach the same statistical precision $\psi_n^2$ as in \prettyref{thm:convex} modulo a multiplicative factor that depends only on the condition number of $Z_\star$.
\begin{theorem}
Let Assumptions \ref{assump:srank} and \ref{assump:init} be satisfied. 
Set the constraint sets as\footnote{When $X=0$, $\projb = \emptyset$ and we replace $M_1/3$ in $\projz$ and $\proja$ with $M_1/2$ in correspondence with the discussion following \eqref{eq:para-nonconvex}.}
\begin{equation*}
\begin{aligned}
\projz & = \{Z\in\mathbb{R}^{n\times k}, JZ = Z, \max_{1\leq i\leq n}\|Z_{i*}\|\leq M_1/3\}, \\
\proja & = \{\alpha\in\mathbb{R}^n, \|\alpha\|_{\infty}\leq M_1/3\}, ~\projb = \{\beta\in\mathbb{R}, \beta\|X\|_{\infty} \leq M_1/3\}.	
\end{aligned}
\end{equation*}
Set the step sizes as in \eqref{eq:stepsize} for any $\eta\leq c$ where $c>0$ is a universal constant. 
Let $\zeta_n = \lambdacond$. Then we have
\begin{enumerate}
  \item \emph{Deterministic errors of iterates:} If $\op{\tz}^2\geq C_1\kappa_{\tz}^2 e^{M_1} \zeta_n^2 \times \max\left\{\sqrt{\eta k e^{M_1}}, 1\right\}$ for a sufficiently large constant $C_1$,  there exist positive constants $\rho$ and $C$ such that 
\begin{equation*}
e_t\leq 2\bigg(1- \frac{\eta}{e^{M_1}\kappa_{\tz}^2}\rho \bigg)^t e_0 + \frac{C\kappa_{\tz}^2}{\rho} e^{2M_1} \zeta_n^2 k .
\end{equation*}
\item \emph{Probabilistic errors of iterates:} If $\op{\tz}^2\geq C_1\kappa_{\tz}^2\sqrt{n}e^{M_1-M_2/2} \max\left\{\sqrt{\eta k e^{M_1}}, 1\right\}$ for a sufficiently large constant $C_1$, there exist positive constants $\rho, c_0$ and $C$ such that uniformly over $\calF(n,k,M_1, M_2, X)$ with probability at least $1-n^{-c_0}$, 
\begin{equation*}
e_t\leq 2\bigg(1- \frac{\eta}{e^{M_1}\kappa_{\tz}^2}\rho \bigg)^t e_0 + 
C\frac{\kappa_{\tz}^2}{\rho}\psi_n^2.
\end{equation*}
For any $T > T_0 = \log  (\frac{M_1^2}{\kappa_{\tz}^2e^{4M_1 - M_2}}\frac{n}{k^2} ) / \log (1- \frac{\eta}{e^{M_1}\kappa_{\tz}^2}\rho)$, 
\[
\fnorm{\Delta_{G^T}}^2, ~\fnorm{\Delta_{\Theta^T}}^2 \leq  C^\prime\kappa_{\tz}^2 \psi_n^2.
\]
for some constant $C^\prime > 0$. 
\end{enumerate}
\label{thm:nonconvex}
\end{theorem}





\begin{remark}
In view of Lemma \ref{lem:error-metric}, the rate obtained by the non-convex approach in terms of $\fro{\dall}^2$ matches the upper bound achieved by the convex method, up to a multiplier of $\kappa_{\tz}^2$. 
As suggested by Lemma~\ref{lem:error-metric}, the extra factor comes partly from the fact that $e_t$ is a slightly stronger loss function than $\fro{\dallt}^2$ and in the worst case can be $c\kappa_{\tz}^2$ times larger than $\fro{\dallt}^2$. 
\end{remark}
\newcommand{\projzc}{\mathcal{C}^1_{Z}}
\newcommand{\projzd}{\mathcal{C}^2_{Z}}
\begin{remark}
Under the setting of Theorem~\ref{thm:nonconvex}, the projection steps for $\alpha, \beta$ in Algorithm~\ref{alg:non-convex} are straightforward and have the following closed form expressions:
$\balpha^{t+1}_i=\widetilde{\balpha}^{t+1}_i\min \{1, {M_1}/{(3|\widetilde{\balpha}^{t+1}_i|)} \}$, $\beta^{t+1}=\widetilde{\beta}^{t+1} \min\{ 1, {M_1}/{(3|\widetilde{\beta}^{t+1}|\max_{i,j}|\bX_{ij}|}) \}$. 
The projection step for $Z$ is slightly more involved. Notice that $\projz = \projzc\bigcap\projzd$ where 
$$\projzc = \{Z\in\mathbb{R}^{n\times k}, JZ = Z\}, ~\projzd = \{Z\in\mathbb{R}^{n\times k}, \max_{1\leq i\leq n}\|Z_{i*}\|^2\leq M_1/3\}.$$
Projecting to either of them has closed form solution, that is 
$\mathcal{P}_{\projzc} (Z) = JZ$, $[\mathcal{P}_{\projzd} (Z)]_{i*} = Z_{i*}\min \{ 1, \sqrt{{M_1}/{(3\|Z_{i*}\|^2)}} \}$.
Then Dykstra's projection algorithm \citep{dykstra1983algorithm} (or alternating projection algorithm) can be applied to obtain $\mathcal{P}_{\projz}(\widetilde{Z}^{t+1})$. 
We note that projections induced by the boundedness constraints for $Z, \alpha, \beta$ are needed for establishing the error bounds theoretically. 
However, when implementing the algorithm, users are at liberty to drop these projections and to only center the columns of the $Z$ iterates as in \eqref{eq:projection-practical}. 
We did not see any noticeable difference thus caused on simulated examples reported in \prettyref{sec:simu}.
\end{remark}

\begin{remark}
When both $M_1$ and $M_2$ are constants and the covariate matrix $X$ is absent, the result in Section 4.5 of \cite{chen2015fast}, in particular Corollary 5, implies the error rate of $O(nk)$ in \prettyref{thm:nonconvex}.
However, when $M_1\to\infty$ and $M_2$ remains bounded as $n\to\infty$, the error rate in \cite{chen2015fast} becomes\footnote{One can verify that in this case we can identify the quantities in Corollary 5 of \cite{chen2015fast} as $\sigma = 1$, $p=1$, $d=n$, $r=k$, $\nu\asymp M_1$, $L_{4\nu}\asymp 1$ and $\ell_{4\nu}\asymp e^{4M_1}$.} $O(e^{8M_1}M_1^2nk)$, which can be much larger than the rate $O(e^{2M_1}nk)$ in \prettyref{thm:nonconvex} even when $X$ is absent.
We feel that this is a byproduct of the pursuit of generality in \cite{chen2015fast} and so the analysis has not been fine-tuned for latent space models.
In addition, Algorithm \ref{alg:non-convex} enjoys nice theoretical guarantees on its performance even when the model is mis-specified and the $\Theta$ matrix is only approximately low rank. 
See \prettyref{thm:kernel-nonconvex-upper} below.
This case which is important from a modeling viewpoint was not considered in \cite{chen2015fast} as its focus was on generic non-convex estimation of low rank positive semi-definite matrices rather than fittings latent space models.
\end{remark}

\subsection{Error bounds for more general models}
\label{sec:general-result} 
We now investigate the performances of the fitting approaches on more general models satisfying the Schoenberg condition \eqref{eq:h-constraint}.
To this end, we consider the following parameter space for the more general class of latent space models
\begin{equation}
	\label{eq:para-general}
\begin{aligned}
\paramK=\Big\{\bTheta & | \bTheta=\balpha\one^\top+\one\balpha^\top+\beta\bX+G, 
G \in\mathbb{S}_+^n, JG = G,
\\ & \max\limits_{1\leq i\leq n} G_{ii},~\|\balpha\|_{\infty}, ~ |\beta| \max\limits_{1\leq i < j\leq n} |X_{ij}| \leq M_1/3, \max\limits_{1\leq i \neq j\leq n} \Theta_{ij}\leq -M_2 \Big\}.
\end{aligned}
\end{equation}
As before, when $X = 0$, we replace the first inequality in \eqref{eq:para-general} with $\max_{1\leq i\leq n}\|Z_{i*}\|^2 ,~ \|\balpha\|_{\infty}\leq M_1/2$.
For the results below, $M_1$, $M_2$ and $X$ are all allowed to change with $n$.
Note that the latent space dimension $k$ is no longer a parameter in \eqref{eq:para-general}.
Then for any positive integer $k$,
let $U_kD_kU_k^\top$ be the best rank-$k$ approximation to $\tg$.
In this case, with slight abuse of notation, we let 
\begin{equation}
	\label{eq:Zstar-general}
\tz=U_kD_k^{\half} \quad \mbox{and}\quad \gresid = \tg - U_kD_kU_k^\top.	
\end{equation}

Note that \eqref{eq:para-general} does not specify the spectral behavior of $G$ which will affect the performance of the fitting methods as the theorems in this section will later reveal.
We choose not to make such specification due to two reasons.
First, the spectral behavior of distance matrices resulting from different kernel functions and manifolds is by itself a very rich research topic.
See, for instance, \cite{mezard1999spectra,bogomolny2003spectral,el2010spectrum,cheng2013spectrum} and the references therein.
In addition, the high probability error bounds we are to develop in this section is going to work uniformly for all models in \eqref{eq:para-general} and can be specialized to any specific spectral decay pattern of $G$ of interest.


\paragraph{Results for the convex approach}
The following theorem is a generalization of Theorem \ref{thm:convex} to the general class \eqref{eq:para-general}.

\begin{theorem}
	For any $k\in\pint$ such that Assumption~\ref{assump:srank} holds and any $\lambda_n$ satisfying $\lambda_n \geq \lambdacond$,
	there exists a constant $C$ such that the solution to the convex program \eqref{eq:convex-obj} satisfies
\begin{equation*}
\fro{\dall}^2\leq C\left(e^{2M_1}\lambda_n^2k + e^{M_1}\lambda_n \|\gresid\|_*\right).
\end{equation*}
Specifically, setting $\lambda_n=\lambdaP$ for a large enough constant $C_0$, there exists positive constants $c, C$ such that uniformly over $\calF_g(n, M_1, M_2, X)$ with probability at least $1-n^{-c}$,
\begin{equation}
	\label{eq:convex-bd-general}
\fro{\dall}^2 \leq C ( \psi_n^2 + e^{M_1 - M_2/2}\sqrt{n}\|\gresid\|_* ).
\end{equation}
\label{thm:kernel-convex-upper}
\end{theorem}

The upper bound in \eqref{eq:convex-bd-general} has two terms.
The first is the same as that for the inner-product model.
The second can be understood as the effect of model mis-specification, since the estimator is essentially based on the log-likelihood of the inner-product model.
We note that the bound holds for any $k$ such that Assumption \ref{assump:srank} holds while the choice of the tuning parameter $\lambda_n$ does not depend on $k$. 
Therefore, we can take the infimum over all admissible values of $k$, depending on the stable rank of $X$.
When $X = 0$, we can further improve the bound on the right side of \eqref{eq:convex-bd-general} to be the infimum of the current expression over all $0\leq k\leq n$.

\paragraph{Results for the non-convex approach}
Under the definition in \eqref{eq:Zstar-general},
we continue to use the error metric $e_t$ defined in equation~\eqref{eq:metric}. 
The following theorem is a generalization of Theorem \ref{thm:nonconvex} to the general class \eqref{eq:para-general}.

\begin{theorem}
Under Assumptions~\ref{assump:srank} and \ref{assump:init}, set the constraint sets $\projz$, $\proja$, $\projb$ and the step sizes $\etaz$, $\etaa$ and $\etab$ as in Theorem \ref{thm:nonconvex}.
Let $\zeta_n = \lambdacond$. Then we have
\begin{enumerate}
\item \emph{Deterministic errors of iterates:} If $\op{\tg}\geq C_1\kappa_{\tz}^2 e^{M_1}\zeta_n^2  \times \max\left\{\sqrt{\eta k e^{M_1}}, \sqrt{\eta\fnorm{\gresid}^2/\zeta_n^2}, 1\right\}$, there exist positive constants $\rho$ and $C$ such that 
 \begin{equation*}
e_t\leq 2\bigg(1- \frac{\eta}{e^{M_1}\kappa_{\tz}^2}\rho \bigg)^t e_0 + \frac{C\kappa_{\tz}^2}{\rho} \left(e^{2M_1}\zeta_n^2 k + e^{M_1}\fro{\gresid}^2\right).
\end{equation*}
	\item \emph{Probabilistic errors of iterates:} If $\op{\tg}\geq C_1\kappa_{\tz}^2\sqrt{n}e^{M_1-M_2/2}  \max\left\{\sqrt{\eta k e^{M_1}}, \sqrt{\eta\fnorm{\gresid}^2/\zeta_n^2}, 1\right\}$ for a sufficiently large constant $C_1$, there exist positive constants $\rho, c_0$ and $C$ such that uniformly over $\calF_g(n,M_1,M_2,X)$ with probability at least $1-n^{-c_0}$, the iterates generated by Algorithm~\ref{alg:non-convex} satisfying
 \begin{equation*}
e_t\leq 2\bigg(1- \frac{\eta}{e^{M_1}\kappa_{\tz}^2}\rho \bigg)^t e_0 + \frac{C\kappa_{\tz}^2}{\rho} \left( \psi_n^2  +   e^{M_1}\fro{\gresid}^2 \right).
\end{equation*}
For any $T > T_0 = \log(\frac{M_1^2}{\kappa_{\tz}^2e^{4M_1 - M_2}}\frac{n}{k^2} ) / \log (1- \frac{\eta}{e^{M_1}\kappa_{\tz}^2}\rho)$, 
\[
\fnorm{\Delta_{G^T}}^2, ~\fnorm{\Delta_{\Theta^T}}^2 \leq  C^\prime \kappa_{\tz}^2 \left( \psi_n^2 +   e^{M_1}\fro{\gresid}^2 \right)
\]
for some constant $C^\prime$. 
\end{enumerate}
\label{thm:kernel-nonconvex-upper}
\end{theorem}

Compared with Theorem \ref{thm:nonconvex}, the upper bound here has an additional term $e^{M_1} \|\bar{G}_k\|_{\mathrm{F}}^2$ that can be understood as the effect of model mis-specification.
Such a term can result from mis-specifying the latent space dimension in Algorithm \ref{alg:non-convex} when the inner-product model holds, or it can arise when the inner-product model is not the true underlying data generating process.
This term is different from its counterpart in Theorem \ref{thm:kernel-convex-upper} which depends on $\bar{G}_k$ through its nuclear norm.
In either case, the foregoing theorem guarantees that Algorithm \ref{alg:non-convex} continues to yield reasonable estimate of $\Theta$ and $G$ as long as $\|\bar{G}_k\|_{\mathrm{F}}^2 = O(e^{M_1-M_2}nk)$, i.e., when the true underlying model can be reasonably approximated by an inner-product model with latent space dimension $k$.


\subsection{Error bounds for initialization}
\label{sec:init-theory}

We conclude this section with error bounds for the two initialization methods in \prettyref{sec:fitting}. 
These bounds justify that the methods yield initial estimates satisfying Assumption \ref{assump:init} under different circumstances.

\paragraph{Error bounds for Algorithm \ref{alg:convex-init}}
The following theorem indicates that Algorithm \ref{alg:convex-init} yields good initial estimates after a small number of iterates as long as 
the latent effect $\tg$ is substantial and the remainder $\bar{G}_k$ is well controlled.

\begin{theorem}
Suppose that Assumption~\ref{assump:srank} holds
and that $\fnorm{\ta\one^\top}, \fnorm{\tb X} \leq C\fnorm{\tg}$ for a numeric constant $C>0$.
Let $\lambda_n$ satisfy $\lambdaP\leq \lambda_n\leq c_0\op{\tg}/(e^{2M_1}\sqrt{k}\kappa_{\tz}^3)$ for sufficiently large constant $C_0$ and sufficiently small constant $c_0$, let $\gamma_n$ satisfy $\gamma_n \leq \delta\lambda_n/\op{\tg}$ for sufficiently small constant $\delta$. 
Choose step size $\eta\leq 2/9$ and set the
constraint sets as
\footnote{When $X=0$, $\projb = \emptyset$ and we replace $M_1/3$ in $\projg$ and $\proja$ with $M_1/2$ in correspondence with the discussion following \eqref{eq:para-nonconvex} and \eqref{eq:para-general}.}
\begin{align*}
\projg & = \{G\in\mathbb{S}_+^{n\times n}, JG = G, \max_{1\leq i, j\leq n}|G_{ij}|\leq M_1/3\}, \\
\proja & = \{\alpha\in\mathbb{R}^n, \|\alpha\|_{\infty}\leq M_1/3\}, ~\projb = \{\beta\in\mathbb{R}, \beta\|X\|_{\infty} \leq M_1/3\}.	
\end{align*}
If the latent vectors contain strong enough signal in the sense that 
\begin{equation}
  \op{\tg}^2\geq C\kappa_{\tz}^6e^{2M_1}\max \Big\{
  \phi_n^2 , ~\|\gresid\|_*^2/k , ~ \fro{\gresid}^2
  \Big\},
  \label{eq:Z-op2}
\end{equation}
for some sufficiently large constant $C$, there exist positive constants $c, C_1$ such that with probability at least $1-n^{-c}$, for any given constant $c_1>0$, $e_T\leq c_1^2e^{-2M_1}\op{\tz}^4/\kappa_{\tz}^4$  as long as $ T \geq T_0$, where
\begin{equation}
  T_0= \log \left(\frac{C_1e^{2M_1}k\kappa_{\tz}^6}{c_1^2}\right)\left(\log\left(\frac{1}{1-\gamma_n\eta}\right)\right)^{-1}.
\end{equation}
\label{thm:init-range-kernel}
\end{theorem} 

We note that the theorem holds for both inner-product models and more general models satisfying condition \eqref{eq:h-constraint}.
In addition, it gives the ranges of $\lambda_n$ and $\gamma_n$ for implementing Algorithm \ref{alg:convex-init}. Note that the choices $\lambda_n$ and $\gamma_n$ affect the number of iterations needed.
To go one step further,
the following corollary characterizes the ideal choices of $\gamma_n$ and $\lambda_n$ in Algorithm \ref{alg:convex-init}.
It is worth noting that the choice of $\lambda_n$ here does not coincide with 
that in Theorem \ref{thm:convex} and Theorem \ref{thm:kernel-convex-upper}.
So this is slightly different from the conventional wisdom that to initialize the non-convex approach, it would suffice to simply run the convex optimization algorithm for a small number of steps.
Interestingly, the corollary shows that when $M_1$, $k$ and $\kappa_{\tz}$ are all upper bounded by universal constants, for appropriate choices of $\gamma_n$ and $\lambda_n$ in Algorithm \ref{alg:convex-init}, the number of iterations needed does not depend on the graph size $n$.
\begin{corollary}
   Specifically in Theorem~\ref{thm:init-range-kernel}, if we choose $\gamma_n =  c_0/(e^{2M_1}\sqrt{k}\kappa_{\tz}^3)$ for some sufficiently small constant $c_0$, and $\lambda_n = C_0\gamma_n\op{\tg}$ for some sufficiently large constant $C_0$, there exist positive constants $c, C_1$ such that with probability at least $1-n^{-c}$, for any given constant $c_1>0$, $e_T\leq c_1^2e^{-2M_1}\op{\tz}^4/\kappa_{\tz}^4$  as long as $ T \geq T_0$, where
\begin{equation}
  T_0= \log \left(\frac{C_1e^{2M_1}k\kappa_{\tz}^6}{c_1^2}\right)\left(\log\left(\frac{1}{1-\gamma_n \eta}\right)\right)^{-1}.
\end{equation}
\label{cor:init-algo-prob}
\end{corollary} 

\begin{remark}
Similar to computing $\mathcal{P}_{\projz}(\cdot)$ in Algorithm \ref{alg:non-convex}, $\mathcal{P}_{\projg}(\cdot)$ could also be implemented by Dykstra's projection algorithm since $\projg$ is the intersection of two convex sets. The boundedness constraint $\max_{i, j}|\bG_{ij}|\leq M/3$ is only for the purpose of proof.  In practice, if ignoring this constraint, $\bG^{t+1}$ will have closed form solution $\bG^{t+1}=\mathcal{P}_{\mathbb{S}_+^n}(\bJ\widetilde{\bG}^{t+1}\bJ)$  where $\mathcal{P}_{\mathbb{S}_+^n}(\cdot)$ can be computed by singular value thresholding. 
\end{remark}


\paragraph{Error bounds for Algorithm \ref{alg:SVT-init}}
 
The following result justifies the singular value thresholding approach to initialization for inner-product models with no edge covariate.
\begin{proposition}
  If no covariates are included in the latent space model and $\fro{\tg}\geq c_0 n$ for some numeric constant $c_0>0$, then there exists constant $c_1$ such that with probability at least $1 - n^{c_1}$, for any $n \geq C(k, M_1, \kappa_{\tz})$ where $C(k, M_1, \kappa_{\tz})$ is a constant depending on $k, M_1$ and $\kappa_{\tz}$, the outputs of Algorithm~\ref{alg:SVT-init} with $\tau\geq 1.1\sqrt{n}$ satisfies the initialization condition in Assumption~\ref{assump:init}.
  \label{prop:svt-init}
\end{proposition}

Although we do not have further theoretical results, Algorithm \ref{alg:SVT-init} worked well on all the simulated data examples reported in \prettyref{sec:simu}.

\section{Simulation studies}
\label{sec:simu}

In this section, we present results of simulation studies on three different aspects of the proposed methods: 
(1) scaling of estimation errors with network sizes, 
(2) impact of initialization on Algorithm \ref{alg:non-convex}, and
(3) performance of the methods on general models.

\paragraph{Estimation errors}
We first investigate how estimation errors scale with network size. 
To this end, we fix $\tb = -\sqrt{2}$ and for any $(n, k) \in \left\{500, 1000, 2000, 4000, 8000\right\}\times \left\{2, 4, 8\right\}$, we set the other model parameters randomly following these steps:
\begin{enumerate}
	\setlength\itemsep{0em}
	\item Generate the degree heterogeneity parameters: $(\ta)_i = -\alpha_i / \sum_{j=1}^n\alpha_j$ for $1\leq i\leq n$, where $\alpha_1, \cdots, \alpha_n \stackrel{iid}{\sim}  U[1, 3]$.
	\item Generate $\mu_1,~ \mu_2\in\mathbb{R}^k$ with coordinates iid following $U[-1, 1]$ as two latent vector centers;
	\item Generate latent vectors: for $i=1,\dots, k$, let $(z_1)_i, \cdots, (z_{\lfloor n/2\rfloor})_i \stackrel{iid}{\sim} (\mu_1)_i + {N}_{[-2, 2]}(0, 1)$ and
	$(z_{\lfloor n/2\rfloor+1})_i, \cdots, (z_{n})_i \stackrel{iid}{\sim} (\mu_2)_i + {N}_{[-2, 2]}(0, 1)$ where ${N}_{[-2, 2]}(0, 1)$ is the standard normal distribution restricted onto the interval $[-2,2]$,
	then set $\tz = JZ$ where $Z=[z_1, \cdots, z_n]^\top$ and $J$ is as defined in \eqref{eq:J}. Finally, we normalize $\tz$ such that $\fro{\tg} = n$;
	\item Generate the covariate matrix: $X = n \wt{X}/\|\wt{X}\|_{\text{F}}$ where $\wt{X}_{ij} \stackrel{iid}{\sim} \min\left\{ |{N}(1, 1)|, 2\right\}$.
\end{enumerate}
For each generated model, we further generated $30$ independent copies of the adjacency matrix for each model configuration. Unless otherwise specified, for all experiments in this section, with given $(n, k)$, the model parameters are set randomly following the above four steps and algorithms are run on $30$ independent copies of the adjacency matrix.

The results of the estimation error for varying $(n, k)$ are summarized in the $\log$-$\log$ boxplots in Figure~\ref{fig:simu-err}, where ``Relative Error - $Z$'' is defined as $\|\wh{Z}\wh{Z}^\top - \tz\tz^\top\|_{\text{F}}^2/\|\tz\tz^\top\|_{\text{F}}^2$ and ``Relative Error - $\Theta$'' is defined as $\|\wh{\bTheta}- \tall\|_{\text{F}}^2/\|\tall\|_{\text{F}}^2$. 
From the boxplots, for each fixed latent space dimension $k$, the relative estimation errors for both $Z_\star$ and $\Theta_\star$ scale at the order of $1/\sqrt{n}$.
This agrees well with the theoretical results in \prettyref{sec:fitting}. For different latent space dimension $k$, the error curve (in log-log scale) with respect to network size $n$ only differs in the intercept. 

\begin{figure}[!h]
\centering
\includegraphics[width=0.9\textwidth]{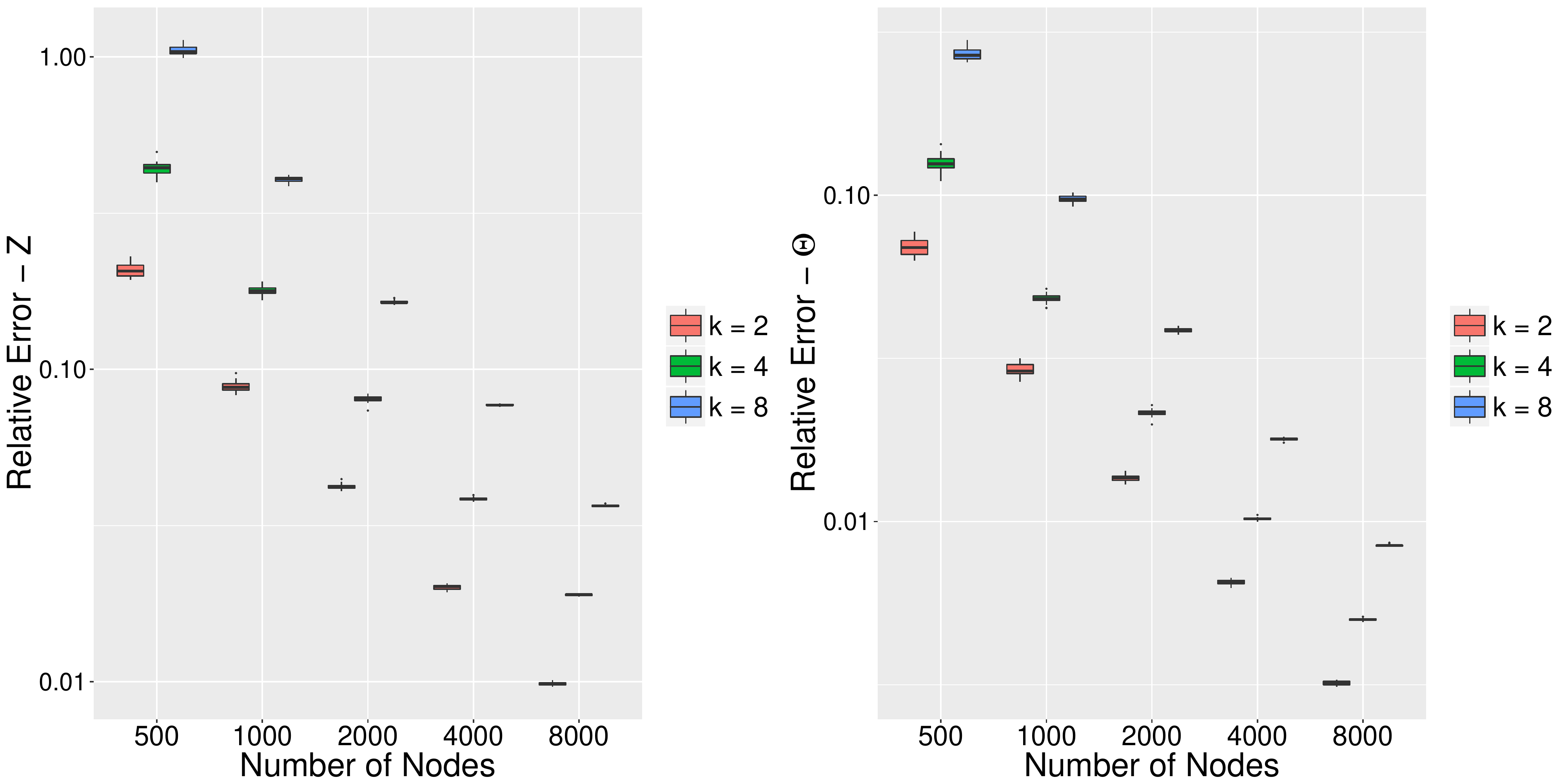}
\caption{$\log$-$\log$ boxplot for relative estimation errors with varying network size and latent space dimension. }
\label{fig:simu-err}
\end{figure}

%

\paragraph{Impact of initialization on Algorithm \ref{alg:non-convex}}
We now turn to the comparison of three different initialization methods for Algorithm \ref{alg:non-convex}: 
the convex method (Algorithm~\ref{alg:convex-init}), singular value thresholding (Algorithm \ref{alg:SVT-init}), and random initialization.
To this end, we fixed $n = 4000, k =4$. 
In Algorithm~\ref{alg:convex-init}, we choose $T = 10$ and $\lambda_n=2\sqrt{n\wh{p}}$ where $\wh{p}=\sum_{ij}A_{ij}/n^2$. In Algorithm~\ref{alg:SVT-init}, we set $M_1 = 4$ and threshold $\tau = \sqrt{n\wh{p}}$. 
The relative estimation errors are summarized as boxplots in Figure~\ref{fig:simu-init}. Clearly, the non-convex algorithm is very robust to the initial estimates. Similar phenomenon is observed in real data analysis where different initializations yield nearly the same clustering accuracy.  
\begin{figure}[!h]
\centering
\includegraphics[width=0.75\textwidth]{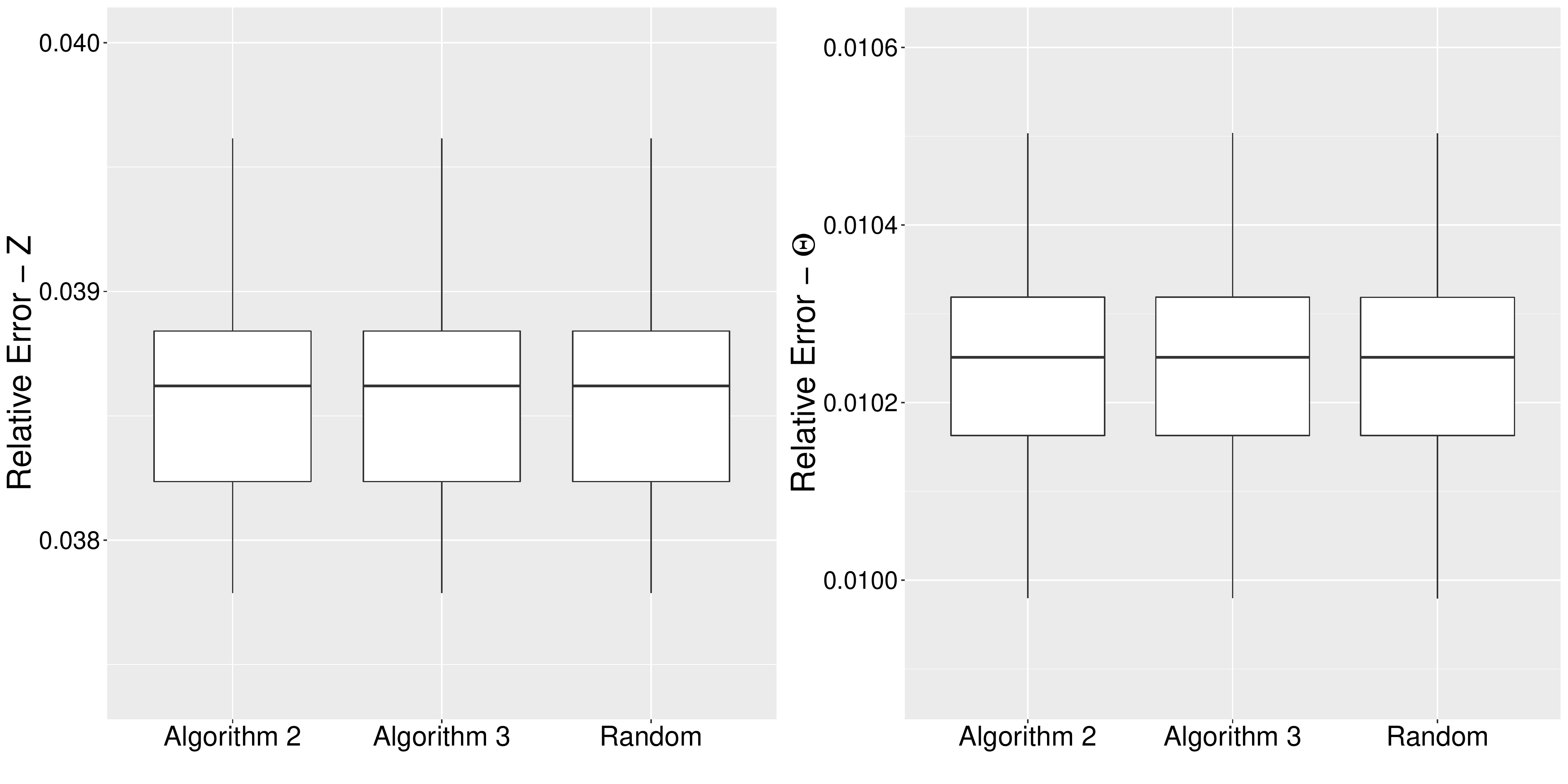}
\caption{Boxplot for relative estimation error with different initialization methods.}
\label{fig:simu-init}
\end{figure}

\paragraph{Performance on general models}
Finally, to investigate the performance of the proposed method under the general model~\eqref{eq:kernel-model}, we try two frequently used kernel functions, distance kernel $\ell_d(z_i, z_j) = -\|z_i - z_j\|$ and Gaussian kernel $\ell_g(z_i, z_j) = 4\exp(-\|z_i - z_j\|^2/9)$. In this part, we use $d$ to represent the dimension of the latent vectors (that is, $z_1, \cdots, z_n\in\mathbb{R}^d$) and $k$ to represent the fitting dimension in Algorithm~\ref{alg:non-convex}. We fix $d=4$ and network size $n = 4000$. Model parameters are set randomly in the same manner as the four step procedure except that the third step is changed to: 
\begin{quote}
	Generate latent vectors: for $i=1,\dots, d$, let $(z_1)_i, \cdots, (z_{\lfloor n/2\rfloor})_i \stackrel{iid}{\sim} (\mu_1)_i + {N}_{[-2, 2]}(0, 1)$ and
	$(z_{\lfloor n/2\rfloor+1})_i, \cdots, (z_{n})_i \stackrel{iid}{\sim} (\mu_2)_i + {N}_{[-2, 2]}(0, 1)$ where ${N}_{[-2, 2]}(0, 1)$ is the standard normal distribution restricted onto the interval $[-2,2]$. Finally for given kernel function $\ell(\cdot, \cdot)$, set $\tg = JLJ$ where $L_{ij} = \ell(z_i, z_j)$. 
\end{quote}
We run both the convex approach and Algorithm~\ref{alg:non-convex} with different fitting dimensions. The boxplot for the relative estimation errors and the singular value decay of the kernel matrix under distance kernel and Gaussian kernel are summarized in Figure~\ref{fig:simu-mis-distance} and Figure~\ref{fig:simu-mis-guassian} respectively. 

As we can see, under the generalized model, the non-convex algorithm exhibits bias-variance tradeoff with respect to the fitting dimension, which depends on the singular value decay of the kernel matrix. The advantage of the convex method is the adaptivity to the unknown kernel function.

\begin{figure}[!h]
\centering
\includegraphics[width=0.75\textwidth]{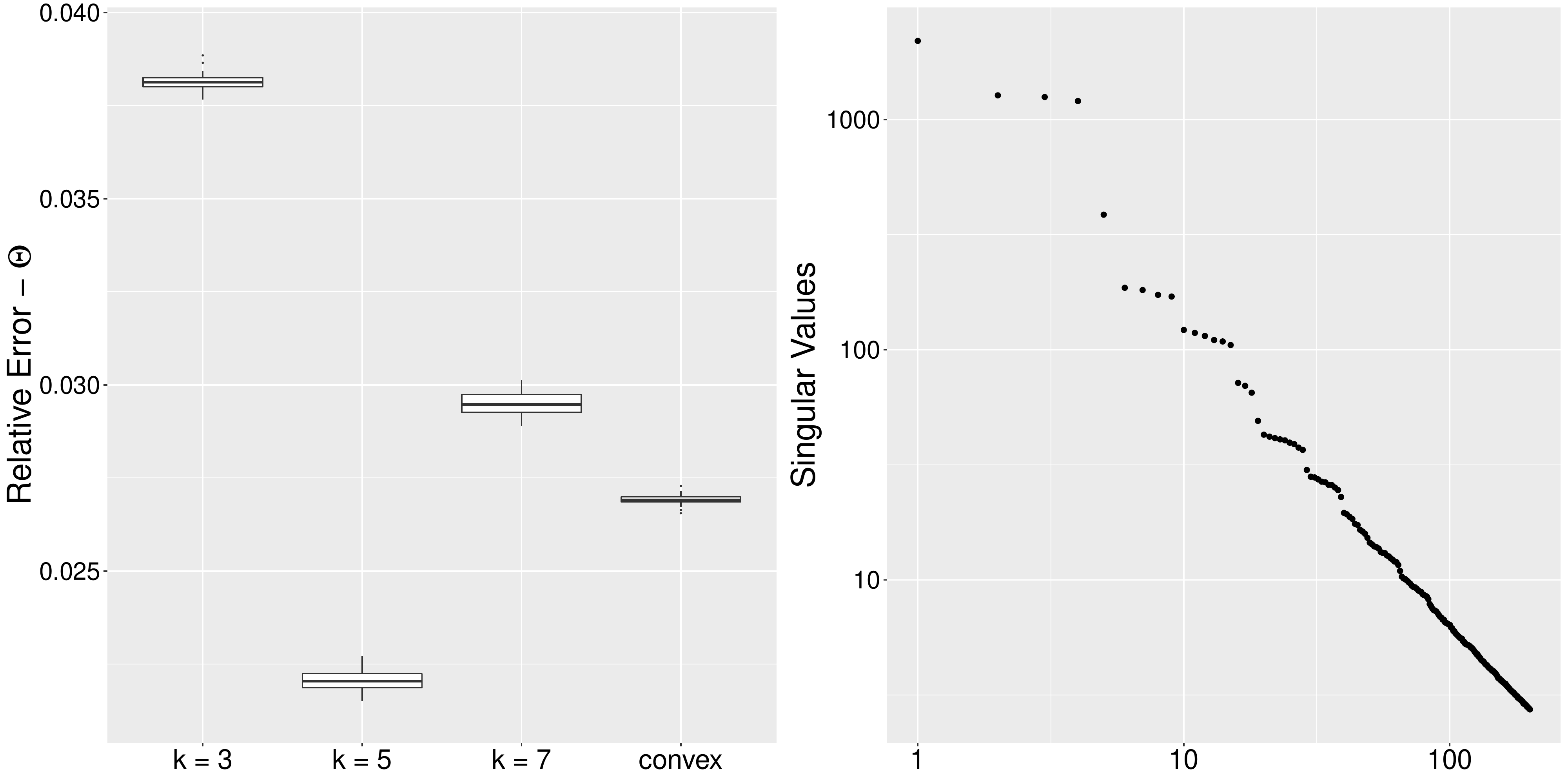}
\caption{The $\log$-$\log$ plot of relative estimation errors of both convex and non-convex approach under the distance kernel $\ell_d(z_i, z_j) = -\|z_i - z_j\|$ (left panel). The $\log$-$\log$ plot of ordered eigenvalues of $G_\star$ (right panel).}
\label{fig:simu-mis-distance}
\end{figure}

\begin{figure}[!h]
\centering
\includegraphics[width=0.75\textwidth]{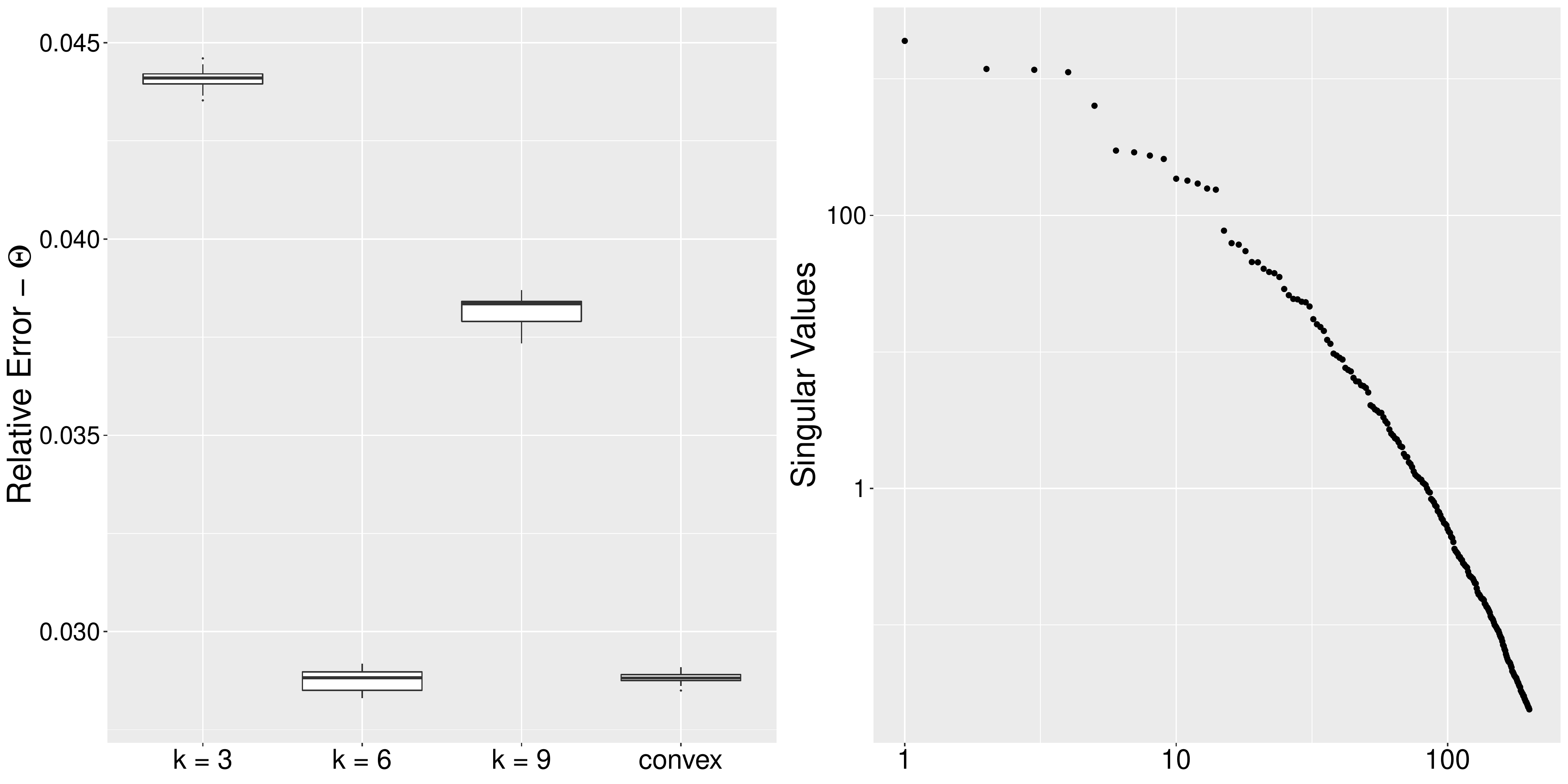}
\caption{$\log$-$\log$ plot for the relative estimation errors of both convex and non-convex approach under the Gaussian kernel $\ell_g(z_i, z_j) = 4\exp(-\|z_i - z_j\|^2/9)$ (left panel). The $\log$-$\log$ plot of ordered eigenvalues of $G_\star$ (right panel).}
\label{fig:simu-mis-guassian}
\end{figure}

When the true underlying model is not the inner-product model,
Theorem~\ref{thm:kernel-nonconvex-upper} indicates that the optimal choice of fitting dimension $k$ should depend on the size of the network. 
To illustrate such dependency, we vary both network size and fitting dimension, of which the results are summarized in Figure~\ref{fig:simu-mis-num_nodes}. As the size of the network increases, the optimal choice of fitting dimension increases as well. 
\begin{figure}[!h]
\centering
\includegraphics[width=0.8\textwidth]{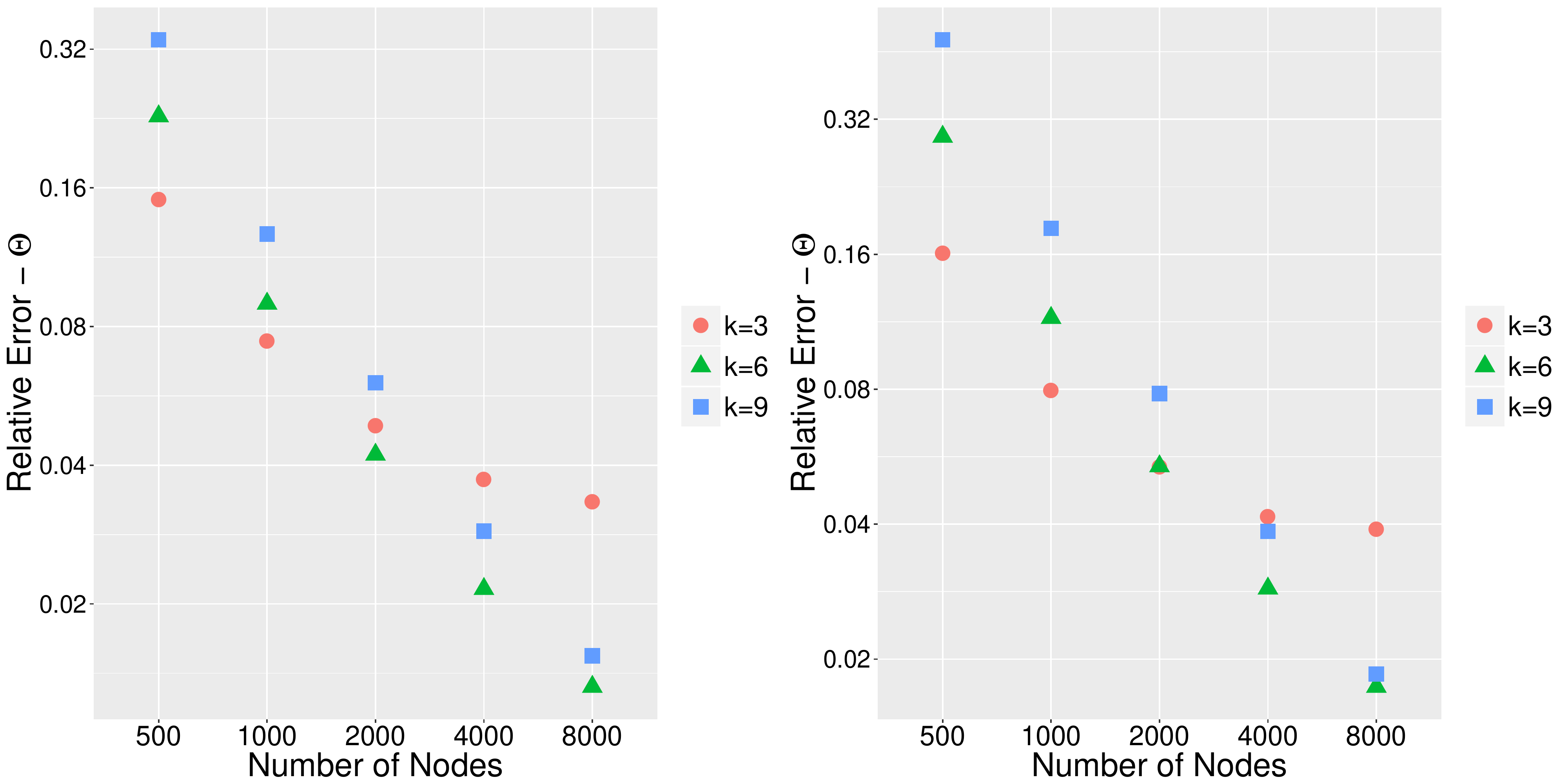}
\caption{$\log$-$\log$ plot for the relative estimation error with varying network size under distance kernel $\ell_d(z_i, z_j) = -\|z_i - z_j\|$ (left panel) and under the Gaussian kernel $\ell_g(z_i, z_j) = 4\exp(-\|z_i - z_j\|^2/9)$ (right panel).}
\label{fig:simu-mis-num_nodes}
\end{figure}


\section{Real data examples}
\label{sec:data}

In this section, we demonstrate how the model and fitting methods can be used to explore real world datasets that involve large networks.
In view of the discussion in \prettyref{sec:general-result} and \prettyref{sec:simu}, we can always employ the inner-product model \eqref{eq:projection} as the working model.
In particular, we illustrate three different aspects. 
First, we consider community detection on networks without covariate. 
To this end, we compare the performance of simple $k$-means clustering on fitted latent variables with several state-of-the-art methods.
Next, we investigate community detection on networks with covariates. 
In this case, we could still apply $k$-means clustering on fitted latent variables.
Whether there is covariate or not, we can always visualize the network by plotting fitted latent variables in some appropriate way.
Furthermore, we study how fitting the model can generate new feature variables to aid content-based classification of documents.  
The ability of feature generation also makes the model and the fitting methods potentially useful in other learning scenarios when additional network information among both training and test data is present.

\subsection{Community detection without covariate}
\newcommand{\lscd}{{LSCD}}

Community detection on networks without covariate has been intensively studied from both theoretical and methodological viewpoints.
Thus, it naturally serves as a test example for the effectiveness of the model and fitting methods we have proposed in previous sections. 
To adapt our method to community detection, we propose to partition the nodes by the following two step procedure:
\begin{enumerate}
\item Fit the inner-product model to data with Algorithm \ref{alg:non-convex};
\item Apply a simple $k$-means clustering on the fitted latent variables.
\end{enumerate}
In what follows, we call this two step procedure \lscd~(Latent Space based Community Detection). 
We shall compare it with three state-of-the-art methods:
(i) SCORE \citep{jin2015fast}: a normalized spectral clustering method developed under degree-corrected block models (DCBM);
(2) OCCAM \citep{zhang2014detecting}: a normalized and regularized spectral clustering method for potentially overlapping community detection;
(3) CMM \citep{chen2015convexified}: a convexified modularity maximization method developed under DCBM.
(4) Latentnet \cite{krivitskyfitting}: a hierachical bayesian method based on the latent space clustering model \citep{handcock2007model}. 

To avoid biasing toward our own method, we compare these methods on three datasets that have been previously used in the original papers to justify the first three methods at comparison:
a political blog dataset \citep{adamic2005political} that was studied in \cite{jin2015fast} and two Facebook datasets (friendship networks of Simmons College and Caltech) \citep{traud2012social} that were studied in \cite{chen2015convexified}. 
To make fair comparison, for all the methods, we supplied the true number of communities in each data. 
When fitting our model, we set the latent space dimension to be the same as the number of communities. 

In the latentnet package \cite{krivitskyfitting}, there are three different ways to predict the community membership. Using the notation of the R package \cite{krivitskyfitting}, they are \texttt{mkl\$Z.K}, \texttt{mkl\$mbc\$Z.K} and \texttt{mle\$Z.K}. We found that \texttt{mkl\$mbc\$Z.K} consistently outperformed the other two on these data examples and we thus used it as the outcome of Latentnet. 
Due to the stochastic nature of the Bayesian approach, we repeated it 20 times on each dataset and reported both the average errors as well as the standard deviations (numbers in parentheses).

Table \ref{table-error} summarizes the performance of all five methods on the three datasets. 
Among all the methods at comparison, all methods performed well on the political blog dataset with Latentnet being the best,  
and \lscd~outperformed all other methods on the two Facebook datasets.
On the Caltech friendship dataset, it improved the best result out of the other four methods by almost $15\%$ in terms of proportion of mis-clustered nodes.
 

\begin{table}[htb]
\begin{center}
\begin{tabular}{c|c|c|c|c|c|c}
\hline
Dataset& \# Clusters & \lscd &SCORE& OCCAM  & CMM & Latentnet \\
\hline
Political Blog & 2 & 4.746\% & 4.746\% & 5.319\% & 5.074\% & \textbf{4.513}\% (0.117\%)  \\

Simmons College&4 & \textbf{11.79\%}&23.57\%& 23.43\% &12.04\%  & 29.09\% (1.226\%)  \\

Caltech &8& \textbf{17.97\%} &31.02\%& 32.03\%& 21.02\%  & 38.47\% (1.190\%) \\
\hline
\end{tabular}
\end{center}
\vspace{-0.15in} 
\caption{A summary on proportions of mis-clustered nodes by different methods on three datasets. 
}
\label{table-error}
\end{table}
In what follows, we provide more details on each dataset and on the performance of these community detection methods on them.
\paragraph{Political Blog}
This well-known dataset was recorded by \cite{adamic2005political} during the 2004 U.S. Presidential Election. 
The original form is a directed network of hyperlinks between 1490 political blogs. 
The blogs were manually labeled as either liberal or conservative according to their political leanings.
The labels were treated as true community memberships. 
Following the literature, we removed the direction information and focused on the largest connected component which contains 1222 nodes and 16714 edges. 
All five methods performed comparably on this dataset with Latentnet achieving the smallest mis-clustered proportion.

\paragraph{Simmons College}
The Simmons College Facebook network is an undirected graph that contains 1518 nodes and 32988 undirected edges. 
For comparison purpose, we followed the same pre-processing steps as in  \cite{chen2015convexified} by considering the largest connected component of the students with graduation year between 2006 and 2009, which led to a subgraph of 1137 nodes and 24257 edges. 
It was observed in \cite{traud2012social} that the class year has the highest assortativity values among all available demographic characteristics, and so we treated the class year as the true community label.  
On this dataset, {\lscd} achieved the lowest mis-clustered proportion among these methods, with CMM a close second lowest.

An important advantage of model \eqref{eq:projection} is that it can provide a natural visualization of the network.
To illustrate,
the left panel of Figure~\ref{fig:simmons} is a 3D visualization of the network with the first three coordinates of the estimated latent variables.
From the plot, one can immediately see three big clusters: class year 2006 and 2007 combined (red), class year 2008 (green) and class year 2009 (blue). 
The right panel zooms into the cluster that includes class year 2006 and 2007 by projecting the the estimated four dimensional latent vectors onto a two dimensional discriminant subspace that was estimated from the fitted latent variables and the clustering results of \lscd. 
It turned out that class year 2006 and 2007 could also be reasonably distinguished by the latent vectors. 
\begin{figure}[!bth]
 \centering
\begin{minipage}{0.5\textwidth}
\includegraphics[width=\textwidth, height=6cm]{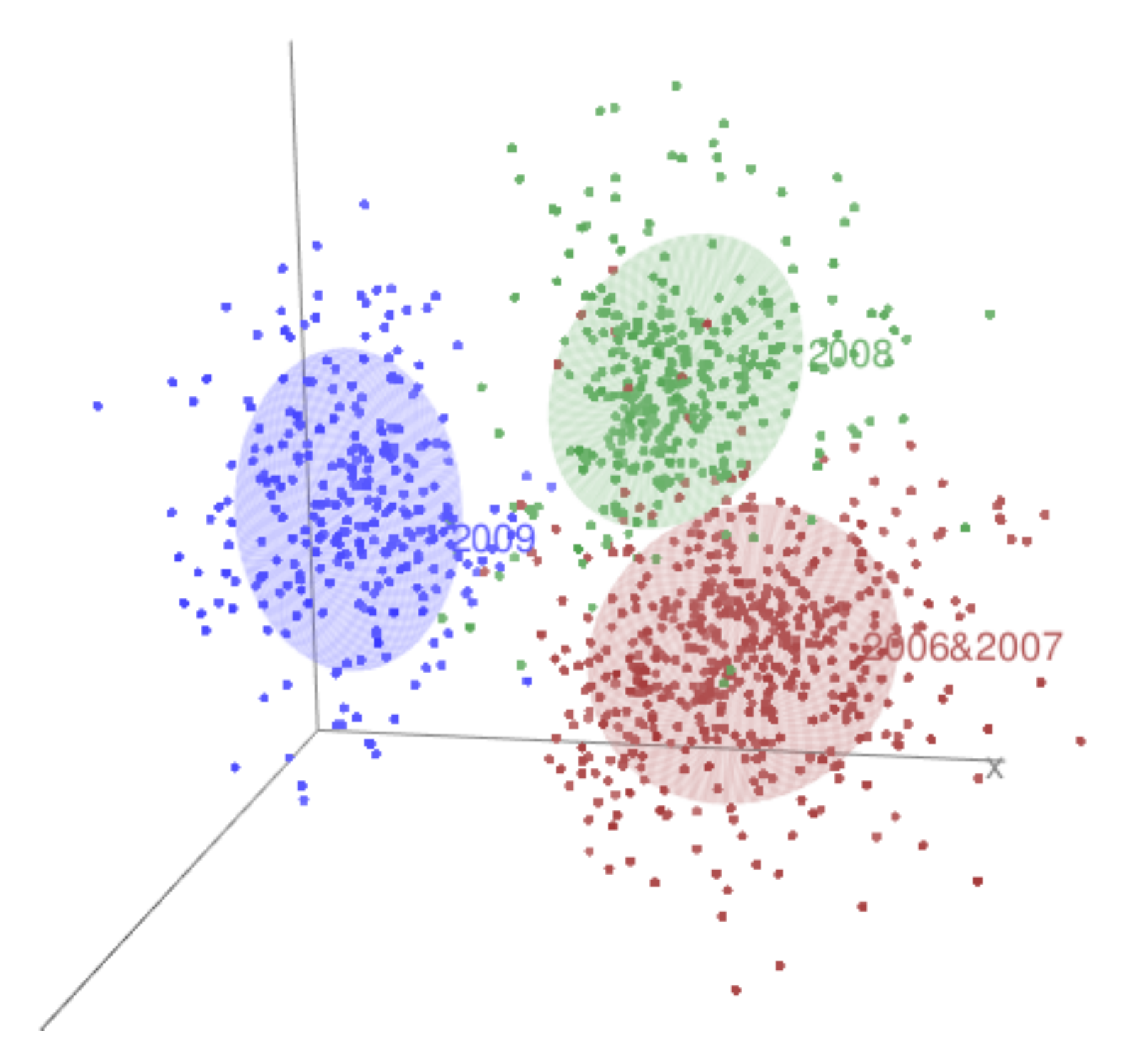}
\end{minipage}
\hfill
\begin{minipage}{0.45\textwidth}
\includegraphics[width=\textwidth, height=6cm]{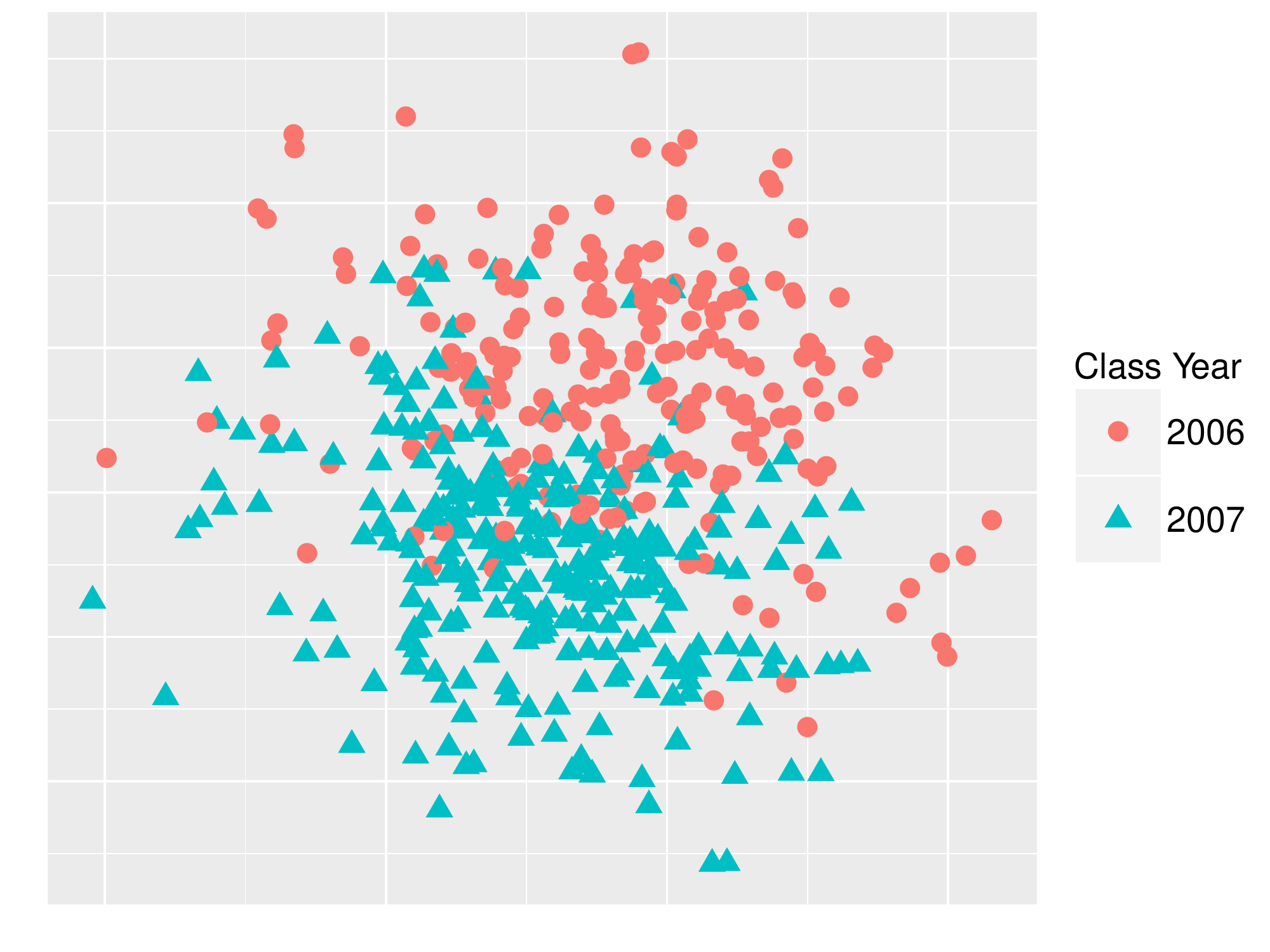}
\end{minipage}
\caption{The left panel is a visualization of the network with the first three coordinates of the estimated latent vectors. The right panel is a visualization of students in class year 2006 and 2007 by projecting the four dimensional latent vectors to an estimated two dimensional discriminant subspace. 
\label{fig:simmons}}
\end{figure}

\paragraph{Caltech Data}
In contrast to the Simmons College network in which communities are formed according to class years, 
communities in the Caltech friendship network are formed according to dorms \cite{traud2011comparing,traud2012social}. 
In particular, students spread across eight different dorms which we treated as true community labels. 
Following the same pre-processing steps as in \cite{chen2015convexified}, we excluded the students whose residence information was missing and considered the largest connected component of the remaining graph, which contained 590 nodes and 12822 undirected edges. 
This dataset is more challenging than the Simmons College network. 
Not only the size of the network halves but the number of communities doubles. 
In some sense, it serves the purpose of testing these methods when the signal is weak. 
\lscd~also achieved the highest overall accuracy on this dataset, reducing the second best error rate (achieved by CMM) by nearly 15\%. 
See the last row of Table \ref{table-error}. 
Moreover, \lscd~achieved the lowest maximum community-wise misclustering error among the four methods.
See Figure \ref{fig:caltech} for a detailed comparison of community-wise misclustering rates of the five methods.

\smallskip

It is worth noting that the two spectral methods, SCORE and OCCAM, 
fell behind on the two Facebook datasets. 
One possible explanation is that the structures of these Facebook networks are more complex than the political blog network and so DCBM suffers more under-fitting on them. 
In contrast, the latent space model \eqref{eq:projection} is more expressive and goes well beyond simple block structure. 
The Latentnet approach did not perform well on the Facebook datasets, either.
One possible reason is the increased numbers of communities compared to the political blog dataset, which substantially increased the difficulty of sampling from posterior distributions.


\begin{figure}[tbh]
\centering
\includegraphics[width=0.9\textwidth]{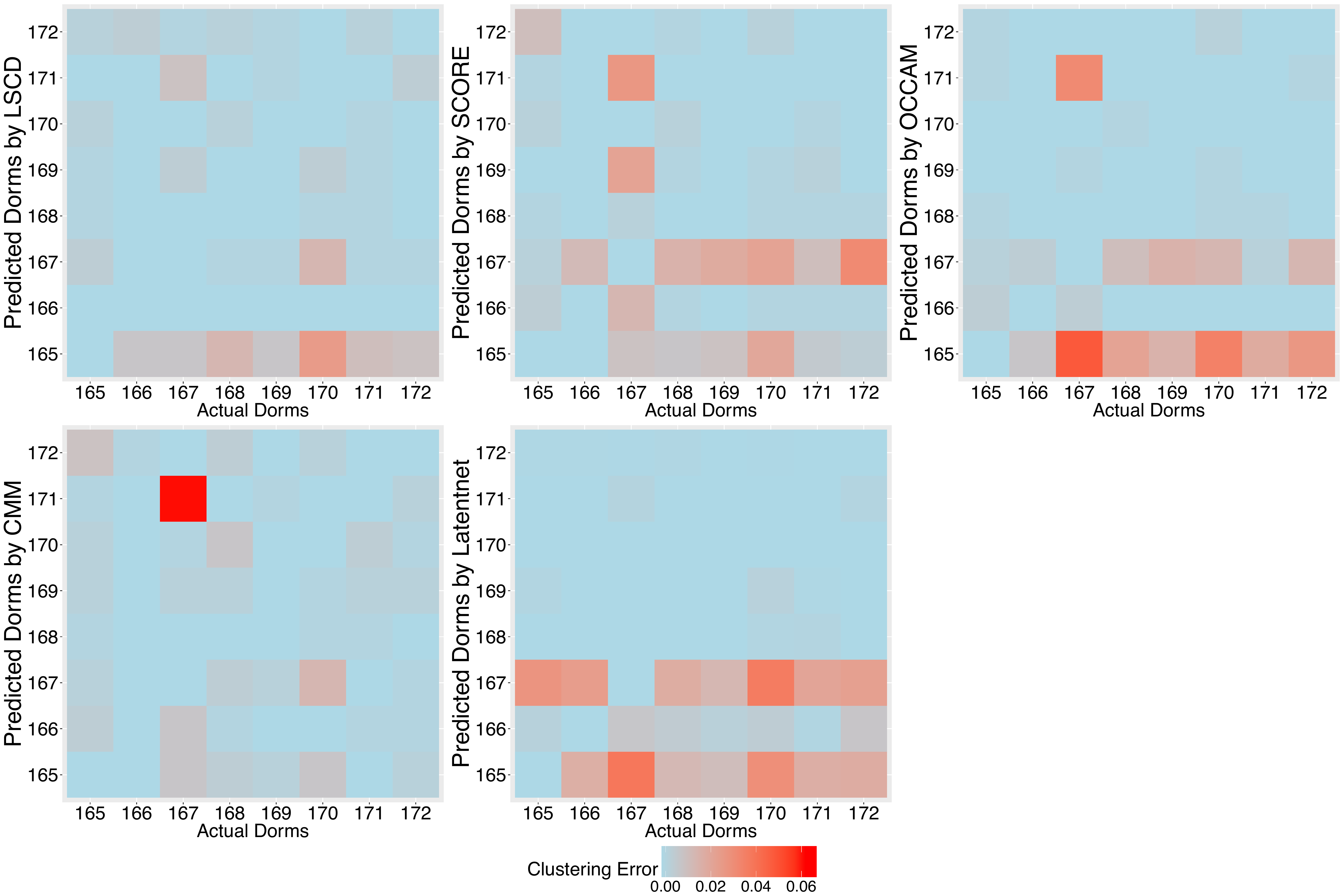}
\caption{Comparison of community-wise misclustering errors in Caltech friendship network. 
Top row, left to right: \lscd, SCORE and OCCAM;
bottom row, left to right: CMM and Latentnet.\label{fig:caltech}
}
\end{figure}

\subsection{Community detection with covariate}
\label{sec:data-covariate}

We now further demonstrate the power of the model and our proposed fitting methods by considering community detection on networks with covariates.
Again, we used the \lscd~procedure laid out in the previous subsection for community detection.

To this end, we consider a lawyer network dataset which was introduced in 
\cite{lazega2001collegial} that studied the relations among 71 lawyers in a New England law firm.
The lawyers were asked to check the names of those who they socialized with outside work, who they knew their family and vice versa. 
There are also several node attributes contained in the dataset: status (partner or associate), gender, office, years in the firm, age, practice (litigation or corporate), and law school attended, among which status is most assortative. 
Following \cite{zhang2015community}, we took status as the true community label. 
Furthermore, we symmetrized the adjacency matrix, excluded two isolated nodes and finally ended up with 69 lawyers connected by 399 undirected edges. 

Visualization and clustering results with and without covariate are shown in Figure~\ref{fig:lawyer}.
On the left panel, as we can see, the latent vectors without adjustment by any covariate worked reasonably well in separating the lawyers of different status and most of the 12 errors (red diamonds) were committed on the boundary. 
On the right panel, we included a covariate `practice' into the latent space model: we set $X_{ij}=X_{ji}=1$ if $i\neq j$ and the $i$th and the $j$th lawyers shared the same practice, and $X_{ij}=X_{ji} = 0$ otherwise.
Ideally, the influence on the network of being the same type of lawyer should be `ruled out' this way and the remaining influence on connecting probabilities should mainly be the effect of having different status. 
In other words, the estimated latent vectors should mainly contain the information of lawyers' status and the effect of lawyers' practice type should be absorbed into the factor $\beta X$. 
The predicted community memberships of lawyers indexed by orange numbers (39, 43, 45, 46, 51, 58) were successfully corrected after introducing this covariate. 
So the number of mis-clustered nodes was reduced by 50\%.
We also observed that lawyer 37, though still mis-clustered, was significantly pushed towards the right cluster. 

\begin{figure}[!h]
\centering
\includegraphics[width=0.85\textwidth]{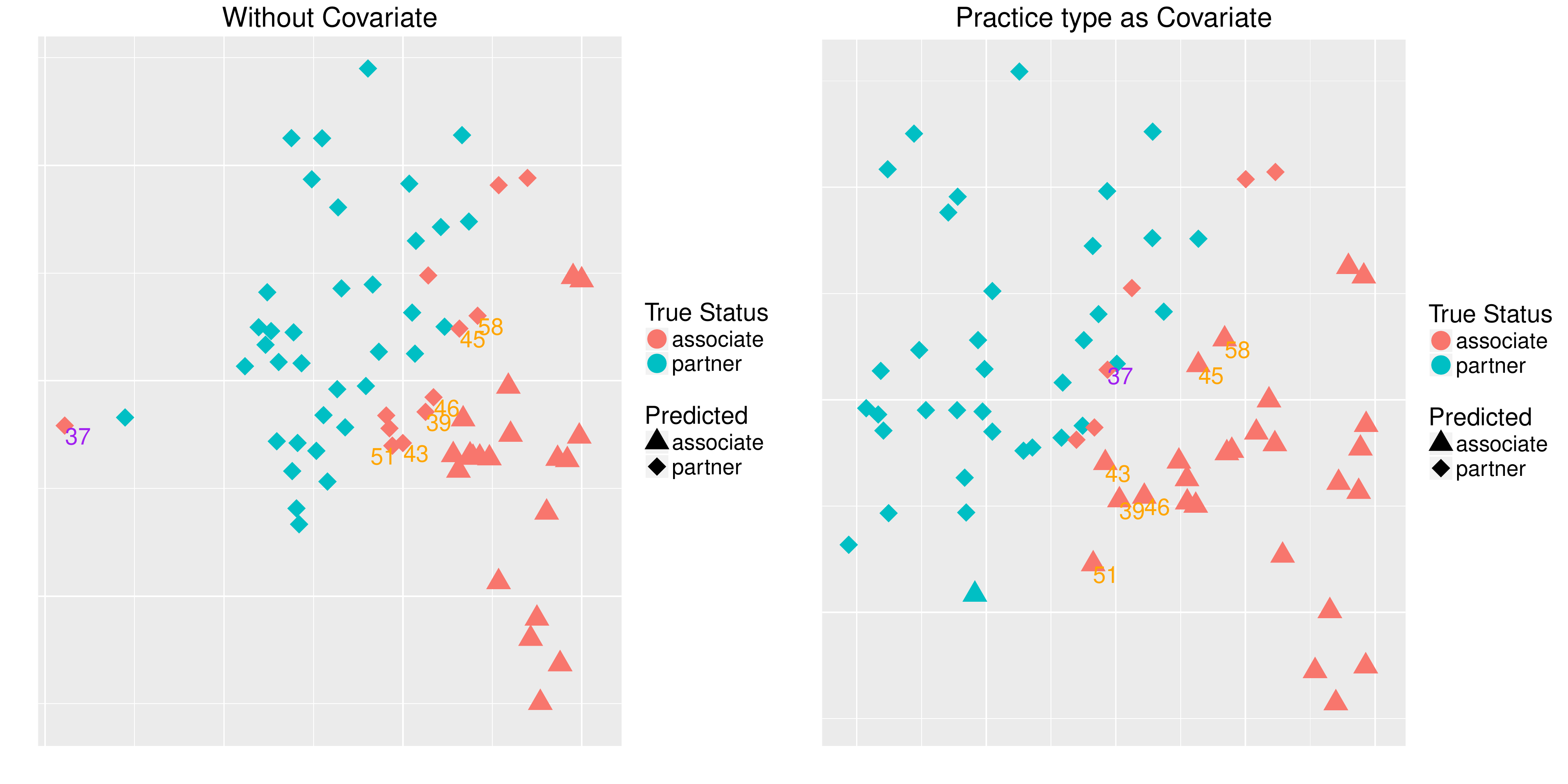}
\caption{Visualization of the lawyer network using the estimated two dimensional latent vectors. 
The left panel shows results without including any covariate while the right panel shows results that used practice type information.\label{fig:lawyer}
}
\end{figure}

\subsection{Network assisted learning}
In this section, we demonstrate how fitting model \eqref{eq:projection} can generate new features to be used in machine learning applications when additional network information is available. 
Consider a network with $n$ nodes and observed adjacency matrix $A$. 
Suppose the profile of the nodes is represented by $d$ dimensional features, denoted by $x_1, \cdots, x_n\in\mathbb{R}^d$. 
Assume each node is associated with a label (or say, variable of interest), denoted by $y$, either continuous or categorical. 
Suppose the labels are only observed for a subset of the nodes in the network. Without loss of generality, we assume $y_1, \cdots, y_m$ are observed for some $m < n$. 
The goal here is to predict the labels $y_{m+1}, \cdots, y_n$ based on the available information. 
Without considering the network information, this is the typical setup of supervised learning with labeled training set $(x_1, y_1), \cdots, (x_m, y_m)$ and unlabeled test set $x_{m+1}, \cdots, x_{n}$. As one way to utilize the network information, we propose to supplement the existing features in the prediction task with the latent vectors estimated by Algorithm~\ref{alg:non-convex} (without any edge covariates). 

To give a specific example, we considered the Cora dataset \citep{mccallum2000automating}.
It contains 2708 machine learning papers which were manually classified into 7 categories: Neural Networks, Rule Learning, Reinforcement Learning, Probabilistic Methods, Theory, Genetic Algorithms and Case Based. 
The dataset also includes the contents of the papers and a citation network, which are represented by a document-word matrix (the vocabulary contains 1433 frequent words) and an adjacency matrix respectively. 
The task is to predict the category of the papers based on the available information. For demonstration purpose, we only distinguish neural network papers from the other categories, and so the label $y$ is binary.

Let $W$ be the document-word matrix. In the present example, $W$ is of size $2708\times 1433$ ($2708$ papers and $1433$ frequent words). 
An entry $W_{ij}$ equals $1$ if the $i$th document contains the $j$th word. 
Otherwise $W_{ij}$ equals zero. 
As a common practice in latent semantic analysis, to represent the text information as vectors, we extract leading-$d$ principal component loadings from $WW^\top$ as the features. We chose $d = 100$ by maximizing the prediction accuracy using cross-validation. 

However, how to utilize the information contained in the citation network for the desired learning problem is less straightforward.
We propose to augment the latent semantic features with the latent vectors estimated from the citation network. 
Based on the simple intuition that papers in the same category are more likely to cite each other, we expect the latent vectors, as low dimensional summary of the network, to contain information about the paper category. 
The key message we want to convey here is that with vector representation of the nodes obtained from fitting the latent space model, 
network information can be incorporated in many supervised and unsupervised learning problems and other exploratory data analysis tasks.

Back to the Cora dataset, for illustration purpose, 
we fitted standard logistic regressions with the following three sets of features:
\begin{enumerate}
	\item the leading 100 principal component loadings;
    \item estimated degree parameters $\hat\alpha_i$ and latent vectors $\hat{z}_i$ obtained from Algorithm \ref{alg:non-convex};
    \item the combination of features in 1 and 2.
\end{enumerate}
We considered three different latent space dimensions: $k = 2, 5, 10$. 
As we can see from Figure \ref{figure:cora}, the latent vectors contained a considerable amount of predictive power for the label. 
Adding the latent vectors to the principal components of the word-document matrix could substantially reduce misclassification rate. 

\begin{figure}[!tbh]
	\centering
\includegraphics[width=0.55\textwidth]{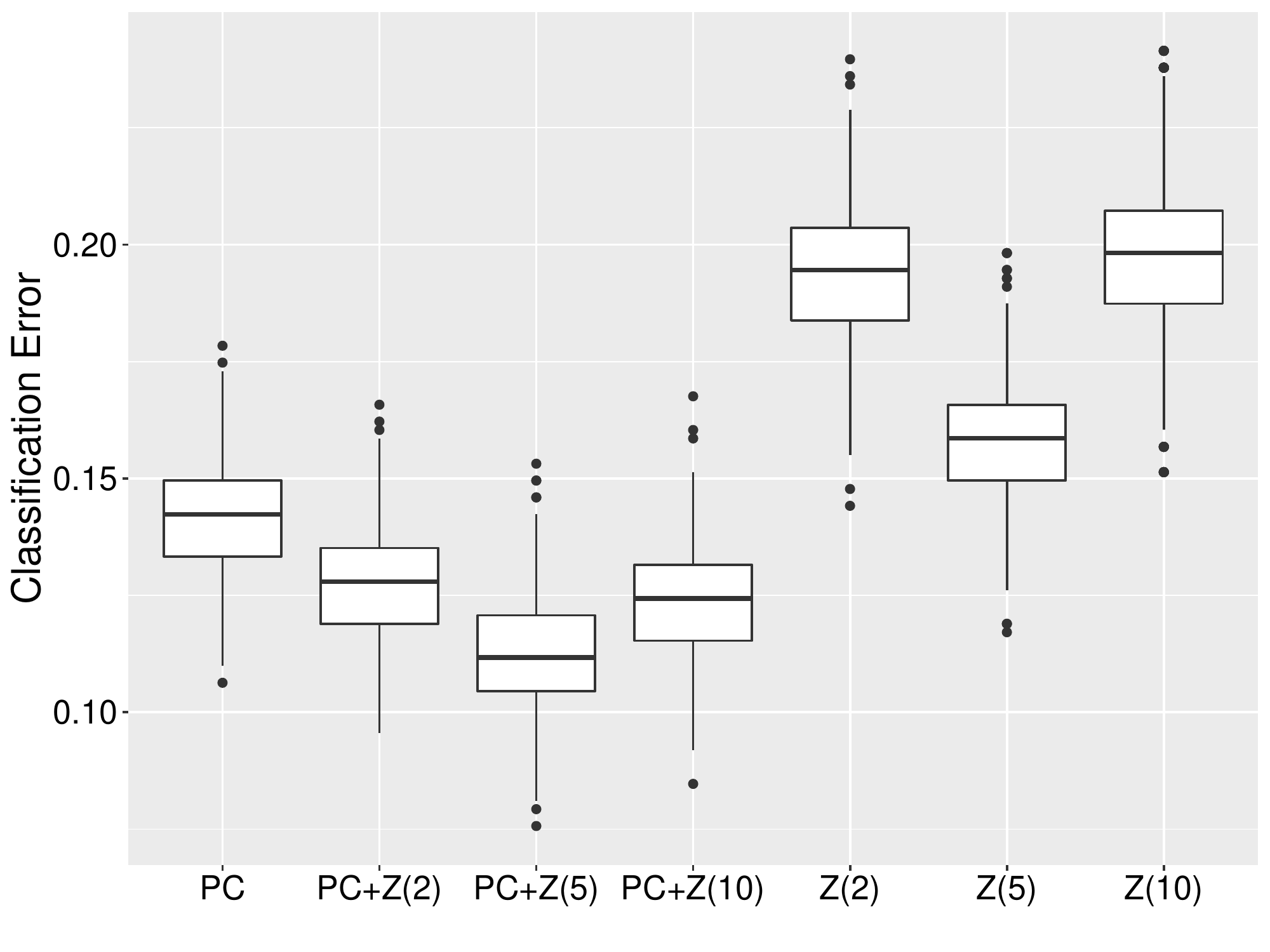}
\caption{Boxplots of misclassification rates using logistic regression with different feature sets. 
We randomly split the dataset into training and test sets with size ratio 3:1 for 500 times and computed misclassification errors for each configuration. 
PC represents the leading 100 principal component loadings of the document-word matrix. $Z(k)$ represents the feature matrix where the $i$th row is the concatenation of the estimated degree parameter $\wh{\alpha}_i$ and the estimated latent vector $\wh{z}_i$ with latent dimension $k$. PC+$Z(k)$ means the combination of the two sets of features.}
\label{figure:cora}
\end{figure}



\section{Discussion}
\label{sec:discuss}

In this section, we discuss a number of related issues and potential problems for future research.

%

\paragraph{Data-driven choice of latent space dimension}
For the projected gradient descent method, i.e., Algorithm \ref{alg:non-convex}, one needs to specify the latent space dimension $k$ as an input.
Although \prettyref{thm:kernel-nonconvex-upper} suggests that the algorithm could still work reasonably well if the specified latent space dimension is slightly off the target, it is desirable to have a systematic approach to selecting $k$ based on data.
One possibility is to inspect the eigenvalues of $G^T$ in Algorithm \ref{alg:convex-init} and set $k$ to be the number of eigenvalues larger than the parameter $\lambda_n$ used in the algorithm.

%
\paragraph{Undirected networks with multiple covariates and weighted edges}
The model \eqref{eq:projection} and the fitting methods can easily be extended to handle multiple covariates.
When the number of covariates is fixed, error bounds analogous to those in Section \ref{sec:theory} can be expected.
We omit the details.
Moreover, as pointed out in \cite{goldenberg2010survey}, latent space models for binary networks such as \eqref{eq:projection} can readily be generalized to weighted networks, i.e., networks with non-binary edges.
We refer interested readers to the general recipe spelled out in Section 3.9 of \cite{goldenberg2010survey}.
If the latent variables enter a model for weighted networks in the same way as in model \eqref{eq:projection}, we expect the key ideas behind our proposed fitting methods to continue to work.

%

\paragraph{Directed networks} 
In many real world networks, edges are directed. 
Thus, it is a natural next step to generalize model \eqref{eq:projection} to handle such data.
Suppose for any $i\neq j$, $A_{ij} = 1$ if there is an edge pointing from node $i$ to node $j$, and $A_{ij} = 0$ otherwise.
We can consider the following model: for any $i\neq j$, 
\begin{equation}
\bA_{ij} \stackrel{ind.}{\sim} \text{Bernoulli} (\bP_{ij}),\quad \text{with}\quad
\text{logit}(\bP_{ij})=\bTheta_{ij}=\balpha_i+\gamma_j+\beta\bX_{ij}+\bz_i^\top
w_j.
\label{eq:model-directed} 
\end{equation}
Here, the $\alpha_i$'s $\in \mathbb{R}$ model degree heterogeneity of outgoing edges while the $\gamma_j$'s $\in \mathbb{R}$ model heterogeneity of incoming edges. 
The meaning of $\beta$ is the same as in model \eqref{eq:projection}.
To further accommodate asymmetry, we associate with each node two latent vectors $z_i, w_i\in \mathbb{R}^k$, where the $z_i$'s are latent variables influencing outgoing edges and the $w_i$'s incoming edges. 
Such a model has been proposed and used in the study of recommender system 
\cite{agarwal2009regression} and it is also closely connected to the latent eigenmodel proposed in \cite{hoff2008modeling} if one further restricts $z_i\in \{w_i, -w_i\}$ for each $i$.
Under this model, the idea behind the convex fitting method in Section \ref{sec:convex} can be extended.
However, it is more challenging to devise a non-convex fitting method with similar theoretical guarantees to what we have in the undirected case.
On the other hand, it should be relatively straightforward to further extend the ideas to directed networks with multiple covariates and weighted edges.
A recent paper \cite{levina2017} has appeared after the initial posting of the present manuscript, which obtained some interesting results along these directions.

%
%
%






\section{Proofs of main theorems}
\label{sec:proof}
We present here the proofs of Theorem \ref{thm:kernel-convex-upper} and Theorem \ref{thm:kernel-nonconvex-upper} since Theorem \ref{thm:convex} is a corollary of the former and Theorem \ref{thm:nonconvex} a corollary of the latter.
Throughout the proof, let $P = (\sigma(\Theta_{\star,ij}))$ and $\Poff = (P_{ij}\mathbf{1}_{i\neq j})$. 
Thus, $\mathbb{E}(A) = \Poff$.
Moreover, for any $\Theta\in \mathbb{R}^{n\times n}$, define 
\begin{equation}
	\label{eq:h}
\wl (\bTheta)=-\sum_{i, j=1}^n\{ \bA_{ij}\bTheta_{ij} + \log(1-\sigma(\bTheta_{ij}))\}. 
\end{equation}
For conciseness, we denote $\calF_g(n, M_1, M_2, X)$ by $\calF_g$ throughout the proof.
We focus on the case where $X$ is nonzero, and the case of $X=0$ is simpler.

\subsection{Proof of Theorem \ref{thm:kernel-convex-upper}}
Let $\tz\in \mathbb{R}^{n\times k}$ such that $\tz\tz^\top$ is the best rank $k$ approximation to $\tg$.
For any matrix $M$, let $\col(M)$ be the subspace spanned by the column vectors of $M$ and $\row(M) = \col(M^\top)$. 
For any subspace $\mathcal{S}$ of $\mathbb{R}^n$ (or $\mathbb{R}^{n\times n}$), let $\mathcal{S}^\perp$ be its orthogonal complement, and $\calP_{\mathcal{S}}$ the projection operator onto the subspace. The proof relies on the following two lemmas.

\begin{lemma}
	\label{lem:kernel-cone}
Let $\mathcal{M}^\perp_k=\{\bM\in\mathbb{R}^{n\times n}: \row(\bM)\subset \col(\tz)^{\perp}\,\,\,\mathrm{and}\,\,\, \col(\bM)\subset \col(\tz)^{\perp}\}$ 
and $\mathcal{M}_k$ be its orthogonal complement in $\mathbb{R}^{n\times n}$ under trace inner product. 
If $\lambda_n\geq 2\op{A-P}$, then for $\gresid = \calP_{\mathcal{M}_k^\perp} \tg$, we have
\begin{equation*}
\|\dg \|_* \leq 4\sqrt{2k}\fnorm{\mathcal{P}_{\mathcal{M}_k}\dg} + 2\fnorm{\da\one^\top}+ \frac{2}{\lambda_n}| \langle \bA-\bP, \db\bX \rangle| + 4\| \gresid\|_*.
\end{equation*}
\end{lemma}
\begin{proof}
See \prettyref{sec:proof-43-lem1}.
\end{proof}

\begin{lemma} 
\label{lem:kernel-identifiability}
For any $k\geq 1$ such that Assumption~\ref{assump:srank} holds. Choose $\lambda_n\geq \max\left\{2\opnorm{\bA-\bP}, 1\right\}$ and $| \langle \bA-\bP, \bX \rangle | \leq \lambda_n \sqrt{k}\fnorm{X}$. 
There exist constants $C>0$ and $0\leq c<1$ such that
\begin{equation*}
\begin{aligned}
\fnorm{\dall}^2 & \geq (1 - c)
\big(\fnorm{\dg}^2 + 2\fnorm{\da\one^\top}^2+\fnorm{\db\bX}^2 \big) - C\|\gresid\|_*^2/k,\quad\mathrm{and} \\
\fnorm{\dall}^2 & \leq (1 + c) 
\big(\fnorm{\dg}^2 + 2\fnorm{\da\one^\top}^2+\fnorm{\db\bX}^2 \big) + C\|\gresid\|_*^2/k.
\end{aligned}
\end{equation*} 
\end{lemma}
\begin{proof}
See \prettyref{sec:proof-43-lem2}.
\end{proof}

 \begin{lemma}
There exist absolute constants $c, C$ such that for any $\Theta\in\calF_g$ with probability at least $1-n^{-c}$, the following inequality holds
 	\begin{equation*}
 \begin{aligned}
 \opnorm{A - P},~\frac{\inner{\bA-\bP, \bX}}{\sqrt{k}\fnorm{X}}\leq C\sqrt{\max\left\{ne^{-M_2}, \log n\right\}}.
 \end{aligned}
 \end{equation*}
 \label{lem:prob}
 \end{lemma}
 \begin{proof}
 For any $\Theta$ in the parameter space, 
 the off diagonal elements of $\Theta$ are uniformly bounded from above by $-M_2$, 
 and so $\max_{i,j}\Poff_{ij}\leq e^{-M_2}$. 
 Moreover, $\max_{i} P_{ii}\leq 1$ under our assumption.
 Thus, $\opnorm{A-P} \leq \opnorm{A-\Poff} + \opnorm{\Poff-P} \leq \opnorm{A-\Poff} + 1$.
 Together with Lemma~\ref{lem:concen-A}, this implies that there exist absolute constants $c_1, C >0$ such that uniformly over the parameter space
 \begin{equation}
 	\P \left(\opnorm{A - P}\leq C\sqrt{\max\left\{ne^{-M_2}, \log n\right\}}\right) \geq 1-n^{-c_1}.
 	\label{eq:lem-prob-part1}
 \end{equation}
 Since the diagonal entries of $X$ are all zeros, we have $\langle A-P, X\rangle = \langle A-\Poff, X\rangle$. 
 Hence, Lemma~\ref{lemma-AX} implies that uniformly over the parameter space,
 \begin{equation}
 \begin{aligned}
 \P \left(\frac{\inner{\bA-\bP, \bX}}{\sqrt{k}\fnorm{X}}\leq C\sqrt{\max\left\{ne^{-M_2}, \log n\right\}}\right) &\geq 1-3\exp \left(-C^2\max\left\{ne^{-M_2}, \log n\right\} k/8\right)\\
 &\geq 1 - 3n^{-C^2k/8}.
 \end{aligned}
 	\label{eq:lem-prob-part2}
 \end{equation}
 Combining~\eqref{eq:lem-prob-part1} and \eqref{eq:lem-prob-part1} finishes the proof.
 \end{proof}

\paragraph{Proof of Theorem~\ref{thm:kernel-convex-upper}}
$1^\circ$ We first establish the deterministic bound.
Observe that $\hbTheta=\hbalpha\one^\top+\one\hbalpha^\top+\hbeta\bX+\hbG$ is the optimal solution to \eqref{eq:convex-obj}, and that the true parameter $\tall=\ta\one^\top+\one\ta^\top+\tb\bX+\tg$ is feasible. 
Thus, we have the basic inequality
\begin{equation}
\wl(\hbTheta)-\wl(\tall)+\lambda_n (\|\hbG\|_*-\|\tg\|_* )\leq 0,
\label{eq:basic}
\end{equation}
where $h$ is defined in \eqref{eq:h}.
For any $\bTheta \in \calF_g$, $|\Theta_{ij}|\leq M_1$ for all $i,j$ and so for $\tau=e^{M_1}/(1+e^{M_1})^2$, the Hessian
\begin{equation*}
\nabla^2 \wl(\bTheta)=\mathrm{diag}\big(\mathrm{vec}\big(\sigma(\bTheta) \circ \left(1-\sigma(\bTheta)\right)\big)\big)\succeq \tau\bI_{n^2\times n^2}.	
\end{equation*}
For any vector $b$, $\mathrm{diag}(b)$ is the diagonal matrix with elements of $a$ on its diagonals.
For any matrix $B = [b_1, \dots, b_n]\in \mathbb{R}^{n\times n}$, $\mathrm{vec}(B)\in \mathbb{R}^{n^2}$ is obtained by stacking $b_1, \dots, b_n$ in order.
For any square matrices $A$ and $B$, $A\succeq B$ if and only if $A-B$ is positive semi-definite.
With the last display, Taylor expansion gives
\begin{equation*}
\wl(\hbTheta)-\wl(\tall)\geq 
\langle \nabla_{\bTheta} \wl(\tall),\dall \rangle +\frac{\tau}{2}\fnorm{\dall}^2.
\end{equation*}
On the other hand, triangle inequality implies
\begin{equation*}
\lambda_n (\|\hbG\|_*-\|\tg\|_* )\geq -\lambda_n\|\Delta_{\bG}\|_*.
\end{equation*}
Together with \eqref{eq:basic}, the last two displays imply
\begin{equation*}
\langle \nabla_{\bTheta} \wl (\tall),\dall \rangle +\frac{\tau}{2}\fnorm{\dall}-\lambda_n\|\dg\|_*\leq 0.
\end{equation*}
Triangle inequality further implies
\begin{equation}
\begin{aligned}
\frac{\tau}{2}\|\Delta_{\bTheta}\|_F^2
& \leq 
\lambda_n\|\Delta_{\bG}\|_*
+ | \langle \nabla_{\bTheta} \wl (\tall),\dg+ \da\one^\top +\one\da^\top  \rangle |
 + | \db \langle \nabla_{\bTheta} \wl(\tall), \bX \rangle |\\
& = \lambda_n\|\Delta_{\bG}\|_*
+ | \langle \bA-\bP, \dg+ 2\da\one^\top \rangle |
+ | \db\, \langle \bA-\bP, \bX \rangle |\\
& \leq \lambda_n\|\Delta_{\bG}\|_*
+ | \langle \bA-\bP, \dg+ 2\da\one^\top \rangle |
 + \lambda_n\sqrt{k}\fnorm{\db X}.
\end{aligned}
 \label{convexity}
\end{equation}
Here the equality is due to the symmetry of $\bA-\bP$ and the last inequality is due to the condition imposed on $\lambda_n$.
We now further upper bound the first two terms on the rightmost side.
First, by Lemma~\ref{lem:kernel-cone} and the assumption that $| \langle \bA-\bP, \bX \rangle|\leq \lambda_n\sqrt{k}\fnorm{X}$, we have
\begin{equation}
\begin{aligned}
\|\Delta_{\bG}\|_* & \leq 4\sqrt{2k}\,\fnorm{\mathcal{P}_{\mathcal{M}_k}\dg} 
+ 2\fnorm{\da\one^\top}+ 2\sqrt{k}\,\fnorm{\db X} + 4 \|\gresid\|_*. 
\end{aligned}
\label{eq:kernel-part2}
\end{equation}
Moreover, H\"{o}lder's inequality implies
\begin{equation}
\begin{aligned}
|\langle  \bA-\bP,\dg+2 \da\one^\top  \rangle | 
& \leq \opnorm{\bA-\bP} (\|\dg\|_*+2\|\da\one^\top\|_* )\\
& = \opnorm{\bA-\bP} (\|\dg\|_*+2\fnorm{\da\one^\top} )\\
&\leq \frac{\lambda_n}{2} (\|\dg\|_* + 2 \fnorm{\da\one^\top} ).
\end{aligned}
\label{eq:kernel-part1}
\end{equation}
Here the equality holds since $\da\one^\top$ is a rank one matrix.
Substituting \eqref{eq:kernel-part2} and \eqref{eq:kernel-part1} into \eqref{convexity}, we obtain that
\begin{equation*}
\begin{aligned}
\frac{\tau}{2}\fnorm{\dall}^2 &\leq \frac{3\lambda_n}{2}\|\dg\|_*+\lambda_n \fnorm{\da\one^\top}+\lambda_n\sqrt{k}\fnorm{\db X}\\
&\leq \frac{3\lambda_n}{2} (4\sqrt{2k}\fnorm{\mathcal{P}_{\mathcal{M}_k}\dg} + 2\fnorm{\da\one^\top}+ 2\sqrt{k}\|\db \bX\|_F + 4\|\gresid\|_* )\\
& \hskip 15em
+\lambda_n \fnorm{\da\one^\top}+\lambda_n\sqrt{k}\fnorm{\db X}\\
& \leq C_1\lambda_n \big( \sqrt{k}\, (\fnorm{\mathcal{P}_{\mathcal{M}_k}\dg} + \fnorm{\da\one^\top} + \fnorm{\db \bX} ) +  \|\gresid\|_* \big).
\end{aligned}
\end{equation*}
By Lemma~\ref{lem:kernel-identifiability}, we can further bound the righthand side as
\begin{equation*}
\begin{aligned}
\frac{\tau}{2}\fnorm{\dall}^2 &\leq  C_2\lambda_n \sqrt{k}\, (\fnorm{\dall} + \|\gresid\|_* /\sqrt{k} ) + C_1\lambda_n \|\gresid\|_*\\
&\leq C_2\lambda_n \sqrt{k}\, \fnorm{\dall} + (C_1 + C_2)\lambda_n \|\gresid\|_*.
\end{aligned}
\end{equation*}
Solving the quadratic inequality, we obtain
\begin{equation*}
\fnorm{\dall}^2 \leq C'\left(\frac{\lambda_n^2k}{\tau^2} + \frac{\lambda_n \|\gresid\|_*}{\tau}\right).
\end{equation*}
Note that $\tau\geq ce^{-M_1}$ for some positive constant $c$. 
Therefore, 
\begin{equation*}
	\fnorm{\dall}^2\leq {C}\left(e^{2M_1}\lambda_n^2k + e^{M_1}\lambda_n \|\gresid\|_*\right).
\end{equation*}

$2^\circ$ We now turn to the probabilistic bound.
By Lemma~\ref{lem:prob}, there exist constants $c_1, C_1$ such that for any $\lambda_n\geq 2C_1\sqrt{\max\left\{ne^{-M_2}, \log n\right\}}$, we have uniformly over the parameter space that
\begin{equation*}
\P \left(\lambda_n \geq 2\max\left\{\opnorm{A - P}, \frac{\inner{\bA-\bP, \bX}}{\sqrt{k}\fnorm{X}} \right\} \right) \geq 1 - n^{-c_1}.
\end{equation*}
Denote this event as $E$. Since the conditions on $\lambda_n$ in the first part of Theorem~\ref{thm:kernel-convex-upper} are satisfied on $E$, it follows that there exists an absolute constant $C>0$ such that uniformly over the parameter space, with probability at least $1-n^{-c_1}$,
$\fnorm{\dall}^2 \leq C \phi_n^2$.
This completes the proof.
\qed

\subsubsection{Proof of Lemma~\ref{lem:kernel-cone}}
\label{sec:proof-43-lem1}

By the convexity of $\wl (\bTheta)$,
\begin{align*}
\wl(\hbTheta)-\wl(\tall) 
&\geq  \langle \nabla_{\bTheta} \wl(\tall), \dall \rangle\\
& = - \langle  \bA-\bP,\, \dg+2 \da\one^\top +\db \bX \rangle\\
&\geq  -\op{\bA-\bP}\big(\|\dg\|_*+2\|\da\one^\top\|_*\big) 
- | \langle  \bA-\bP, \db\bX \rangle|\\
&\geq 
-\frac{\lambda_n}{2}
\big(\|\mathcal{P}_{\mathcal{M}_k}\dg\|_*+\|\mathcal{P}_{\mathcal{M}^\perp_k}\dg\|_*+2\fnorm{\da\one^\top}\big) 
- | \langle  \bA-\bP, \db\bX \rangle|\,.
\end{align*}
The last inequality holds since $\lambda_n\geq 2\op{\bA-\bP}$ and $\calP_{\calM_k}+\calP_{\calM_k^\perp}$ equals identity.
On the other hand, by the definition of $\gresid$, 
\begin{equation*}
\begin{aligned}
\|\hbG\|_*-\|\tg\|_*&=\|\mathcal{P}_{\mathcal{M}_k}\tg+ \gresid + \mathcal{P}_{\mathcal{M}_k}\dg+\mathcal{P}_{\mathcal{M}^\perp_k}\dg\|_*-\|\mathcal{P}_{\mathcal{M}_k}\tg+ \gresid\|_*\\
&\geq \|\mathcal{P}_{\mathcal{M}_k}\tg+\mathcal{P}_{\mathcal{M}^\perp_k}\dg\|_* -\|\gresid\|_*-\|\mathcal{P}_{\mathcal{M}_k}\dg\|_* - \|\mathcal{P}_{\mathcal{M}_k}\tg\|_*- \|\gresid\|_*\\
&=\|\mathcal{P}_{\mathcal{M}_k}\tg\|_*+\|\mathcal{P}_{\mathcal{M}^\perp_k}\dg\|_* - 2\|\gresid\|_* - \|\mathcal{P}_{\mathcal{M}_k}\dg\|_* - \|\mathcal{P}_{\mathcal{M}_k}\tg\|_* \\
&=\|\mathcal{P}_{\mathcal{M}^\perp_k}\dg\|_*-\|\mathcal{P}_{\mathcal{M}_k}\dg\|_* - 2\|\gresid\|_*\,.
\end{aligned}
\end{equation*} 
Here, the second last equality holds since $\calP_{\calM_k}\tg$ and $\mathcal{P}_{\mathcal{M}^\perp_k}\dg$ have orthogonal column and row spaces.
Furthermore, since $\hbTheta$ is the optimal solution to \eqref{eq:convex-obj}, and $\tall$ is feasible, the basic inequality and the last two displays imply
\begin{equation*}
\begin{aligned}
0 &\geq \wl (\hbTheta) - \wl (\tall)+\lambda_n \big( \|\hbG\|_* - \|\tg\|_* \big)\\
&\geq -\frac{\lambda_n}{2}\big(\|\mathcal{P}_{\mathcal{M}_k}\dg\|_*+\|\mathcal{P}_{\mathcal{M}^\perp_k}\dg\|_*+2\fnorm{\da\one^\top}\big)\\
&\quad\quad\quad\quad\quad - | \langle  \bA-\bP, \db\bX \rangle| + \lambda_n\big(\|\mathcal{P}_{\mathcal{M}^\perp_k}\dg\|_*-\|\mathcal{P}_{\mathcal{M}_k}\dg\|_* - 2\|\gresid\|_*\big)\\
& =  \frac{\lambda_n}{2} \big(\|\mathcal{P}_{\mathcal{M}^\perp_k}\dg\|_*- 3\|\mathcal{P}_{\mathcal{M}_k}\dg\|_*  - 4\|\gresid\|_* - 2\fnorm{\da\one^\top}\big)
- | \langle  \bA-\bP, \db\bX \rangle|\,.
\end{aligned}
\end{equation*}
Rearranging the terms leads to
\begin{equation*}
\|\mathcal{P}_{\mathcal{M}^\perp_k}\dg \|_* \leq 3\|\mathcal{P}_{\mathcal{M}_k}\dg \|_* + 2\fnorm{\da\one^\top}  + \frac{2}{\lambda_n}| \langle \bA-\bP, \db\bX \rangle| + 4\|\gresid\|_*\,,
\end{equation*}
and triangle inequality further implies 
\begin{equation*}
\|\dg \|_* \leq 4\|\mathcal{P}_{\mathcal{M}_k}\dg \|_* + 2\fnorm{\da\one^\top}+ \frac{2}{\lambda_n}|\big \langle \bA-\bP, \db\bX \big\rangle| + 4\|\gresid\|_*\,.
\end{equation*}
Last but not least, note that 
the rank of $\mathcal{P}_{\mathcal{M}_k}\dg$ is at most $2k$, and so we complete the proof by further bounding the first term on the righthand side of the last display by  $4\sqrt{2k}\fnorm{\mathcal{P}_{\mathcal{M}_k}\dg}$.

\subsubsection{Proof of Lemma~\ref{lem:kernel-identifiability}}
\label{sec:proof-43-lem2}
By definition, we have the decomposition
\begin{equation*}
\begin{aligned}
\fnorm{\dall}^2&=\fnorm{\dg+\da\one^\top+\one\da^\top+\db\bX}^2\\
& = \fnorm{\dg+\da\one^\top+\one\da^\top}^2
+ \fnorm{\db X}^2
+ 2\, \langle \dg+\da\one^\top+\one\da^\top, \db\bX \rangle\\
& = \fnorm{\dg}^2
+2\fnorm{\da\one^\top}^2
+2\,\tr(\da\one^\top\da\one^\top)+\fnorm{\db X}^2
+2\,\langle \dg+2\da\one^\top, \db\bX \rangle\,.
\end{aligned}
\end{equation*}
Here the last equality is due to the symmetry of $\bX$ and the fact that $\dg \one = 0$. 
Since $\tr(\da\one^\top\da\one^\top)=\tr(\one^\top\da\one^\top\da)=|\one^\top\da|^2\geq 0$, the last display implies
\begin{equation}
\begin{aligned}
\fnorm{\dall}^2&\geq \fnorm{\dg}^2+2\fnorm{\da\one^\top}^2+\fnorm{\db X}^2+2\,\langle \dg+2\da\one^\top, \db\bX \rangle.
\end{aligned}
\label{eq:identi-obj}
\end{equation}
Furthermore, we have
\begin{align*}
&  |\langle \dg +2\da\one^\top, \db\bX \rangle| \\
& \leq \|\dg\|_*\opnorm{\db\bX} + 2 \|\da\one^\top\|_*\opnorm{\db\bX}
\\
& \leq \big(4\sqrt{2k}\fnorm{\mathcal{P}_{\mathcal{M}_k}\dg} + 4\fnorm{\da\one^\top}+ \frac{2}{\lambda_n}| 
\langle \bA-\bP, \db\bX \rangle| + 4\|\gresid\|_*\big)\opnorm{\db\bX}\\
&\leq \big(4\sqrt{2k}\fnorm{\mathcal{P}_{\mathcal{M}_k}\dg} + 4\fnorm{\da\one^\top}+ 2\sqrt{k}\fnorm{\db\bX} + 4\|\gresid\|_* \big)
\frac{\fnorm{\db X}}{\sqrt{\srank(\bX)}}\\
& \leq \frac{C_0\sqrt{k}}{\sqrt{\srank(\bX)}} \big(\fnorm{\dg}^2+2\fnorm{\da\one^\top}^2+\fnorm{\db\bX}^2\big) + \frac{4\|\gresid\|_*}{\sqrt{\srank(\bX)}} \fnorm{\db X}\\
& \leq \frac{C_0\sqrt{k}}{\sqrt{\srank(\bX)}} \big(\fnorm{\dg}^2+2\fnorm{\da\one^\top}^2+\fnorm{\db\bX}^2\big) + \frac{2\|\gresid\|_*^2}{c_0\srank(\bX)} + 2c_0\fnorm{\db X}^2
\end{align*}
for any constant $c_0\geq 0$. 
Here, the first inequality holds since the operator norm and the nuclear norm are dual norms under trace inner product.
The second inequality is due to Lemma~\ref{lem:kernel-cone} and the fact that $\|\da\one^\top\|_* = \fnorm{\da\one^\top}$ since $\da\one^\top$ is of rank one.
The third inequality is due to the definition of $\srank(X)$ and that $| \langle \bA-\bP, \bX \rangle | \leq \lambda_n \sqrt{k}\fnorm{X}$ by assumption and $\db$ is a scalar.
The fourth inequality is due to Assumption \ref{assump:srank} and the last due to $2ab\leq a^2 + b^2$ for any $a,b\in \mathbb{R}$.
Substituting these inequalities into \eqref{eq:identi-obj} leads to
\begin{equation*}
\begin{aligned}
\fnorm{\dall}^2 & \geq \left(1-\frac{2C_0\sqrt{k}}{\sqrt{\srank(\bX)}} \right)\fnorm{\dg}^2 + \left(2 - \frac{2C_0\sqrt{k}}{\sqrt{\srank(\bX)}}\right)\fnorm{\da\one^\top}^2\\
&\quad\quad+\left(1 - \frac{2C_0\sqrt{k}}{\sqrt{\srank(\bX)}} - 4c_0 \right)\fnorm{\db X}^2 - \frac{4\|\gresid\|_*^2}{c_0\srank(\bX)}\,.
\end{aligned}
\end{equation*}
On the other hand, notice that $\tr(\da\one^\top\da\one^\top)\leq \|\da\one^\top\|_F^2$, we have
\begin{equation*}
\begin{aligned}
\fnorm{\dall}^2&\leq \left(1 + \frac{2C_0\sqrt{k}}{\sqrt{\srank(\bX)}}  \right)\fnorm{\dg}^2+\left(4 + \frac{2C_0\sqrt{k}}{\sqrt{\srank(\bX)}} \right)\fnorm{\da\one^\top}^2\\
&\quad\quad+\left(1 + \frac{2C_0\sqrt{k}}{\sqrt{\srank(\bX)}} + 4c_0  \right)\fnorm{\db X}^2 + \frac{4\|\gresid\|_*^2}{c_0\srank(\bX)}\,.
\end{aligned}
\end{equation*}
Together with Assumption \ref{assump:srank}, the last two displays complete the proof.

\subsection{Proofs of Lemma~\ref{lem:error-metric} and Theorem~\ref{thm:kernel-nonconvex-upper}}
\newcommand{\et}{\wt{e}}
Again, we directly prove the results under the general model. Recall that $\tg \approx U_kD_kU_k^\top$ is the top-$k$ eigen-decomposition of $\tg$, $\tz=U_kD_k^{\half}$, $\gresid = \tg - U_kD_kU_k^\top$ and $\dgt = \zt(\zt)^\top - \tz\tz^\top$. For the convenience of analysis, we will instead analyze the following quantity, 
\begin{equation*}
\et_t=\op{Z^0}^2\fnorm{\dzt}^2+2\fnorm{\dat\one^\top}^2+\fnorm{\dbt X}^2.
\end{equation*}
Under Assumption~\ref{assump:init}, 
\begin{equation}
	\op{\Delta_{Z^0}} \leq \delta \op{\tz}, ~~
	(1-\delta)e_t \leq \et_t \leq (1+\delta) e_t. 
\label{eq:ett}
\end{equation}
for some sufficiently small constant $\delta \in (0, 1)$. The rest of the proof relies on the following  lemmas. 
\begin{lemma}
	For any $\tall\in\paramK$, $\max\limits_{1\leq i\leq n}\|(\tz)_i\|_2^2\leq M_1/3$.
\end{lemma}
\begin{proof}
	By definition, $\tg - \tz\tz^\top\in\mathcal{S}_+^n$, which implies, $e_i^\top \left(\tg - \tz\tz^\top\right)e_i = G_{ii} - \|(\tz)_i\|_2^2\geq 0$, that is $\|(\tz)_i\|_2^2 \leq G_{ii}\leq M_1/3$ for any $1\leq i\leq n$.
\end{proof}

\begin{lemma}
If Assumption~\ref{assump:srank} holds, there exist constants $0\leq c_0<1$ and $C_0$  such that  
\begin{equation*}
\begin{gathered}
\fnorm{\dallt}^2  \geq (1 - c_0)\left(\fnorm{Z^t(Z^t)^\top - \tz\tz^\top}^2+2\fnorm{\dat\one^\top}^2+\fnorm{\dbt X}^2\right) - C_0\fnorm{\gresid}^2, \\
\fnorm{\dallt}^2  \leq (1 + c_0)\left(\fnorm{Z^t(Z^t)^\top - \tz\tz^\top}^2+2\fnorm{\dat\one^\top}^2+\fnorm{\dbt X}^2\right) + C_0\fnorm{\gresid}^2. 
\end{gathered}
\end{equation*}
\label{lem:kernel-identifiability-nonconvex}
\end{lemma}
\begin{proof}
See \prettyref{sec:proof-44-lem1}.
\end{proof}
 
\begin{lemma} 
Under Assumption~\ref{assump:srank}, let $\zeta_n = \lambdacond$, if $\fnorm{\dzt}\leq c_0e^{-M_1}\op{\tz}/\kappa^2_{\tz}$ and $\op{\tz}^2\geq C_0e^{M_1}\kappa_{\tz}^2\zeta_n^2$ for sufficiently small constant $c_0$ and sufficiently large constant $C_0$, there exist a constant $c$ such that, for any $\eta\leq c$, there exist positive constants $\rho$ and $C$, 
\begin{equation*}
\begin{aligned}
\et_{t+1}&\leq \left(1- \frac{\eta}{e^{M_1}\kappa^2}\rho\right)\et_t + \eta C \left(\fnorm{\gresid}^2 + e^{M_1}\zeta_n^2k \right).
\end{aligned}
\end{equation*}
\label{lem:kernel-nonconvex-1}
\end{lemma}
\begin{proof}
See \prettyref{sec:proof-44-lem2}.
\end{proof}

 
\begin{lemma}
Under Assumption~\ref{assump:srank}, let $\zeta_n = \lambdacond$, if $\op{\tz}^2\geq C_1\kappa_{\tz}^2\zeta_n^2 e^{M_1}\max\left\{\sqrt{\eta\fnorm{\gresid}^2/\zeta_n^2}, \sqrt{\eta k e^{M_1}}, 1\right\}$ for a sufficiently large constant $C_1$ and $\et_0\leq c_0^2e^{-2M_1}\op{\tz}^4/4\kappa_{\tz}^4$, then for all $t\geq 0$, 
\begin{equation*}
\fnorm{\dzt}\leq \frac{c_0}{e^{M_1}\kappa_{\tz}^2}\op{\tz} .
\end{equation*}
\label{lem:kernel-nonconvex-2}
\end{lemma}
\begin{proof}
See \prettyref{sec:proof-44-lem3}.
\end{proof}

\paragraph{Proof of Lemma~\ref{lem:error-metric}}
By Lemma~\ref{lem:kernel-identifiability-nonconvex}, notice that $\gresid = 0$ under the inner product model, 
\begin{equation}
\begin{gathered}
\fnorm{\dallt}^2  \geq (1 - c_0)\left(\fnorm{Z^t(Z^t)^\top - \tz\tz^\top}^2+2\fnorm{\dat\one^\top}^2+\fnorm{\dbt X}^2\right),\\
\fnorm{\dallt}^2  \leq (1 + c_0)\left(\fnorm{Z^t(Z^t)^\top - \tz\tz^\top}^2+2\fnorm{\dat\one^\top}^2+\fnorm{\dbt X}^2\right).
\end{gathered}
\label{eq:error-metric-inner}
\end{equation}
By Lemma~\ref{lemma-tu1}, 
\begin{equation*}
	\fnorm{Z^t(Z^t)^\top - \tz\tz^\top}^2 \geq 2(\sqrt{2}-1)\kappa_{\tz}^{-2} \op{\tz}^2\fnorm{\dzt}^2 
\end{equation*}
which implies, 
\begin{equation*}
	e_t \leq \frac{\kappa_{\tz}^{2}}{2(\sqrt{2}-1)} \fnorm{Z^t(Z^t)^\top - \tz\tz^\top}^2 +2\fnorm{\dat\one^\top}^2+\fnorm{\dbt X}^2 
	\leq \frac{\kappa_{\tz}^{2}}{2(\sqrt{2}-1)(1-c_0)}\fnorm{\dallt}^2. 
\end{equation*}
Similarly, by Lemma~\ref{lemma-tu2}, when $\text{dist}(\zt, \tz)\leq c\op{\tz}$, 
\begin{equation*}
	\fnorm{Z^t(Z^t)^\top - \tz\tz^\top}^2 \leq (2+c)^2 \op{\tz}^2\fnorm{\dzt}^2,
\end{equation*}
and this implies, 
\begin{equation*}
\begin{aligned}
		e_t &\geq \frac{1}{(2+c)^2}\fnorm{Z^t(Z^t)^\top - \tz\tz^\top}^2 + 2\fnorm{\dat\one^\top}^2+\fnorm{\dbt X}^2 \\
		& \geq 	\frac{1}{(2+c)^2(1+c_0)}\op{\tz}^2\fnorm{\dzt}^2.
\end{aligned}
\end{equation*}
\qed
 
\paragraph{Proof of \prettyref{thm:kernel-nonconvex-upper}}
Consider the deterministic bound first.
	By Lemma~\ref{lem:kernel-nonconvex-2}, for all $t\geq 0$, 
\begin{equation*}
\fnorm{\dzt}\leq \frac{c_0}{e^{M_1}\kappa_{\tz}^2}\op{\tz}.
\end{equation*}
Then apply Lemma~\ref{lem:kernel-nonconvex-1}, there exists positive constants $\rho$ and $M$ such that for all $t\geq 0$,
\begin{equation*}
\begin{aligned}
\et_{t+1}\leq \left(1- \frac{\eta}{e^{M_1}\kappa_{\tz}^2}\rho\right)\et_t + \eta C \left(\fnorm{\gresid}^2 + e^{M_1}\zeta_n^2k\right) .
\end{aligned}
\end{equation*}
Therefore, 
\begin{equation*}
\begin{aligned}
\et_t&\leq \left(1- \frac{\eta}{e^{M_1}\kappa_{\tz}^2}\rho\right)^t \et_0+ \sum_{i=0}^t\eta C \left(\fnorm{\gresid}^2 + e^{M_1}\zeta_n^2k\right) \left(1-\frac{\eta}{e^{M_1}\kappa_{\tz}^2}\rho\right)^i\\
&\leq \left(1 - \frac{\eta}{e^{M_1}\kappa_{\tz}^2}\rho\right)^t \et_0 + \frac{C\kappa^2}{\rho} \left( e^{2M_1}\zeta_n^2k + e^{M_1}\fnorm{\gresid}^2 \right).
\end{aligned}
\end{equation*}
Notice that $0.9 e_t\leq \et_t \leq 1.1 e_t$, 
\begin{equation*}
\begin{aligned}
e_t&\leq 2\left(1 - \frac{\eta}{e^{M_1}\kappa_{\tz}^2}\rho\right)^t e_0 + \frac{2C\kappa^2}{\rho} \left( e^{2M_1}\zeta_n^2k + e^{M_1}\fnorm{\gresid}^2 \right).
\end{aligned}
\end{equation*}

Given the last display, 
the proof of the probabilistic bound is nearly the same as that of the counterpart in Theorem~\ref{thm:kernel-convex-upper} and we leave out the details. 
\qed

\subsubsection{Proof of Lemma~\ref{lem:kernel-identifiability-nonconvex}}
\label{sec:proof-44-lem1}
By definition,
\begin{equation*}
\begin{aligned}
\fnorm{\dgt}^2 &= \fnorm{Z^t(Z^t)^\top - \tz\tz^\top - \gresid}^2\\
& \geq \fnorm{Z^t(Z^t)^\top - \tz\tz^\top}^2 + \fnorm{\gresid}^2 - 2| \big\langle Z^t(Z^t)^\top - \tz\tz^\top, \gresid \big\rangle|\\
& \geq \fnorm{Z^t(Z^t)^\top - \tz\tz^\top}^2 + \fnorm{\gresid}^2 - 2\fnorm{Z^t(Z^t)^\top - \tz\tz^\top}\fnorm{\gresid}\\
&\geq \fnorm{Z^t(Z^t)^\top - \tz\tz^\top}^2 + \fnorm{\gresid}^2 - c_1\fnorm{Z^t(Z^t)^\top - \tz\tz^\top}^2 - c_1^{-1}\fnorm{\gresid}^2\\
&\geq (1-c_1)\fnorm{Z^t(Z^t)^\top - \tz\tz^\top}^2 - (c_1^{-1}-1)\fnorm{\gresid}^2
\end{aligned}
\end{equation*}
where the second last inequality comes from $a^2 + b^2\geq 2ab$ and holds for any $c_1\geq 0$. Similarly, it could be shown that 
\begin{equation*}
	\fnorm{\dgt}^2\leq (1+c_1)\fnorm{Z^t(Z^t)^\top - \tz\tz^\top}^2 + (1 + c_1^{-1})\fnorm{\gresid}^2.
\end{equation*}
Expanding the term $\fnorm{\dallt}^2$, we obtain
\begin{equation*}
\begin{aligned}
\fnorm{\dallt}^2&=\fnorm{\dgt+\dat\one^\top+\one\dat^\top+\dbt\bX}^2\\
&=\fnorm{\dgt+\dat\one^\top+\one\dat^\top}^2+\fnorm{\dbt X}^2+2\big\langle \dgt+\dat\one^\top+\one\dat^\top, \dbt\bX \big\rangle\\
&=\fnorm{\dgt}^2+2\fnorm{\dat\one^\top}^2+2\tr(\dat\one^\top\dat\one^\top)+\fnorm{\dbt X}^2\\
&\quad\quad+2\big\langle \dgt+2\dat\one^\top, \dbt\bX \big\rangle,
\end{aligned}
\end{equation*}
where the last equality is due to the symmetry of $\bX$. 
Notice that $\,\tr(\da\one^\top\da\one^\top)=\tr(\one^\top\da\one^\top\da)=|\one^\top\da|^2\geq 0$, 
\begin{equation}
\begin{aligned}
\fnorm{\dallt}^2&\geq (1-c_1)\fnorm{Z^t(Z^t)^\top - \tz\tz^\top}^2 - (c_1^{-1}-1)\fnorm{\gresid}^2+2\fnorm{\dat\one^\top}^2+\fnorm{\dbt X}^2\\
&\quad+2\big\langle Z^t(Z^t)^\top - \tz\tz^\top +  2\dat\one^\top, \dbt\bX \big\rangle - 2\inner{\gresid, \dbt\bX }.
\end{aligned}
\label{identi-obj}
\end{equation}
By H\"{o}lder's inequality,
\begin{equation*}
\begin{aligned}
|\big\langle Z^t(Z^t)^\top - \tz\tz^\top &+  2\dat\one^\top, \dbt\bX \big\rangle|\leq \left(\|Z^t(Z^t)^\top - \tz\tz^\top\|_* + 2\|\dat\one^\top\|_* \right)\op{\dbt\bX}\\
&\leq \left( \sqrt{2k}\fnorm{Z^t(Z^t)^\top - \tz\tz^\top} + 2\fnorm{\dat\one^\top}\right)\op{\dbt\bX} \\
&\leq \left( \sqrt{2k}\fnorm{Z^t(Z^t)^\top - \tz\tz^\top} + 2\fnorm{\dat\one^\top}\right)\fnorm{\dbt\bX}/\sqrt{\srank(X)}\\
&\leq C_1\sqrt{\frac{k}{\srank(X)}}\left(\fnorm{Z^t(Z^t)^\top - \tz\tz^\top}^2 + \fnorm{\dat\one^\top}^2+\fnorm{\dbt\bX}^2\right),
\end{aligned}
\end{equation*}
and for any $c>0$, 
\begin{equation*}
	|\inner{\gresid, \dbt\bX }| \leq \fnorm{\gresid}\fnorm{\dbt\bX} \leq c\fnorm{\dbt\bX}^2 + \frac{1}{4c}\fnorm{\gresid}^2. 
\end{equation*}
Substitute these inequalities into \eqref{identi-obj}, 
\begin{equation*}
\begin{aligned}
\fnorm{\dall}^2 & \geq \left(1-2C_1\sqrt{\frac{k}{\srank(X)}}-c_1 \right)\fnorm{Z^t(Z^t)^\top - \tz\tz^\top}^2 + \left(2 - 2C_1\sqrt{\frac{k}{\srank(X)}}\right)\fnorm{\dat\one^\top}^2\\
&\quad\quad+\left(1 - 2C_1\sqrt{\frac{k}{\srank(X)}} - 2c \right)\fnorm{\dbt X}^2 - (c_1^{-1}+1/2c)\fnorm{\gresid}^2.
\end{aligned}
\end{equation*}
On the other hand, notice that $\tr(\dat\one^\top\dat\one^\top)\leq \|\dat\one\|_F^2$, we have
\begin{equation*}
\begin{aligned}
\fnorm{\dall}^2 & \leq \left(1+2C_1\sqrt{\frac{k}{\srank(X)}}+c_1 \right)\fnorm{Z^t(Z^t)^\top - \tz\tz^\top}^2 + \left(2 + 2C_1\sqrt{\frac{k}{\srank(X)}}\right)\fnorm{\da\one^\top}^2\\
&\quad\quad+\left(1 + 2C_1\sqrt{\frac{k}{\srank(X)}} + 2c \right)\fnorm{\dbt X}^2 + (c_1^{-1}+1/2c)\fnorm{\gresid}^2.
\end{aligned}
\end{equation*}
This completes the proof.

\subsubsection{Proof of Lemma~\ref{lem:kernel-nonconvex-1}}
\label{sec:proof-44-lem2}

\newcommand{\dalltk}{\Delta_{\widebar{\Theta}^t}}
Let $\allt=\at\one^\top+\one(\at)^\top+\bt\bX+\zt(\zt)^\top$, $\bR^t=\argmin\limits_{\bR\in\mathbb{R}^{r\times r}, \bR\bR^\top=\bI_r} \fnorm{\zt-\tz\bR}$, $\widetilde{\bR}^t=\argmin\limits_{\bR\in\mathbb{R}^{r\times r}, \bR\bR^\top=\bI_r}\fnorm{\widetilde{\bZ}_t-\tz\bR}$ and $\dzt=\zt-\tz\bR^t$, then
\begin{equation*}
\begin{aligned}
\fnorm{\bZ^{t+1}-\tz\bR^{t+1}}^2&\leq \fnorm{\bZ^{t+1}-\tz\widetilde{\bR}^{t+1}}^2
\leq \fnorm{\widetilde{\bZ}_{t+1}-\tz\widetilde{\bR}^{t+1}}^2 
\leq \fnorm{ \widetilde{\bZ}_{t+1}-\tz\bR^{t}}^2.
\end{aligned}
\end{equation*}
The first and the last inequalities are due to the definition of $\bR^{t+1}$ and $\widetilde{\bR}^{t+1}$, and the second inequality is due to the projection step. 
Plugging in the definition of $ \widetilde{\bZ}^{t+1}$, we obtain
\begin{equation*}
\begin{aligned}
\fnorm{\bZ^{t+1}-\tz \bR^{t+1}}^2&\leq \fnorm{\zt-\tz\bR^{t}}^2+\etaz^2\fnorm{\samt \zt}^2- 2\etaz\inner{\samt\zt, \zt-\tz\bR^t }\\
&=\fnorm{\zt-\tz\bR^{t}}^2+\etaz^2\fnorm{\samt \zt}^2- 2\etaz\inner{\samt, (\zt-\tz\bR^t)(\zt)^\top }.
\end{aligned}
\end{equation*} 
Note that 
\begin{equation*}
\zt(\zt)^\top-\tz\bR^t(\zt)^\top=\frac{1}{2}(\zt(\zt)^\top-\tz\tz^\top)+\frac{1}{2}(\zt(\zt)^\top+\tz\tz^\top)-\tz\bR(\zt)^\top.
\end{equation*}
Also due to the symmetry of $\samt$,
\begin{equation*}
\inner{\samt , \frac{1}{2}(\zt(\zt)^\top+\tz\tz^\top)-\tz\bR(\zt)^\top}=\frac{1}{2}\inner{\samt, \dzt\dzt^\top}.
\end{equation*}
Therefore, combine the above three equations,
\begin{equation}
\begin{aligned}
\fnorm{ \bZ^{t+1}-\tz\bR^{t+1}}^2&\leq \fnorm{ \zt-\tz\bR^{t}}^2+\etaz^2\fnorm{\samt \zt}^2-\etaz\inner{\samt , \dzt\dzt^\top} \\
&\quad -\etaz\inner{\samt ,(\zt(\zt)^\top-\tz\tz^\top)}.
\end{aligned}
\label{non-convex-obj-Z}
\end{equation}
By similar and slightly simpler arguments, we also obtain
\begin{align}
\| \balpha^{t+1}-\ta\|^2&\leq \|\widetilde{\balpha}_{t+1}-\ta\|^2 \nonumber \\
&=\| \balpha^{t}-\ta\|^2+\etaa^2\fnorm{\samt \one}^2- 2\etaa\inner{\samt\one, \balpha^{t}-\ta}.
\label{non-convex-obj-a}
\end{align}
\begin{align}
\| \beta^{t+1}-\tb\|^2&\leq \|\widetilde{\beta}_{t+1}-\tb\|^2 \nonumber \\
&=\| \beta^{t}-\tb\|^2+\etab^2\inner{\samt, \bX}^2- 2\etab\inner{\samt, (\beta^{t}-\tb)\bX}.
\label{non-convex-obj-b}
\end{align}

For $h(\Theta)$ in \eqref{eq:h}, define
\begin{equation*}
	H(\Theta) = \E_{\Theta_\star}[h(\Theta)] - \sum_{i=1}^n \Theta_{ii}\sigma(\Theta_{\star,ii}).
\end{equation*}
Then it is straightforward to verify that $\nabla H(\Theta) = \sigma(\Theta) - \sigma(\Theta_\star)$ and so $\nabla H(\Theta_\star) = 0$.
With $\etaz=\eta/\opnorm{Z^0}^2, \etaa=\eta/2n, \etab=\eta/2\fnorm{X}^2$ , the weighted sum $\op{Z^0}^2\times$\eqref{non-convex-obj-Z}+$2n\times$\eqref{non-convex-obj-a}+$\fnorm{X}^2\times$\eqref{non-convex-obj-b} is equivalent to
\begin{equation*}
\begin{aligned}
	\et_{t+1} &\leq \et_t-\eta\,\inner{\samt,\dzzt+2(\balpha^{t}-\ta)\one^\top+(\beta^{t}-\tb)\bX+\dzt\dzt^\top} \\
	& ~~~ + \left(\op{Z^0}^2\etaz^2\fnorm{\samt \zt}^2+2n\etaa^2\fnorm{\samt \one}^2+\fnorm{X}^2\etab^2\inner{\samt, \bX}^2 \right)\\
	&\leq \et_t-\eta\,\inner{\samt, \dalltk}-\eta\,\inner{\samt,\dzt\dzt^\top} \\
	& ~~~ +\left(\frac{\eta^2}{\op{Z^0}^2}\fnorm{\samt \zt}^2+\frac{\eta^2}{2n}\fnorm{\samt \one}^2+\frac{\eta^2}{4\fnorm{X}^2}\inner{\samt, \bX}^2 \right),\\
\end{aligned}
\end{equation*}
where $\dalltk = \dzzt + \dat\one^\top + \one(\dat)^\top + \dbt\bX = \dallt - \gresid$. 
Then, simple algebra further leads to
\begin{align}
\et_{t+1}&\leq \et_t-\eta\inner{\samt - \popt, \dalltk}-\eta\inner{\popt,\dallt}- \eta\inner{\popt, \gresid} - \eta\inner{\samt,\dzt\dzt^\top} 
\nonumber \\
& ~~~ + \left(\frac{\eta^2}{\op{Z^0}^2}\fnorm{\samt \zt}^2+\frac{\eta^2}{2n}\fnorm{\samt \one}^2+\frac{\eta^2}{4\fnorm{X}^2}\inner{\samt, \bX}^2 \right) 
\nonumber \\
&\leq  \et_t-\eta\inner{\popt,\dallt}+\eta |\inner{\samt-\popt,\dalltk}|  + \eta |\inner{\sam,\dzt\dzt^\top}|  
\nonumber \\
& ~~~ + \eta|\inner{\popt,\gresid}| +\eta^2\left(\frac{1}{\op{Z^0}^2}\fnorm{\samt \zt}^2+\frac{1}{2n}\fnorm{\samt \one}^2+\frac{1}{4\fnorm{X}^2}\inner{\samt, \bX}^2 \right) \nonumber \\
&=\et_t-\eta D_1+\eta D_2+\eta D_3+\eta D_4+\eta^2 D_5.
\label{nonconvex-obj}
\end{align}

In what follows, we control 
Note that for any $\bTheta\in \calF_g$, 
\begin{equation*}
\frac{1}{4}\bI_{n^2\times n^2}\succeq \nabla^2 \wll(\bTheta)
= \mathrm{diag}\Big(\mathrm{vec}\big(\sigma(\bTheta)\circ \left(1-\sigma(\bTheta)\right)\big)\Big)\succeq \tau\bI_{n^2\times n^2}	
\end{equation*}
where $\tau=e^{M_1}/(1+e^{M_1})^2\asymp e^{-M_1}$. Hence $\wll(\cdot)$ is $\tau$-strongly convex and $\frac{1}{4}$-smooth. Further notice that $\popstar=0$, then by Lemma~\ref{lemma-stongly-convex}, 
\begin{equation*}
\begin{aligned}
D_1=\inner{\popt,\dallt} \geq \frac{\tau/4}{\tau + 1/4}\fnorm{\dallt}^2 + \frac{1}{\tau+1/4}\fnorm{\sigt-\sigp}^2.
\end{aligned}
\end{equation*}
By triangle inequality, 
\begin{equation*}
D_2\leq |\inner{\sigp - \bA, \dzzt}|+ 2|\inner{\sigp - \bA , \dat\one^\top}| + |\inner{\sigp - \bA , \dbt\bX}|.	
\end{equation*}
Recall that $\zeta_n = \lambdacond$, and so
\begin{equation*}
D_2\leq \frac{\zeta_n}{2} \|\dzzt\|_*  + \zeta_n\|\dat\one^\top\|_*+ \zeta_n\sqrt{k}\fnorm{\dbt X}.	
\end{equation*}
Notice that $\dzzt$ has rank at most $2k$, 
\begin{equation*}
D_2\leq \frac{\zeta_n\sqrt{2k}}{2}\fnorm{\dzzt}+ \zeta_n\fnorm{\dat\one^\top} + \zeta_n\sqrt{k}\|\dbt\bX\|_{\rm F}.	
\end{equation*}
Further by Cauchy-Schwarz inequality, there exists constant $C_2$ such that for any positive constant $c_2$ which we will specify later, 
\begin{equation*}
D_2\leq c_2 \left(\fnorm{\dzzt}^2 + 2\fnorm{\dat\one^\top}^2 + \|\dbt\bX\|_{\rm F}^2 \right) + \frac{C_2}{4c_2}\zeta_n^2k.	
\end{equation*}
By Lemma~\ref{lem:kernel-identifiability-nonconvex}, there exist constants $c_1, C_1$ such that
\begin{equation}
\begin{aligned}
D_1 - D_2 &\geq \left(\frac{(1-c_1)\tau}{4\tau + 1} - c_2\right) \left(\fnorm{\dzzt}^2+ 2\|\dat\one^\top\|_F^2 +\|\dbt\bX\|_{\rm F}^2 \right) \\
	& \quad + \frac{1}{\tau+1/4}\fnorm{\sigt-\sigp}^2  - C_1\fnorm{\gresid}^2 - \frac{C_2}{4c_2}\zeta_n^2k.
\end{aligned}
\label{eq:nonconvex-1}
\end{equation}
By Lemma~\ref{lemma-tu1}, 
\begin{equation*}
\begin{aligned}
	D_1 - D_2 &\geq \frac{2(\sqrt{2}-1)}{\kappa^2}\left(\frac{(1-c_1)\tau}{4\tau + 1} - c_2\right)e_t  + \frac{1}{\tau+1/4}\fnorm{\sigt-\sigp}^2  - C_1\fnorm{\gresid}^2	- \frac{C_2}{4c_2}\zeta_n^2 k.
\end{aligned}
\end{equation*}
To bound $D_3$, notice that $\dzt\dzt^\top$ is a positive semi-definite matrix,
\begin{equation*}
\begin{aligned}
D_3&\leq |\inner{\samt, \dzt\dzt^\top}|\leq \op{\samt}\left\|\dzt\dzt^\top\right\|_* \\
& = \op{\samt} \tr(\dzt\dzt^\top)\leq  \op{\samt}\fnorm{\dzt}^2\\
&=\op{\samt-\popt+\popt}\fnorm{\dzt}^2\\
&\leq \op{\samt-\popt}\fnorm{\dzt}^2+\op{\popt}\fnorm{\dzt}^2\\
&= \op{\sigp-\bA} \fnorm{\dzt}^2+\op{\sigt-\sigp}\fnorm{\dzt}^2\\
&\leq \frac{\zeta_n}{2}\fnorm{\dzt}^2+\fnorm{\sigt-\sigp}\fnorm{\dzt}^2.
\end{aligned}
\end{equation*}
By the assumption that $\fnorm{\dzt}\leq \frac{c_0}{e^{M_1}\kappa^2}\op{\tz}$, 
\begin{equation*}
\begin{aligned}
\fnorm{\sigt-\sigp}\fnorm{\dzt}^2&\leq \frac{c_0}{e^{M_1}\kappa^2}\fnorm{\sigt-\sigp}\fnorm{\dzt}\op{\tz}\\
&\leq c_3\fnorm{\sigt-\sigp}^2 + \frac{c_0}{4c_3 e^{M_1}\kappa^2}\fnorm{\dzt}^2\op{\tz}^2.
\end{aligned}
\end{equation*}
for any constant $c_3$ to be specified later. Then 
\begin{equation*}
\begin{aligned}
D_3\leq \left(\frac{\zeta_n}{2\op{\tz}^2}+ \frac{c_0}{4c_3e^{M_1}\kappa^2}\right) e_t + c_3\fnorm{\sigt-\sigp}^2.
\end{aligned}
\end{equation*}
By the assumption that $\op{\tz}^2\geq C_0\kappa^2\zeta_n e^{M_1}$ for sufficiently large constant $C_0$, 
\begin{equation}
D_3\leq \left(\frac{1}{2C_0e^{M_1}\kappa^2}+ \frac{c_0}{4c_3e^{M_1}\kappa^2}\right) e_t + c_3\fnorm{\sigt-\sigp}^2.
\label{eq:nonconvex-2}
\end{equation}
For $D_4$ simple algebra leads to
\begin{equation}
\begin{aligned}
D_4 &= |\inner{\popt,\gresid}| = |\inner{\sigt - \sigp, \gresid}|\leq \fnorm{\sigt - \sigp}\fnorm{\gresid} \\
&\leq c_4\fnorm{\sigt - \sigp}^2 + \frac{1}{4c_4}\fnorm{\gresid}^2
\end{aligned}
\label{eq:nonconvex-3}
\end{equation}
for any constant $c_4$ to be specified later.

We now turn to bounding $D_5$.
To this end, we upper bound its three terms separately as follows. First,
\begin{equation*}
\begin{aligned}
\fnorm{\samt \zt}^2&=\fnorm{\left(\samt-\popt\right) \zt+\popt \zt}^2\\
&\leq 2( \fnorm{(\samt-\popt)\zt}^2 +\fnorm{\popt\zt}^2 )\\
&\leq 2( \fnorm{(\sigp-A)\zt}^2 +\fnorm{(\sigt-\sigp)\zt}^2 )\\
&\leq 2 ( \opnorm{\sigp-A}^2\fnorm{\zt}^2 +\fnorm{\sigt-\sigp}^2\opnorm{\zt}^2 )\\
&\leq 2\Big( \frac{\zeta_n^2}{4}\fnorm{\zt}^2 + \op{\zt}^2 \fnorm{\sigt-\sigp}^2 \Big).
\end{aligned}
\end{equation*}
Next,
\begin{equation*}
\begin{aligned}
\|\samt \one\|^2&=\left\|\left(\samt-\popt\right) \one+\popt \one\right\|^2\\
&\leq 2\Big( \vnorm{\left(\samt-\popt\right)\one}^2 +\|\popt\one\|^2 \Big)\\
&\leq 2\Big( \vnorm{(\sigp-A)\one}^2 +\vnorm{(\sigt-\sigp)\one}^2 \Big)\\
&\leq 2n\Big(\frac{\zeta_n^2}{4} + \fnorm{\sigt-\sigp}^2\Big).
\end{aligned}
\end{equation*}
Furthermore,
\begin{equation*}
\begin{aligned}
\inner{\popt, \bX}^2&=\Big(\inner{\samt-\popt, \bX}+\inner{\popt, \bX}\Big)^2\\
&\leq 2\Big(\inner{\sigp-A, \bX}^2+\inner{\sigt-\sigp, \bX}^2\Big)\\
&\leq  2\Big(\zeta_n^2 k\fnorm{X}^2+\fnorm{\sigt-\sigp}^2\fnorm{X}^2\Big).
\end{aligned}
\end{equation*}
When $\text{dist}(\zt, \tz)\leq c\op{\tz}$, combining these inequalities yields
\begin{equation*}
\begin{aligned}
D_5&\leq \left(\frac{\op{\zt}^2}{\op{\tz}^2} \frac{\zeta_n^2k}{2} + \frac{\zeta_n^2}{4} + \frac{\zeta_n^2 k}{2}\right) + \left(\frac{2\op{\zt}^2}{\op{\tz}^2}  + \frac{3}{2}\right)\fnorm{\sigt-\sigp}^2.
\end{aligned}
\end{equation*}
By the assumption that $\fnorm{\dzt}\leq \frac{c_0}{e^{M_1}\kappa^2}\op{\tz}$ for some sufficiently small $c_0$, 
\begin{equation}
\begin{aligned}
D_5&\leq C_5 \left(\zeta_n^2k + \fnorm{\sigt-\sigp}^2\right).
\end{aligned}
\label{eq:nonconvex-4}
\end{equation}
Combining \eqref{nonconvex-obj}, \eqref{eq:nonconvex-1}, \eqref{eq:nonconvex-2}, \eqref{eq:nonconvex-3} and \eqref{eq:nonconvex-4}, we obtain
\begin{equation*}
\begin{aligned}
\et_{t+1}&\leq \et_t-\eta\left(\frac{2(\sqrt{2}-1)}{\kappa^2}\left(\frac{(1-c_1)\tau}{4\tau + 1} - c_2\right) - \frac{1}{2C_0e^{M_1}\kappa^2}+ \frac{c_0}{4c_3e^{M_1}\kappa^2} \right)e_t + \eta \left(C_1 + \frac{1}{4c_4}\right)\fnorm{\gresid}^2	 \\
& \quad - \left(\frac{1}{\tau + 1/4} - c_3 - c_4 - C_5\eta\right)\fnorm{\sigt-\sigp}^2 + \eta\frac{C_2}{4c_2}\zeta_n^2k + \eta^2C_5\zeta_n^2k 
\end{aligned}
\end{equation*}
where $c_2, c_3, c_4$ are arbitrary constants, $c_0$ is a sufficiently small constant, and $C_0$ is a sufficiently large constant. Notice that $\tau\asymp e^{-M_1}$. Choose $c_2 = c\tau$ and $c, c_3, c_4, \eta$ small enough such that 
\begin{equation*}
\begin{gathered}
2(\sqrt{2}-1)\left(\frac{(1-c_1)\tau}{4\tau + 1} - c_2\right) -  \frac{1}{2e^{M_1}C_0} - \frac{c_0}{4c_3e^{M_1}}> \wt{\rho} e^{-M_1}, \quad \text{and}\\
		\frac{1}{\tau + 1/4} - c_3 - c_4 - C_5\eta \geq 0,
\end{gathered}
\end{equation*}
for some positive constant $\wt{\rho}$. Recall that $\et_t\geq (1-\delta)e_t$. Then there exists a universal constant $C > 0$ such that 
\begin{equation*}
\begin{aligned}
\et_{t+1}&\leq \left(1- \frac{\eta}{e^{M_1}\kappa^2}\wt{\rho}(1-\delta)\right)\et_t + \eta C \left(\fnorm{\gresid}^2 + e^{M_1}\zeta_n^2k\right).  
\end{aligned}
\end{equation*}
The proof is completed by setting $\rho = (1-\delta)\wt{\rho}$.

\subsubsection{Proof of Lemma~\ref{lem:kernel-nonconvex-2}}
\label{sec:proof-44-lem3}

Note the claim is deterministic in nature and we prove by induction. 
At initialization we have
\begin{equation*}
\begin{aligned}
\fnorm{\Delta_{\bZ^0}}\leq \left(\frac{\et_0}{\op{Z^0}^2}\right)^{\frac{1}{2}}\leq \left(\frac{c_0^2}{4e^{2M_1}\kappa^4}\frac{\op{\tz}^4}{\op{Z^0}^2}\right)^\frac{1}{2}=\frac{c_0}{2e^{M_1}\kappa^2}\op{\tz}\frac{\op{\tz}}{\op{Z^0}}\leq \frac{c_0}{e^{M_1}\kappa^2}\op{\tz},
\end{aligned}
\end{equation*}
where the last inequality is obtained from 
%
\[
\op{Z^0} \geq \op{\tz} - \op{\Delta_{\bZ^0}} \geq \left( 1 - \frac{c_0}{2e^{M_1}\kappa^2} \right) \op{\tz} \geq \frac{3}{4}\op{\tz},
\]
where the second the the last inequalities are due to Assumption \ref{assump:init}.

Suppose the claim is true for all $t\leq t_0$, by Lemma~\ref{lem:kernel-nonconvex-1}, 
\begin{equation*}
\begin{aligned}
\et_{t_0+1}&\leq \left(1-\frac{\eta}{e^{M_1}\kappa^2}\rho\right)^{t_0}\et_0 + \eta C \left(\fnorm{\gresid}^2 + e^{M_1}\zeta_n^2k\right)  \\
&\leq \et_0 + \eta C \left(\fnorm{\gresid}^2 + e^{M_1}\zeta_n^2k\right)\\  
&\leq \frac{c_0^2}{4e^{2M_1}\kappa^4}\op{\tz}^4 + \eta C \left(\fnorm{\gresid}^2 + e^{M_1}\zeta_n^2k\right)\\
&=\frac{c_0^2}{e^{2M_1}\kappa^4}\op{\tz}^4  \left(\frac{1}{4}+ \eta\frac{Ce^{2M_1}\zeta_n^2\kappa^4}{c_0^2\op{\tz}^4}\left(\frac{\fnorm{\gresid}^2}{\zeta_n^2} + e^{M_1}k\right)\right)\\
&\leq \frac{c_0^2}{e^{2M_1}\kappa^4}\op{\tz}^4  \left(\frac{1}{4}+ \frac{C}{c_0^2C_1^2}\right).
\end{aligned}
\end{equation*}
Choosing $C_1$ large enough such that $C_1^2\geq  \frac{4C}{c_0^2}$, then
\begin{equation*}
\et_{t_0+1}\leq \frac{c_0^2}{2e^{2M_1}\kappa^4}\op{\tz}^4  
\end{equation*}
and therefore,  
\begin{equation*}
\fnorm{\Delta_{\bZ^{t_0+1}}}\leq \left(\frac{\et_{t_0+1}}{\op{\tz}^2}\right)^{\frac{1}{2}}\leq \frac{c_0}{\sqrt{2}e^{M_1}\kappa^2}\op{\tz}\frac{\op{\tz}}{\op{Z^0}} \leq  \frac{c_0}{e^{M_1}\kappa^2}\op{\tz}.
\end{equation*}
This completes the proof.

\subsection{Additional technique lemmas}

We state below additional technical lemmas used in the proofs.
\begin{lemma}[\cite{chung2002connected}]
Let $X_1, \cdots, X_n$ be independent Bernoulli random variables with $P(X_i = 1) = p_i$. For $S_n = \sum_{i=1}^n a_iX_i$ and $\nu = \sum_{i=1}^n a_i^2p_i$. Then we have 
\begin{equation*}
\begin{aligned}
&P(S_n-\E S_n < -\lambda) \leq \mathrm{exp}(-\lambda^2/2\nu),\\
&P(S_n-\E S_n > \lambda) \leq \mathrm{exp}\left(-\frac{\lambda^2}{2(\nu+a\lambda/3)}\right),
\end{aligned}
\end{equation*}
where $a=\max \{|a_1|, \cdots, |a_n|\} $. 
\label{lemma-sum-bernoulli}
\end{lemma}



\begin{lemma}[\cite{tu2015low}]
For any $\bZ_1, \bZ_2\in\mathbb{R}^{n\times k}$, we have 
\begin{equation*}
\mathrm{dist}(\bZ_1, \bZ_2)^2\leq \frac{1}{2(\sqrt{2}-1)\sigma_k^2(\bZ_1)}\fnorm{\bZ_1\bZ_1^\top-\bZ_2\bZ_2^\top}^2.
\end{equation*}
\label{lemma-tu1}
\end{lemma}

\begin{lemma}[\cite{tu2015low}]
For any $\bZ_1, \bZ_2\in\mathbb{R}^{n\times k}$ such that $\mathrm{dist}(\bZ_1, \bZ_2)\leq c\op{\bZ_1}$, we have 
\begin{equation*}
\fnorm{\bZ_1\bZ_1^\top-\bZ_2\bZ_2^\top}\leq (2+c)\op{\bZ_1}\mathrm{dist}(\bZ_1, \bZ_2).
\end{equation*}
\label{lemma-tu2}
\end{lemma}

\begin{lemma}[\cite{nesterov2004introductory}]
For a continuously differentiable function $f$, if it is $\mu$-strongly convex and $L$-smooth on a convex domain $\mathcal{D}$, say for any $\bx, \by\in \mathcal{D}$, 
\begin{equation*}
\frac{\mu}{2}\|\bx-\by\|^2\leq  f(\by)-f(\bx)-\inner{f'(\bx), \by-\bx}\leq \frac{L}{2}\|\bx-\by\|^2,
\end{equation*}
then
\begin{equation*}
\inner{f'(\bx)-f'(\by), \bx-\by}\geq \frac{\mu L}{\mu+L}\|\bx-\by\|^2+\frac{1}{\mu+L}\|f'(\bx)-f'(\by)\|^2 ,
\end{equation*}
and also
\begin{equation*}
\inner{f'(\bx)-f'(\by), \bx-\by}\geq \mu \|\bx-\by\|^2.
\end{equation*}
\label{lemma-stongly-convex}
\end{lemma}

\begin{lemma}[\cite{lei2015consistency}, \cite{gao2015achieving}]
Let $\bA$ be the symmetric adjacency matrix of a random graph on $n$ nodes in which edges occur independently. 
Let $\E[\bA_{ij}]=\bP_{ij}$ for all $i\neq j$ and $P_{ii}\in [0,1]$.
Assume that $n\max_{i,j}\bP_{ij}\leq d$. Then for any $C_0$, there is a constant $C=C(C_0)$ such that 
\begin{equation*}
\op{A-P} \leq C\sqrt{d + \log n}
\end{equation*}
with probability at least $1-n^{-C_0}$.
\label{lem:concen-A}
\end{lemma}

\begin{lemma}
Let $\bA$ be the symmetric adjacency matrix of a random graph of $n$ nodes in which edges occur independently. Let $\E[\bA_{ij}]=\bP_{ij}$ for all $i\neq j$ and $P_{ii}\in [0,1]$ for all $i$ and $X$ be deterministic with $X_{ii}=0$ for all $i$. Then, 
\begin{equation*}
| \langle \bA-\bP, \bX \rangle | \leq C\fnorm{X}
\end{equation*}
with probability at least $1- 2\mathrm{exp}(-C^2/8p_{\max}) - \mathrm{exp}(-C^2\fnorm{X}/8\|X\|_{\infty})$, where $p_{\max} = \max_{i\neq j}P_{ij}$.
\label{lemma-AX}
\end{lemma}
\begin{proof}
Observe that $\langle \bA-\bP, \bX \rangle = 2\sum_{i < j} (A_{ij} - P_{ij}) X_{ij}$ and $A_{ij}$ are independent Bernoulli random variables with $\E[A_{ij}] = P_{ij}$. 
Apply Lemma~\ref{lemma-sum-bernoulli} to $\sum_{i < j} (A_{ij} - P_{ij}) X_{ij}$ with $\lambda = C\fnorm{X}/2$, we have $\nu = \sum_{i<j}X_{ij}^2P_{ij}\leq p_{\max}\fnorm{X}^2$ and  
\begin{equation*}
\begin{aligned}
P\big(|\langle \bA-\bP, \bX \rangle | \leq C\fnorm{X}\big)&\leq \mathrm{exp}(-C^2\fnorm{X}^2/8\nu) + \mathrm{exp}\left(-\frac{C^2\fnorm{X}^2}{8\max\left\{\nu, C\|X\|_{\infty}\fnorm{X}\right\}}\right)\\
&\leq 2\mathrm{exp}(-C^2\fnorm{X}^2/8\nu) + \mathrm{exp}(-C^2\fnorm{X}/8\|X\|_{\infty})\\
&\leq 2\mathrm{exp}(-C^2/8p_{\max}) + \mathrm{exp}(-C^2\fnorm{X}/8\|X\|_{\infty}).
\end{aligned}
\end{equation*}
This completes the proof.
\end{proof}

\appendix

\section{Proofs of results for initialization}

\label{sec:appendix}

This section presents the proofs of Theorem \ref{thm:init-range-kernel}, Corollary \ref{cor:init-algo-prob} and Proposition \ref{prop:svt-init}. 

\subsection{Preliminaries}

We introduce two technical results used repeatedly in the proofs.

\begin{lemma}
	If $\op{\tg}^2 \geq C\fnorm{\gresid}^2$ for some constant $C>0$, then $\op{\tg}^2 \geq c\fnorm{\tg}^2/k$ for some constant $c>0$. 
	\label{lem:Z}
\end{lemma}
\begin{proof}
By definition,
\begin{equation*}
\begin{aligned}
\fnorm{\tg}^2 &\leq 2 \left(\fnorm{\gresid}^2 + \fnorm{\tz\tz^\top}^2\right) 
\leq 2\fnorm{\gresid}^2 + 2k\op{\tz}^4 
\leq (2k+1/C)\op{\tz}^4.
\end{aligned}
\end{equation*}
Therefore, $\op{\tz}^4 \geq c\fnorm{\tg}^2/k$ for some constant $c>0$. 
\end{proof}

\begin{theorem}
Under Assumption~\ref{assump:srank}, choose $\lambda_n, \gamma_n$ such that 
\begin{equation}
	\lambda_n\geq 2 \max\left\{\op{\bA-\bP} + \gamma_n\op{\tg}, \op{\bA-\bP} + \frac{\gamma_n}{\sqrt{k}}\fnorm{\ta\one^\top}, \frac{\inner{\bA-\bP, \bX}}{\sqrt{k}\fnorm{X}} + \frac{\gamma_n}{\sqrt{k}}\fnorm{X\tb}\right\}.
\label{eq:init-lambda}
\end{equation}
Let the constant step size $\eta\leq 2/9$ and the constraint sets $\projg, \proja$ and $\projb$ as specified in Theorem \ref{thm:init-range-kernel}.
If the latent vectors contain strong enough signal in the sense that 
\begin{equation}
	\op{\tg}^2\geq C\kappa_{\tz}^6e^{2M_1}\max \Big\{
	e^{2M_1}\lambda_n^2k , ~\|\gresid\|_*^2/k , ~ \fnorm{\gresid}^2
  \Big\}
  \label{eq:Z-op1}
\end{equation}
for some sufficiently large constant $C$, then for any given constant $c_1>0$, there exists a universal constant $C_1$ such that for any $ T \geq T_0$, the error will satisfy $e_T\leq c_1^2e^{-2M_1}\op{\tz}^4/\kappa_{\tz}^4$, where
\begin{equation*}
	T_0= \log \left(\frac{C_1e^{2M_1}k\kappa_{\tz}^6}{c_1^2}\frac{\fnorm{\tg}^2 + 2\fnorm{\ta\one^\top}^2 + \fnorm{X\tb}^2 }{\fnorm{\tg}^2}\right)\left(\log\left(\frac{1}{1-\gamma_n\eta}\right)\right)^{-1}.
\end{equation*}
\label{thm:init-algo}
\end{theorem} 
\begin{proof}
See \prettyref{sec:thm-deter-init}.
\end{proof}

\subsection{Proof of Theorem~\ref{thm:init-range-kernel}}
We focus on the case where $X$ is nonzero, and the case of $X=0$ is simpler.
By Lemma~\ref{lem:prob}, there exist constants $C_2, c$ such that with probability at least $1-n^{c}$, 
\begin{equation*}
\begin{aligned}
\opnorm{A - P},~\frac{\inner{\bA-\bP, \bX}}{\sqrt{k}\fnorm{X}}\leq C_2\sqrt{\max\left\{ne^{-M_2}, \log n\right\}}.
\end{aligned}
\end{equation*}
All the following analysis is conditional on this event. Since $\fnorm{\ta\one^\top}, \fnorm{\tb X} \leq C\fnorm{\tg}$, by Lemma~\ref{lem:Z}, 
\begin{equation*}
	\fnorm{\ta\one^\top} + \fnorm{X\tb} \leq C_3\sqrt{k}\op{\tg}. 
\end{equation*}
for some constant $C_3$. 
Combining these two inequalities leads to 
\begin{equation*}
\begin{aligned}
& ~~~\max\left\{\op{\bA-\bP} + \gamma_n\op{\tg}, \op{\bA-\bP} + \frac{\gamma_n}{\sqrt{k}}\fnorm{\ta\one^\top}, \frac{\inner{\bA-\bP, \bX}}{\sqrt{k}\fnorm{X}} + \frac{\gamma_n}{\sqrt{k}}\fnorm{X\tb}\right\} \\
& \leq C_2\sqrt{\max\left\{ne^{-M_2}, \log n\right\}} + (1+C_3)\gamma_n \op{\tg} 
\leq  C_2/C_0 \lambda_n + (1+C_3)\delta\lambda_n 
\leq \lambda_n/2.
\end{aligned}
\end{equation*}
Here the last inequality is due to fact that $C_0$ is sufficiently large and $\delta$ is sufficiently small. Furthermore, 
\begin{equation*}
	\wt{C}\kappa_{\tz}^6e^{4M_1}\lambda_n^2k \leq \wt{C}c_0^2\op{\tg}^2\leq \op{\tg}^2,
\end{equation*}
since $c_0$ is a sufficient small constant. Therefore, the inequality~\eqref{eq:Z-op1} holds. Apply Theorem~\ref{thm:init-algo}, there exists a universal constant $C_1$ such that for any given constant $c_1>0$, $e_T\leq c_1^2e^{-2M_1}\op{\tz}^4/\kappa_{\tz}^4$, as long as $T\geq T_0$, where 
\begin{equation*}
\begin{aligned}
	T_0 = \log \left(\frac{C_1e^{2M_1}k\kappa_{\tz}^6}{c_1^2}\frac{\|\tx\|_{\mathcal{D}}^2}{\fnorm{\tg}^2}\right)\left(\log\left(\frac{1}{1-\gamma_n\eta}\right)\right)^{-1} .
\end{aligned}
\end{equation*}
Notice that when $\fnorm{\ta\one^\top}, \fnorm{\tb X} \leq C\fnorm{\tg}$, $\|\tx\|_{\mathcal{D}}^2\leq C_4\fnorm{\tg}^2$ for some constant $C_4$. Therefore, 
\begin{equation*}
T_0 \leq \log \left(\frac{C_1C_4e^{2M_1}k\kappa_{\tz}^6}{c_1^2}\right)\left(\log\left(\frac{1}{1-\gamma_n\eta}\right)\right)^{-1}.
\end{equation*}
This completes the proof.

\subsection{Proof of Corollary~\ref{cor:init-algo-prob}}
By Lemma~\ref{lem:prob}, there exist constants $C_2, c_2$ such that with probability at least $1-n^{c_2}$, 
\begin{equation*}
\begin{aligned}
\opnorm{A - P},~\frac{\inner{\bA-\bP, \bX}}{\sqrt{k}\fnorm{X}}\leq C_2\sqrt{\max\left\{ne^{-M_2}, \log n\right\}}.
\end{aligned}
\end{equation*}
All the following analysis is conditional on this event. Since $\fnorm{\ta\one^\top}, \fnorm{\tb X} \leq C\fnorm{\tg}$, by Lemma~\ref{lem:Z}, 
\begin{equation*}
	\fnorm{\ta\one^\top} + \fnorm{X\tb} \leq C_3\sqrt{k}\op{\tg}. 
\end{equation*}
for some constant $C_3$. 
Combining these two inequalities leads to 
\begin{equation*}
\begin{aligned}
& ~~~\max\left\{\op{\bA-\bP} + \gamma_n\op{\tg}, \op{\bA-\bP} + \frac{\gamma_n}{\sqrt{k}}\fnorm{\ta\one^\top}, \frac{\inner{\bA-\bP, \bX}}{\sqrt{k}\fnorm{X}} + \frac{\gamma_n}{\sqrt{k}}\fnorm{X\tb}\right\} \\
& \leq C_2\sqrt{\max\left\{ne^{-M_2}, \log n\right\}} + (1+C_3)\gamma_n \op{\tg} \\
& \leq C_2\gamma_n\op{\tg} + (1+C_3)\gamma_n \op{\tg} = (1 + C_2 + C_3)\gamma_n\op{\tg},
\end{aligned}
\end{equation*}
where the last inequality is due to equation~\eqref{eq:Z-op2}. Since we choose $\lambda_n = C_0\gamma_n\op{\tz}^2$ for some sufficiently large constant $C_0$, inequality~\eqref{eq:init-lambda} holds. Further, notice that $\gamma_n = \gamma = c_0/(e^{2M_1}\sqrt{k}\kappa_{\tz}^3)$ for some sufficiently small constant $c_0$,
\begin{equation*}
\begin{aligned}
\wt{C}e^{4M_1}\kappa_{\tz}^6\lambda_n^2 k = \wt{C}c_0^2\op{\tg}^2\leq \op{\tg}^2.
\end{aligned}
\end{equation*}
Therefore, the inequality~\eqref{eq:Z-op1} holds. Apply Theorem~\ref{thm:init-algo}, there exists a universal constant $C_1$ such that for any given constant $c_1>0$, $e_T\leq c_1^2e^{-2M_1}\op{\tz}^4/\kappa_{\tz}^4$, as long as $T\geq T_0$, where 
\begin{equation*}
\begin{aligned}
	T_0 = \log \left(\frac{C_1e^{2M_1}k\kappa_{\tz}^6}{c_1^2}\frac{\|\tx\|_{\mathcal{D}}^2}{\fnorm{\tg}^2}\right)\left(\log\left(\frac{1}{1-\gamma\eta}\right)\right)^{-1}.
\end{aligned}
\end{equation*}
Notice that when $\fnorm{\ta\one^\top}, \fnorm{\tb X} \leq C\fnorm{\tg}$, $\|\tx\|_{\mathcal{D}}^2\leq C_4\fnorm{\tg}^2$ for some constant $C_4$. Therefore, 
\begin{equation*}
T_0 \leq \log \left(\frac{C_1C_4e^{2M_1}k\kappa_{\tz}^6}{c_1^2}\right)\left(\log\left(\frac{1}{1-\gamma\eta}\right)\right)^{-1}.
\end{equation*}
This completes the proof.

\subsection{Proof of Proposition~\ref{prop:svt-init}}
Applying Theorem 2.7 in \cite{chatterjee2015matrix} we obtain
\[
\frac{1}{n^2}\fnorm{\wh{P} - P}^2 \leq C(k, M_1, \kappa_{\tz}) n^{-\frac{1}{k+3}}. 
\]
where the constant $C(k, M_1, \kappa_{\tz})$ depends on $k, M_1, \kappa_{\tz}$. Notice that $\Theta_{ij} = \logit (P_{ij})$ and $\logit(\cdot)$ is $4e^{M_1}$-Lipchitz continuous in the interval $\left[\frac{1}{2}e^{-M_1}, \frac{1}{2}\right]$, and so
\[
\frac{1}{n^2}\fnorm{\wh{\Theta} - \Theta}^2 \leq C'(k, M_1, \kappa_{\tz}) n^{-\frac{1}{k+3}}. 
\]
Let $\Delta_{\wh{\Theta}} = \wh{\Theta} - \tall$ It is easy to verify, 
\begin{equation*}
\begin{aligned}
\alpha^0 &= \left(2n \bI_n + 2\one\one^\top\right)^{-1}\wh{\Theta}\one 
= \ta + \frac{1}{n} \left(\bI_n - \frac{1}{2n}\one\one^\top\right)\Delta_{\wh{\Theta}}\one,
\end{aligned}
\end{equation*}
and hence 
\begin{equation*}
\begin{aligned}
\fnorm{\alpha^0\one^\top - \ta\one^\top} &=  \frac{1}{n}\fnorm{\left(\bI_n - \frac{1}{2n}\one\one^\top\right)\Delta_{\wh{\Theta}}\one\one^\top}\\
& \leq \frac{1}{n} \op{\bI_n - \frac{1}{2n}\one\one^\top} \fnorm{\Delta_{\wh{\Theta}}}\fnorm{\one\one^\top}
\leq \fnorm{\Delta_{\wh{\Theta}}}. 
\end{aligned}
\end{equation*}
Notice that $\tg\in\mathbb{S}_+^n$,
\begin{equation*}
\begin{aligned}
\fnorm{\wh{G} - \tg} \leq  \fnorm{ \wh{G} - J\wh{\Theta}J + J\wh{\Theta}J - \tg} \leq 2\fnorm{J\wh{\Theta}J - \tg}\leq 2\fnorm{\Delta_{\wh{\Theta}}}. 
\end{aligned}
\end{equation*}
Further notice that $\rank(\tg) = k$, 
\begin{equation*}
\begin{aligned}
\fnorm{Z^0(Z^0)^\top - \tg} \leq  \fnorm{ Z^0(Z^0)^\top - \wh{G} + \wh{G} - \tg} \leq 2\fnorm{ \wh{G} - \tg} \leq 4\fnorm{\Delta_{\wh{\Theta}}}. 
\end{aligned}
\end{equation*}
Then, by Lemma~\ref{lemma-tu1}, 
\begin{equation*}
\begin{aligned}
e_0 &\leq \op{\tz}^2 \text{dist}(Z^0, \tz)^2 + 2n\vnorm{\alpha^0 - \ta}^2 \\
& \leq \frac{\kappa_{\tz}^2}{2(\sqrt{2}-1)}\fnorm{Z^0(Z^0)^\top - \tg}^2 +  2n\vnorm{\alpha^0 - \ta}^2
\leq 24\kappa_{\tz}^2\fnorm{\Delta_{\wh{\Theta}}}^2 + 2\fnorm{\Delta_{\wh{\Theta}}}^2 \\
& \leq 26\kappa_{\tz}^2 C'(k, M_1, \kappa_{\tz}) n^2\times n^{-\frac{1}{k+3}} 
\leq \frac{26k \kappa_{\tz}^2 C'(k, M_1, \kappa_{\tz})}{c_0}\frac{ \fnorm{\tg}^2}{k}\times n^{-\frac{1}{k+3}}\\
& \leq C_1(k, M_1, \kappa_{\tz})\op{\tz}^4\times n^{-\frac{1}{k+3}}.
\end{aligned}
\end{equation*}
Therefore, the initialization condition in Assumption~\ref{assump:init} will hold for large enough $n$.

\subsection{Proof of Theorem~\ref{thm:init-algo}}
\label{sec:thm-deter-init}

\subsubsection{Preparations}
Recall the definition of $f(G, \alpha, \beta)$ in \eqref{eq:f}
where  
\begin{equation*}
(\bG, \balpha, \beta) \in \mathcal{D}=\left\{(\bG, \balpha, \beta) | \bG\bJ=\bG, \bG\in\mathbb{S}_+, \max_{i, j} |\bG_{ij}|, \max_{i}|\balpha_i|\leq \frac{M}{3}, |\beta|\leq \frac{M}{3\max_{i,j}|\bX_{ij}|}\right\}.
 \end{equation*} 
 Define the norm $\|\cdot\|_{\mathcal{D}}$ in the domain $\mathcal{D}$ by 
$$\|(\bG, \balpha, \beta)\|_{\mathcal{D}}=\Big(\fnorm{\bG}^2+2\fnorm{\balpha\one^\top}^2+\fnorm{\bX\beta}^2\Big)^{\half}.$$

\begin{lemma}
The function $f$ is $\gamma_n$-strongly convex and $(\gamma_n + 9/2)$-smooth in the convex domain $\mathcal{D}$ with respect to the norm  $\|\cdot\|_{\mathcal{D}}$, that is, for $(\bG_i, \balpha_i, \beta_i)\in\mathcal{D},  i=1 ,2$,  let $(\Delta_{\bG}, \Delta_{\balpha}, \Delta_{\beta})=(\bG_1-\bG_2, \balpha_1-\balpha_2, \beta_1-\beta_2)$, then
\begin{equation*}
\begin{aligned}
\frac{\gamma_n}{2}\|(\Delta_{\bG}, \Delta_{\balpha}, \Delta_{\beta})\|_{\mathcal{D}}^2&\leq  f(\bG_1, \balpha_1, \beta_1)-f(\bG_2, \balpha_2, \beta_2)-\inner{\nabla_G f(G_2, \alpha_2, \beta_2) , \Delta_{\bG}}\\
&~~~ -\inner{\nabla_\alpha f(G_2, \alpha_2, \beta_2) , \Delta_{\balpha}}-\inner{\nabla_\beta f(G_2, \alpha_2, \beta_2) , \Delta_{\beta}}\\
& \leq \frac{\gamma_n + 9/2}{2}\|(\Delta_{\bG}, \Delta_{\balpha}, \Delta_{\beta})\|_{\mathcal{D}}^2.
\end{aligned}
\end{equation*}
\label{lem:init-convex-smooth}
\end{lemma}
\begin{proof}
With slight abuse of notation, let 
\begin{equation}
	h(G, \alpha, \beta) = -\sum_{i, j}\left\{ \bA_{ij}\bTheta_{ij}+\log\Big(1-\sigma(\bTheta_{ij})\Big)\right\}
	\label{eq:h-new}
\end{equation}
which is a convex function of $G,\alpha$ and $\beta$.
In addition, let
\begin{equation}
	r(G, \alpha, \beta) = \frac{\gamma_n}{2} \left(\fnorm{G}^2 + 2\fnorm{\alpha\one^\top}^2 + \fnorm{X\beta}^2\right)  +  \lambda_n \tr(\bG)
	\label{eq:r}
\end{equation}
which is $\gamma_n$-strongly convex w.r.t. the norm $\|\cdot\|_{\mathcal{D}}$. Thus $f(\bG, \balpha, \beta)$ is $\gamma_n$-strongly convex. On the other hand, $r(\cdot, \cdot, \cdot)$ is $\gamma_n$ smooth and 
\begin{equation*}
\begin{aligned}
&~~~\wl(\bG_1, \balpha_1, \beta_1)-\wl(\bG_2, \balpha_2, \beta_2) -\inner{\nabla_G \wl(G_2, \alpha_2, \beta_2) , \Delta_{\bG}} \\
& ~~~~~~~~~~~ -\inner{\nabla_\alpha \wl(G_2, \alpha_2, \beta_2) , \Delta_{\alpha}}-\inner{\nabla_\beta \wl(G_2, \alpha_2, \beta_2) , \Delta_{\beta}}\\
 &=\wl(\bTheta_1)-\wl(\bTheta_2)-\inner{\nabla_{\Theta} \wl(\bTheta_2) ,  \Delta_{\bG}}-\inner{2\nabla_{\Theta} \wl(\bTheta_2)\one,   \Delta_{\alpha}}-\inner{\nabla_{\Theta} \wl(\bTheta_2),  \bX}  \Delta_{\beta}\\
 &=\frac{1}{2}\wl(\bTheta_1)-\wl(\bTheta_2)-\inner{\nabla_{\Theta} \wl(\bTheta_2) , \bTheta_1-\bTheta_2}\\
 &\leq \frac{1}{8} \fnorm{\bTheta_1-\bTheta_2}^2= \frac{1}{8} \fnorm{\Delta_{\bG} + 2\Delta_{\alpha}\one^\top + X\Delta_{\beta}}^2\\
 &\leq \frac{9}{8} \left(\fnorm{\Delta_{\bG}}^2 + 4\fnorm{\Delta_{\alpha}\one^\top}^2 + \fnorm{X\Delta_{\beta}}^2 \right)\leq \frac{9}{4} \left(\fnorm{\Delta_{\bG}}^2 + 2\fnorm{\Delta_{\alpha}\one^\top}^2 + \fnorm{X\Delta_{\beta}}^2 \right).
\end{aligned}
\end{equation*}
This finishes the proof. 
\end{proof}

Define $(\wt{G}, ~\wt{\alpha},~ \wt{\beta}) = \arg\min_{(\bG, \balpha, \beta) \in \mathcal{D}} ~~  f(G, \alpha, \beta)$, $\dtg = \wt{G} - \tg, \dta = \wt{\alpha} - \ta, \dtb = \wt{\beta} - \tb, \dtall = \wt{\Theta} - \tall$. Similar to the analysis of the convex programming in~\eqref{eq:convex-obj}, one can obtain the following results. 
\begin{lemma}
	\label{lem:init-cone}
Let $\mathcal{M}^\perp_k=\{\bM\in\mathbb{R}^{n\times n}: \row(\bM)\subset \col(\tz)^{\perp}\,\,\,\mathrm{and}\,\,\, \col(\bM)\subset \col(\tz)^{\perp}\}$ 
and $\mathcal{M}_k$ be its orthogonal complement in $\mathbb{R}^{n\times n}$ under trace inner product. If 
\begin{equation*}
	\lambda_n\geq 2 \max\left\{\op{\bA-\bP} + \gamma_n\op{\tg}, \op{\bA-\bP} + \frac{\gamma_n}{\sqrt{k}}\fnorm{\ta\one^\top}, \frac{\inner{\bA-\bP, \bX}}{\sqrt{k}\fnorm{X}} + \frac{\gamma_n}{\sqrt{k}}\fnorm{X\tb}\right\},
\end{equation*}
then for $\gresid = \calP_{\mathcal{M}_k^\perp} \tg$, we have
\begin{equation*}
\|\dtg \|_* \leq 4\sqrt{2k}\fnorm{\mathcal{P}_{\mathcal{M}_k}\dtg	} +  2\sqrt{k}\fnorm{\dta\one^\top} + \sqrt{k}\fnorm{X\dtb} + 4\|\gresid\|_*\,.
\end{equation*}
\end{lemma}
\begin{proof}
Let 
\begin{equation*}
\wtl (G, \alpha, \beta)=-\sum_{1\leq i, j\leq n}\{ \bA_{ij}\bTheta_{ij} + \log(1-\sigma(\bTheta_{ij}))\} + \frac{\gamma_n}{2} \left(\fnorm{G}^2 + 2\fnorm{\alpha\one^\top}^2 + \fnorm{X\beta}^2\right).
\end{equation*}
By the convexity of $\wtl$,
\begin{align*}
& \wtl (\wt{G}, \wt{\alpha}, \wt{\beta}) -\wtl(\tg,\ta,\tb)  \\
& \geq  \inner{\nabla_{G}\wtl(\tg,\ta,\tb), \dtg} + \inner{\nabla_{\alpha}\wtl(\tg,\ta,\tb), \dta} + \inner{\nabla_{\beta}\wtl(\tg,\ta,\tb), \dtb}\\
& = - \langle  \bA-\bP,\, \dtg+2 \dta\one^\top +\dtb \bX \rangle + \gamma_n \left(\inner{\tg, \dtg} + 2n\inner{\ta, \dta} + \fnorm{X}^2 \inner{\tb, \dtb}\right) \\
&\geq  -\op{\bA-\bP}\big(\|\dtg\|_*+2\|\dta\one^\top\|_*\big) 
- | \langle  \bA-\bP, \dtb\bX \rangle| \\
& ~~~ - \gamma_n \left(\op{\tg}\nuc{\tg} + 2\fnorm{\ta\one^\top}\fnorm{\dta\one^\top} + \fnorm{X\tb}\fnorm{X\dtb}\right) \\
&\geq  - \left(\op{\bA-\bP} + \gamma_n\op{\tg}\right) \nuc{\dtg} - \left(\op{\bA-\bP} + \gamma_n/\sqrt{k}\fnorm{\ta\one^\top} \right)2\sqrt{k}\fnorm{\dta\one^\top} \\
&~~~ - \left(\inner{\bA-\bP, \bX}/\left(\sqrt{k}\fnorm{X}\right) + \gamma_n/\sqrt{k}\fnorm{X\tb} \right)\sqrt{k}\fnorm{X\dtb}\\
&\geq-\frac{\lambda_n}{2}
\big(\nuc{\dtg}+2\sqrt{k}\fnorm{\dta\one^\top} + \sqrt{k}\fnorm{X\dtb}\big) \\
&\geq -\frac{\lambda_n}{2}
\big(\|\mathcal{P}_{\mathcal{M}_k}\dtg\|_*+\|\mathcal{P}_{\mathcal{M}^\perp_k}\dtg\|_*+2\sqrt{k}\fnorm{\dta\one^\top} + \sqrt{k}\fnorm{X\dtb}\big).
\end{align*}
The last inequality holds since $\calP_{\calM_k}+\calP_{\calM_k^\perp}$ equals identity and 
\begin{equation*}
\lambda_n\geq 2 \max\left\{\op{\bA-\bP} + \gamma_n\op{\tg}, \op{\bA-\bP} + \frac{\gamma_n}{\sqrt{k}}\fnorm{\ta\one^\top}, \frac{\inner{\bA-\bP, \bX}}{\sqrt{k}\fnorm{X}} + \frac{\gamma_n}{\sqrt{k}}\fnorm{X\tb}\right\}.
\end{equation*}
On the other hand, by the definition of $\gresid$, 
\begin{equation*}
\begin{aligned}
\|\wt{G}\|_*-\|\tg\|_*&=\|\mathcal{P}_{\mathcal{M}_k}\tg+ \gresid + \mathcal{P}_{\mathcal{M}_k}\dtg+\mathcal{P}_{\mathcal{M}^\perp_k}\dtg\|_*-\|\mathcal{P}_{\mathcal{M}_k}\tg+ \gresid\|_*\\
&\geq \|\mathcal{P}_{\mathcal{M}_k}\tg+\mathcal{P}_{\mathcal{M}^\perp_k}\dtg\|_* -\|\gresid\|_*-\|\mathcal{P}_{\mathcal{M}_k}\dtg\|_* - \|\mathcal{P}_{\mathcal{M}_k}\tg\|_*- \|\gresid\|_*\\
&=\|\mathcal{P}_{\mathcal{M}_k}\tg\|_*+\|\mathcal{P}_{\mathcal{M}^\perp_k}\dtg\|_* - 2\|\gresid\|_* - \|\mathcal{P}_{\mathcal{M}_k}\dtg\|_* - \|\mathcal{P}_{\mathcal{M}_k}\tg\|_* \\
&=\|\mathcal{P}_{\mathcal{M}^\perp_k}\dtg\|_*-\|\mathcal{P}_{\mathcal{M}_k}\dtg\|_* - 2\|\gresid\|_*\,.
\end{aligned}
\end{equation*} 
Here, the second last equality holds since $\calP_{\calM_k}\tg$ and $\mathcal{P}_{\mathcal{M}^\perp_k}\dtg$ have orthogonal column and row spaces.
Furthermore, since $\hbTheta$ is the optimal solution to \eqref{eq:convex-obj}, and $\tall$ is feasible, the basic inequality and the last two displays imply
\begin{equation*}
\begin{aligned}
0 &\geq \wtl (\wt{G}, \wt{\alpha}, \wt{\beta}) - \wtl(\tg,\ta,\tb) +\lambda_n \big( \|\wt{G}\|_* - \|\tg\|_* \big)\\
&\geq -\frac{\lambda_n}{2}
\big(\|\mathcal{P}_{\mathcal{M}_k}\dtg\|_*+\|\mathcal{P}_{\mathcal{M}^\perp_k}\dtg\|_*+2\sqrt{k}\fnorm{\dta\one^\top} + \sqrt{k}\fnorm{X\dtb}\big) \\
& ~~~ + \lambda_n\big(\|\mathcal{P}_{\mathcal{M}^\perp_k}\dtg\|_*-\|\mathcal{P}_{\mathcal{M}_k}\dtg\|_* - 2\|\gresid\|_*\big)\\
& =  \frac{\lambda_n}{2} \big(\|\mathcal{P}_{\mathcal{M}^\perp_k}\dtg\|_*- 3\|\mathcal{P}_{\mathcal{M}_k}\dtg\|_*  - 4\|\gresid\|_* - 2\sqrt{k}\fnorm{\dta\one^\top} - \sqrt{k}\fnorm{X\dtb}\big).
\end{aligned}
\end{equation*}
Rearranging the terms leads to
\begin{equation*}
\|\mathcal{P}_{\mathcal{M}^\perp_k}\dtg \|_* \leq 3\|\mathcal{P}_{\mathcal{M}_k}\dtg \|_* +  2\sqrt{k}\fnorm{\dta\one^\top} + \sqrt{k}\fnorm{X\dtb} + 4\|\gresid\|_*\,,
\end{equation*}
and triangle inequality further implies 
\begin{equation*}
\|\dtg \|_* \leq 4\|\mathcal{P}_{\mathcal{M}_k}\dtg \|_* +  2\sqrt{k}\fnorm{\dta\one^\top} + \sqrt{k}\fnorm{X\dtb} + 4\|\gresid\|_*\,.
\end{equation*}
Finally, note that the rank of $\mathcal{P}_{\mathcal{M}_k}\dtg$ is at most $2k$,
\begin{equation*}
\|\dtg \|_* \leq 4\sqrt{2k} \fnorm{\mathcal{P}_{\mathcal{M}_k}\dtg} +  2\sqrt{k}\fnorm{\dta\one^\top} + \sqrt{k}\fnorm{X\dtb} + 4\|\gresid\|_*\,.
\end{equation*}
This completes the proof.	
\end{proof}

\begin{lemma} 
\label{lem:init-identifiability}
For any $k\geq 1$ such that Assumption~\ref{assump:srank} holds. Choose $\lambda_n\geq \max\{2\opnorm{\bA-\bP}, 1\}$ and $| \langle \bA-\bP, \bX \rangle | \leq \lambda_n \sqrt{k}\fnorm{X}$. 
There exist constants $C>0$ and $0\leq c<1$ such that
\begin{equation*}
\begin{aligned}
\fnorm{\dtall}^2 & \geq (1 - c)
\big(\fnorm{\dtg}^2 + 2\fnorm{\dta\one^\top}^2+\fnorm{\dtb\bX}^2 \big) - C\|\gresid\|_*^2/k,\quad\mathrm{and} \\
\fnorm{\dtall}^2 & \leq (1 + c) 
\big(\fnorm{\dtg}^2 + 2\fnorm{\dta\one^\top}^2+\fnorm{\dtb\bX}^2 \big) + C\|\gresid\|_*^2/k.
\end{aligned}
\end{equation*} 
\end{lemma}
\begin{proof}
The proof is the same as the proof of Lemma~\ref{lem:kernel-identifiability} and we leave out the details. 
\end{proof}

\begin{theorem}
Under Assumption~\ref{assump:srank}, for any $\lambda_n$ satisfying 
\begin{equation*}
	\lambda_n\geq 2 \max\left\{\op{\bA-\bP} + \gamma_n\op{\tg}, \op{\bA-\bP} + \frac{\gamma_n}{\sqrt{k}}\fnorm{\ta\one^\top}, \frac{\inner{\bA-\bP, \bX}}{\sqrt{k}\fnorm{X}} + \frac{\gamma_n}{\sqrt{k}}\fnorm{X\tb}\right\},
\end{equation*}
there exists a constant $C$ such that
\begin{equation*}
\left(\fnorm{\dtg} + 2\fnorm{\dta\one^\top} +\fnorm{\dtb\bX}\right)^2 \leq C\left(e^{2M_1}\lambda_n^2k + \frac{\|\gresid\|_*^2}{k}\right).
\end{equation*}
\label{thm:init}
\end{theorem}
\begin{proof}
Recall the definition of $h(G,\alpha,\beta)$ in \eqref{eq:h-new}.	
Observe that $\hbTheta=\hbalpha\one^\top+\one\hbalpha^\top+\hbeta\bX+\hbG$ is the optimal solution to \eqref{eq:convex-obj}, and that the true parameter $\tall=\ta\one^\top+\one\ta^\top+\tb\bX+\tg$ is feasible. 
Thus, we have the basic inequality
\begin{equation}
\wtl (\wt{G}, \wt{\alpha}, \wt{\beta}) -\wtl(\tg,\ta,\tb)+\lambda_n (\|\wt{G}\|_*-\|\tg\|_* )\leq 0.
\label{eq:basic2}
\end{equation}
By definition, 
\begin{equation*}
	\wtl (G, \alpha, \beta) = \wl(G, \alpha, \beta) + \frac{\gamma_n}{2} \left(\fnorm{G}^2 + \|\alpha\one^\top\|_2^2 + \fnorm{X\beta}^2\right) .
\end{equation*}
On the one hand, 
\begin{equation*}
\begin{aligned}
	& \wl(\wt{G}, \wt{\alpha}, \wt{\beta}) - \wl(\tg,\ta,\tb)  
	\\ & 
	~~~ -\inner{\nabla_{G}\wl(\tg,\ta,\tb), \dtg}  -  \inner{\nabla_{\alpha}\wl(\tg,\ta,\tb), \dta} - \inner{\nabla_{\beta}\wl(\tg,\ta,\tb), \dtb} \\
	 & = h(\wt{\Theta}) - h(\tall) - \inner{\nabla_{\Theta}h(\tall), \dtall} \geq  \frac{\tau}{2}\fnorm{\dtall}^2 ,
\end{aligned}
\end{equation*}
where the last inequality is by the strong convexity of $h(\cdot)$ with respect to $\Theta$ in the domain $\mathcal{F}_g$ and $\tau=e^{M_1}/(1+e^{M_1})^2$ as in the proof of Theorem~\ref{thm:convex}. Further by Lemma~\ref{lem:init-identifiability}, 
\begin{equation*}
 	\frac{\tau}{2}\fnorm{\dtall}^2 \geq \frac{\tau(1 - c)}{2}
\big(\fnorm{\dtg}^2 + 2\fnorm{\dta\one^\top}^2+\fnorm{\dtb\bX}^2 \big) - \frac{C\tau}{2}\|\gresid\|_*^2/k.
 \end{equation*} 
On the other hand, the $l_2$ regularization term is strongly convex with respect to $(G, \alpha, \beta)$. Then we have 
\begin{equation*}
\begin{aligned}
	& \wtl (\wt{G}, \wt{\alpha}, \wt{\beta}) -\wtl(\tg,\ta,\tb) \\
	&  \geq \inner{\nabla_{G}\wtl(\tg,\ta,\tb), \dtg}  +  \inner{\nabla_{\alpha}\wtl(\tg,\ta,\tb), \dta} + \inner{\nabla_{\beta}\wtl(\tg,\ta,\tb), \dtb} \\
	& ~~~ + \frac{\tau(1 - c)}{2}
\big(\fnorm{\dtg}^2 + 2\fnorm{\dta\one^\top}^2+\fnorm{\dtb\bX}^2 \big) - \frac{C\tau}{2}\|\gresid\|_*^2/k\\
& \geq -\frac{\lambda_n}{2}
\big(\nuc{\dtg}+2\sqrt{k}\fnorm{\dta\one^\top} + \sqrt{k}\fnorm{X\dtb}\big) \\
& ~~~ + \frac{\tau(1 - c)}{2}
\big(\fnorm{\dtg}^2 + 2\fnorm{\dta\one^\top}^2+\fnorm{\dtb\bX}^2 \big) - \frac{C\tau}{2}\|\gresid\|_*^2/k.
\end{aligned}
\end{equation*}
By triangle inequality, 
\begin{equation*}
\lambda_n (\|\hbG\|_*-\|\tg\|_* )\geq -\lambda_n\|\Delta_{\bG}\|_*.
\end{equation*}
Together with \eqref{eq:basic2}, the last two inequalities imply
\begin{equation*}
\begin{aligned}
& \frac{\tau(1 - c)}{2}
\left(\fnorm{\dtg}^2 + 2\fnorm{\dta\one^\top}^2 +\fnorm{\dtb\bX}^2\right)\\
&~~~ \leq \frac{\lambda_n}{2} \left(\nuc{\dtg}+2\sqrt{k}\fnorm{\dta\one^\top}+ \sqrt{k}\fnorm{X\dtb}\right) 
+ \lambda_n\|\dtg\|_* + \frac{C\tau}{2}\|\gresid\|_*^2/k.
\end{aligned}
\end{equation*}
By Lemma~\ref{lem:init-cone}, 
\begin{equation*}
\begin{aligned}
&\frac{\tau(1 - c)}{2}
\left(\fnorm{\dtg}^2 + 2\fnorm{\dta\one^\top}^2 +\fnorm{\dtb\bX}^2\right)\\
& ~~~ \leq C_0\lambda_n \sqrt{k}\left(\fnorm{\dtg}+2\fnorm{\dta\one^\top}+\fnorm{X\dtb}\right)
 + C_1 \lambda_n\|\gresid\|_* + \frac{C\tau}{2}\|\gresid\|_*^2/k.
\end{aligned}
\end{equation*}
This implies that there exists some constant $c_0$ such that 
\begin{equation*}
\begin{aligned}
 	& c_0\tau \left(\fnorm{\dtg} + 2\fnorm{\dta\one^\top} +\fnorm{\dtb\bX}\right)^2 \\ 
	& ~~~~~ \leq  C_0\lambda_n \sqrt{k}\left(\fnorm{\dtg}+2\fnorm{\dta\one^\top}+\fnorm{X\dtb}\right) 
  + C_1\lambda_n\|\gresid\|_* + \frac{C\tau}{2}\|\gresid\|_*^2/k.
\end{aligned}
\end{equation*}
Solving the quadratic inequality, there exists some constant $C_2$ such that
\begin{equation*}
\left(\fnorm{\dtg} + 2\fnorm{\dta\one^\top} +\fnorm{\dtb\bX}\right)^2 \leq C_2\left(\frac{\lambda_n^2k}{\tau^2} + \frac{\lambda_n \|\gresid\|_*}{\tau} + \frac{\|\gresid\|_*^2}{k}\right).
\end{equation*}
Note that $\tau\geq c_1e^{-M_1}$ and $e^{M_1}\lambda_n \|\gresid\|_* \leq c_2 \left(e^{2M_1}\lambda_n^2k  +  \frac{\|\gresid\|_*^2}{k}\right)$ for positive constants $c_1, c_2$. 
Therefore, 
\begin{equation*}
\left(\fnorm{\dtg} + 2\fnorm{\dta\one^\top} +\fnorm{\dtb\bX}\right)^2 \leq C_2\left(e^{2M_1}\lambda_n^2k + \frac{\|\gresid\|_*^2}{k}\right),
\end{equation*}
which completes the proof.
\end{proof}

\begin{lemma}[\cite{bubeck2014theory}]
Let $\bx\in\mathcal{D}$ and $\by\in\mathbb{R}^n$, then 
\begin{equation*}
\inner{\pi_{\mathcal{D}}(y)-x, \pi_\mathcal{D}(y)-y} \leq 0
\end{equation*}
where $\mathcal{D}$ is a convex set and $\pi_{\mathcal{D}}(x)=\arg\min_{\by\in\mathcal{D}}\|\bx-\by\|$.
\label{lemma-projection}
\end{lemma}

\begin{lemma}
With $\etag=\eta, \etaa=\eta/2n, \etab=\eta/\fnorm{X}^2$,
\begin{equation*}
\begin{aligned}
&\quad \inner{\gt-\bG^{t+1}, \gt-\wt{G}}+ 2n\inner{\at-\balpha^{t+1}, \at-\wt{\alpha} }+\inner{\bt-\beta^{t+1}, \bt-\wt{\beta}}\fnorm{X}^2\\
&~~~\geq \frac{\eta\mu}{2}\|\xt-\wt{x}\|_{\mathcal{D}}^2+\left(1-\frac{\eta L}{2}\right)\Big\{\fnorm{\bG^{t+1}-\gt}^2+2\fnorm{\left(\balpha^{t+1}-\balpha^{t}\right)\one^\top}^2+\fnorm{(\beta^{t+1}-\beta^{t})\bX}^2\Big\}
\end{aligned}
\end{equation*}
where $\mu = \gamma_n$ and $L= \gamma_n + 9/2.$
\label{lem:init-descent}
\end{lemma}
\begin{proof}
Let $\xt=(\gt, \at, \bt)$ and $\wt{x} = (\wt{G}, \wt{\alpha}, \wt{\beta})$.
Then
\begin{equation*}
\begin{aligned}
f(&\bx^{t+1})-f(\wt{x})=f(\bx^{t+1})-f(\xt)+f(\xt)-f(\wt{x})\\
&\leq \inner{\nabla f(\xt), \bx^{t+1}-\bx_{t}}+\frac{L}{2}\|\bx^{t+1}-\xt\|_{\mathcal{D}}^2+\inner{\nabla f(\xt), \bx_{t}-\wt{x}}-\frac{\mu}{2}\|\xt-\wt{x}\|_{\mathcal{D}}^2\\
&\leq \inner{\nabla f(\xt), \bx^{t+1}-\wt{x}}+\frac{L}{2}\|\bx^{t+1}-\xt\|_{\mathcal{D}}^2-\frac{\mu}{2}\|\xt-\wt{x}\|_{\mathcal{D}}^2\\
&=\inner{\nabla_G f(G^t, \alpha^t, \beta^t) , \bG^{t+1}-\wt{G}}+\inner{\nabla_\alpha f(G^t, \alpha^t, \beta^t) , \balpha^{t+1}-\wt{\alpha}}+\inner{\nabla_\beta f(G^t, \alpha^t, \beta^t) , \beta^{t+1}-\wt{\beta}}\\
&\quad+\frac{L}{2}\|\bx^{t+1}-\xt\|_{\mathcal{D}}^2-\frac{\mu}{2}\|\xt-\wt{x}\|_{\mathcal{D}}^2.
\end{aligned}
\end{equation*}
Notice that $\widetilde{\bG}^{t+1}=\gt-\etag\frac{\partial f}{\partial \bG}|_{\bG=\gt}$ and $\bG^{t+1}$ is the projection of $\widetilde{\bG}^{t+1}$ to the convex set $\{\bG| \bG\bJ=\bG, \bG\in\mathbb{S}_+, \max_{i,j}\|\bG_{ij}\|\leq M_1\}$. Therefore by Lemma~\ref{lemma-projection}, 
\begin{equation*}
\inner{\bG^{t+1}-\widetilde{\bG}^{t+1}, \bG^{t+1}-\wt{G}}\leq 0
\end{equation*}
which implies that 
\begin{equation*}
\begin{aligned}
\inner{\frac{\partial f}{\partial \bG}|_{\bG=\gt} , \bG^{t+1}-\wt{G}}&\leq \frac{1}{\etag}\inner{\gt-\bG^{t+1}, \bG^{t+1}-\wt{G}}\\
&=\frac{1}{\etag}\inner{\gt-\bG^{t+1}, \gt-\wt{G}}-\frac{1}{\etag}\fnorm{\gt-\bG^{t+1}}^2.
\end{aligned}
\end{equation*}
Similar argument will yield 
\begin{equation*}
\begin{gathered}
\inner{\frac{\partial f}{\partial \balpha}|_{\balpha=\balpha^{t+1}} , \balpha^{t+1}-\wt{\alpha}}\leq \frac{1}{\etaa}\inner{\at-\balpha^{t+1}, \balpha^{t}-\wt{\alpha}}-\frac{1}{\etaa}\|\at-\balpha^{t+1}\|^2, \\
\inner{\frac{\partial f}{\partial \beta}|_{\beta=\beta^{t+1}} , \beta^{t+1}-\wt{\beta}}\leq \frac{1}{\etab}\inner{\bt-\beta^{t+1}, \beta^{t}-\wt{\beta}}-\frac{1}{\etab}\|\bt-\beta^{t+1}\|^2.
\end{gathered}
\end{equation*}
Also notice that $f(\bx^{t+1})-f(\wt{x})\geq 0$, therefore
\begin{equation*}
\begin{aligned}
0&\leq \eta(f(\bx^{t+1})-f(\wt{x}))\leq \inner{\gt-\bG^{t+1}, \gt-\wt{G}}+ 2n\inner{\at-\balpha^{t+1}, \at-\wt{\alpha}}\\
&\quad+\fnorm{X}^2\inner{\bt-\beta^{t+1}}-\|\xt-\bx^{t+1}\|_{\mathcal{D}}^2+\frac{\eta L}{2}\|\xt-\bx^{t+1}\|_{\mathcal{D}}^2-\frac{\eta\mu}{2}\|\xt-\wt{x}\|_{\mathcal{D}}^2.
\end{aligned}
\end{equation*}
This completes the proof.
\end{proof}

\subsubsection{Proof of the theorem}
Let $\xt=(\gt, \at, \bt)$, $\wt{x} = (\wt{G}, \wt{\alpha}, \wt{\beta})$. By definition, 
\begin{equation*}
\begin{aligned}
\| \bx^{t+1}-\wt{x}\|_{\mathcal{D}}^2&= \fnorm{\bG^{t+1}-\wt{G}}^2+2\fnorm{\left(\balpha^{t+1}-\balpha^{t}\right)\one^\top}^2+\fnorm{(\beta^{t+1}-\wt{\beta})\bX}^2.
\end{aligned}
\end{equation*}
Notice that for each component, the error can be decomposed as (with $G$ as an example),
\begin{equation*}
 \fnorm{\bG^{t+1}-\wt{G}}^2= \fnorm{\gt-\wt{G}}^2-2\inner{\gt-\bG^{t+1}, \gt-\wt{G}}+\fnorm{\bG^{t+1}-\gt}^2.
\end{equation*} 
Summing up these equations leads to
\begin{equation*}
\begin{aligned}
\| \bx^{t+1}-\wt{x}\|_{\mathcal{D}}^2
&=\|\xt-\wt{x}\|_{\mathcal{D}}^2 \\
& ~~~ -2\Big\{\inner{\gt-\bG^{t+1}, \gt-\wt{G}}+ 2n\inner{\at-\balpha^{t+1}, \at-\wt{\alpha}}
+\fnorm{X}^2\inner{\bt-\beta^{t+1}, \bt-\wt{\beta}} \Big\}\\
& ~~~ +\Big\{\fnorm{\bG^{t+1}-\gt}^2+2\fnorm{\left(\balpha^{t+1}-\balpha^{t}\right)\one^\top}^2 +\fnorm{(\beta^{t+1}-\beta^{t})\bX}^2\Big\}.
\end{aligned}
\end{equation*}
By Lemma~\ref{lem:init-descent}, 
\begin{equation*}
\begin{aligned}
\| \bx^{t+1}-\wt{x}\|_{\mathcal{D}}^2\leq (1-\eta\mu)\|\xt-\wt{x}\|_{\mathcal{D}}^2 - (1-\eta L)\|\xt-\bx^{t+1}\|_{\mathcal{D}}^2 .
\end{aligned}
\end{equation*}
Then for any $\eta\leq 1/L$, 
\begin{equation*}
\begin{aligned}
\| \bx^{t+1}-\wt{x}\|_{\mathcal{D}}^2\leq (1-\eta\mu)\|\xt-\wt{x}\|_{\mathcal{D}}^2.
\end{aligned}
\end{equation*}
By Lemma~\ref{lem:kernel-identifiability-nonconvex}, and repeatedly using the inequality $(a+b)^2\leq 2(a^2+b^2)$, 
\begin{equation*} 
\begin{aligned}
e_t &\leq \frac{\kappa^2_{\tz}}{2(\sqrt{2}-1)}\fnorm{Z^t(Z^t)^\top - \tz\tz^\top}^2+2\fnorm{\dat\one^\top}^2+\fnorm{\dbt X}^2 \\
& \leq \frac{\kappa^2_{\tz}}{(\sqrt{2}-1)} \left(\fnorm{Z^t(Z^t)^\top - G^t}^2 + \fnorm{G^t - \tz\tz^\top}^2\right)+2\fnorm{\dat\one^\top}^2+\fnorm{\dbt X}^2\\
& \leq \frac{2\kappa^2_{\tz}}{(\sqrt{2}-1)}\fnorm{G^t - \tz\tz^\top}^2+2\fnorm{\dat\one^\top}^2+\fnorm{\dbt X}^2\\
& \leq \frac{4\kappa^2_{\tz}}{(\sqrt{2}-1)}\left(\fnorm{G^t - \tg}^2 + \fnorm{\tg- \tz\tz^\top}^2\right) +2\fnorm{\dat\one^\top}^2+\fnorm{\dbt X}^2\\
& \leq \frac{4\kappa^2_{\tz}}{(\sqrt{2}-1)}\fnorm{G^t - \tg}^2 + \frac{4\kappa^2_{\tz}}{(\sqrt{2}-1)}\fnorm{\gresid}^2 + 2\fnorm{\dat\one^\top}^2+\fnorm{\dbt X}^2.
\end{aligned}  
\end{equation*}
By the definition of $\|\cdot\|_\mathcal{D}$, we further have
\begin{align*}
e_t & \leq \frac{4\kappa^2_{\tz}}{(\sqrt{2}-1)}\left(\|\xt-\tx\|_{\mathcal{D}}^2 + \fnorm{\gresid}^2\right)
\leq \frac{4\kappa^2_{\tz}}{(\sqrt{2}-1)}\left(2\|\xt-\wt{x}\|_{\mathcal{D}}^2 + 2\|\wt{x}-\tx\|_{\mathcal{D}}^2 + \fnorm{\gresid}^2\right)\\
&\leq \frac{4\kappa^2_{\tz}}{(\sqrt{2}-1)}\left(2 (1-\eta\gamma_n)^t \|x^0-\wt{x}\|_{\mathcal{D}}^2 + 2\|\wt{x}-\tx\|_{\mathcal{D}}^2 + \fnorm{\gresid}^2\right)\\
&\leq \frac{4\kappa^2_{\tz}}{(\sqrt{2}-1)}\left(4 (1-\eta\gamma_n)^t \|x^0-\tx\|_{\mathcal{D}}^2 + 4 (1-\eta\gamma_n)^t \|\wt{x}-\tx\|_{\mathcal{D}}^2 + 2\|\wt{x}-\tx\|_{\mathcal{D}}^2 + \fnorm{\gresid}^2\right).
\end{align*}
According to Theorem~\ref{thm:init}, there exists constant $C_0>0$ such that 
$$\|\wt{x}-\tx\|_{\mathcal{D}}^2\leq C_0\left(e^{2M_1}\lambda_n^2k + \frac{\|\gresid\|_*^2}{k}\right).$$
Therefore, $e_t\leq  C_1\kappa_{\tz}^2\left((1-\eta\gamma_n)^t\|x^0-\tx\|_{\mathcal{D}}^2  + e^{2M_1}\lambda_n^2k + \|\gresid\|_*^2/k + \fnorm{\gresid}^2\right)$. Since $x^0 = 0$, 
\begin{equation*}
\begin{aligned}
e_t&\leq  C_1\kappa_{\tz}^2\left((1-\eta\gamma_n)^t\|\tx\|_{\mathcal{D}}^2  + e^{2M_1}\lambda_n^2k + \|\gresid\|_*^2/k + \fnorm{\gresid}^2\right)\\
&\leq \frac{c_1^2}{\kappa_{\tz}^4e^{2M_1}} \op{\tz}^4 \times \frac{C_1\kappa_{\tz}^6e^{2M_1}}{c_1^2\op{\tz}^4} \left((1-\eta\gamma_n)^t\|\tx\|_{\mathcal{D}}^2  + e^{2M_1}\lambda_n^2k + \|\gresid\|_*^2/k + \fnorm{\gresid}^2\right).
\end{aligned}
\end{equation*}
Under our assumptions, there exists some sufficiently large constant $C_2$ such that 
\begin{equation*}
	\op{\tz}^4\geq C_2\kappa_{\tz}^6e^{2M_1}\max \Big\{
	e^{2M_1}\lambda_n^2k , ~\|\gresid\|_*^2/k , ~ \fnorm{\gresid}^2
  \Big\}.
\end{equation*}
Therefore, 
\begin{equation*}
\begin{aligned}
e_t &\leq  \frac{c_1^2}{\kappa_{\tz}^4e^{2M_1}} \op{\tz}^4  \times \left(\frac{C_1e^{2M_1}\kappa_{\tz}^6\fnorm{\tall}^2}{c_1^2\tau^2\op{\tz}^4}(1-\eta\gamma_n)^t + \frac{3C_1}{c_1^2C_2} \right).
\end{aligned}
\end{equation*}
Choose large enough $C_2 > 6C_1/c_1^2$, then 
\begin{equation*}
\begin{aligned}
e_t &\leq  \frac{c_1^2}{\kappa_{\tz}^4e^{2M_1}} \op{\tz}^4  \times \left(\frac{Ce^{2M_1}\kappa_{\tz}^6\|\tx\|_{\mathcal{D}}^2}{c_1^2\op{\tz}^4}(1-\eta\gamma_n)^t + \frac{1}{2} \right).
\end{aligned}
\end{equation*}
Therefore, $e_t\leq  \frac{c_1^2\tau^2}{\kappa^4}\op{\tz}^4$ when 
\[
\frac{C_1e^{2M_1}\kappa_{\tz}^6\|\tx\|_{\mathcal{D}}^2}{c_1^2\op{\tg}^2}(1-\eta\gamma_n)^t\leq \frac{1}{2}.
\]
By Lemma~\ref{lem:Z}, $\op{\tg}^2 \geq c\fnorm{\tg}^2/k$. Therefore, it suffices to have 
$$t\geq \log \left(k\frac{\|\tx\|_{\mathcal{D}}^2}{\fnorm{\tg}^2}\frac{2C_1e^{2M_1}\kappa_{\tz}^6}{c_1^2c}\right)\left(\log\left(\frac{1}{1-\eta\gamma_n}\right)\right)^{-1}.$$
This completes the proof.

\bibliography{network}

\begin{thebibliography}{60}
\providecommand{\natexlab}[1]{#1}
\providecommand{\url}[1]{\texttt{#1}}
\expandafter\ifx\csname urlstyle\endcsname\relax
  \providecommand{\doi}[1]{doi: #1}\else
  \providecommand{\doi}{doi: \begingroup \urlstyle{rm}\Url}\fi

\bibitem[Adamic and Glance(2005)]{adamic2005political}
L.~A. Adamic and N.~Glance.
\newblock The political blogosphere and the 2004 us election: divided they
  blog.
\newblock In \emph{Proceedings of the 3rd International Workshop on Link
  Discovery}, pages 36--43. ACM, 2005.

\bibitem[Agarwal et~al.(2012)Agarwal, Negahban, and
  Wainwright]{agarwal2012noisy}
A.~Agarwal, S.~Negahban, and M.~J. Wainwright.
\newblock Noisy matrix decomposition via convex relaxation: Optimal rates in
  high dimensions.
\newblock \emph{The Annals of Statistics}, 40:\penalty0 1171--1197, 2012.

\bibitem[Agarwal and Chen(2009)]{agarwal2009regression}
D.~Agarwal and B.-C. Chen.
\newblock Regression-based latent factor models.
\newblock In \emph{Proceedings of the 15th ACM SIGKDD international conference
  on Knowledge discovery and data mining}, pages 19--28. ACM, 2009.

\bibitem[Airoldi et~al.(2013)Airoldi, Costa, and Chan]{airoldi2013stochastic}
E.~M. Airoldi, T.~B. Costa, and S.~H. Chan.
\newblock Stochastic blockmodel approximation of a graphon: Theory and
  consistent estimation.
\newblock In \emph{Advances in Neural Information Processing Systems}, pages
  692--700, 2013.

\bibitem[Asta and Shalizi(2014)]{asta2014geometric}
D.~Asta and C.~R. Shalizi.
\newblock Geometric network comparison.
\newblock \emph{arXiv preprint arXiv:1411.1350}, 2014.

\bibitem[Bickel and Chen(2009)]{bickel2009nonparametric}
P.~J. Bickel and A.~Chen.
\newblock A nonparametric view of network models and newman--girvan and other
  modularities.
\newblock \emph{Proceedings of the National Academy of Sciences}, 106\penalty0
  (50):\penalty0 21068--21073, 2009.

\bibitem[Bogomolny et~al.(2003)Bogomolny, Bohigas, and
  Schmit]{bogomolny2003spectral}
E.~Bogomolny, O.~Bohigas, and C.~Schmit.
\newblock Spectral properties of distance matrices.
\newblock \emph{Journal of Physics A: Mathematical and General}, 36\penalty0
  (12):\penalty0 3595, 2003.

\bibitem[Bubeck(2014)]{bubeck2014theory}
S.~Bubeck.
\newblock Theory of convex optimization for machine learning.
\newblock \emph{arXiv preprint arXiv:1405.4980}, 2014.

\bibitem[Burer and Monteiro(2005)]{burer2005local}
S.~Burer and R.~D. Monteiro.
\newblock Local minima and convergence in low-rank semidefinite programming.
\newblock \emph{Mathematical Programming}, 103\penalty0 (3):\penalty0 427--444,
  2005.

\bibitem[Cand\`{e}s and Recht(2012)]{candes2012exact}
E.~Cand\`{e}s and B.~Recht.
\newblock Exact matrix completion via convex optimization.
\newblock \emph{Communications of the ACM}, 55\penalty0 (6):\penalty0 111--119,
  2012.

\bibitem[Cand{\`e}s and Tao(2010)]{candes2010power}
E.~J. Cand{\`e}s and T.~Tao.
\newblock The power of convex relaxation: Near-optimal matrix completion.
\newblock \emph{IEEE Transactions on Information Theory}, 56\penalty0
  (5):\penalty0 2053--2080, 2010.

\bibitem[Candes et~al.(2015)Candes, Eldar, Strohmer, and
  Voroninski]{candes2015phase1}
E.~J. Candes, Y.~C. Eldar, T.~Strohmer, and V.~Voroninski.
\newblock Phase retrieval via matrix completion.
\newblock \emph{SIAM review}, 57\penalty0 (2):\penalty0 225--251, 2015.

\bibitem[Cand\`{e}s et~al.(2015)Cand\`{e}s, Li, and
  Soltanolkotabi]{candes2015phase}
E.~J. Cand\`{e}s, X.~Li, and M.~Soltanolkotabi.
\newblock Phase retrieval via wirtinger flow: Theory and algorithms.
\newblock \emph{IEEE Transactions on Information Theory}, 61\penalty0
  (4):\penalty0 1985--2007, 2015.

\bibitem[Chatterjee(2015)]{chatterjee2015matrix}
S.~Chatterjee.
\newblock Matrix estimation by universal singular value thresholding.
\newblock \emph{The Annals of Statistics}, 43\penalty0 (1):\penalty0 177--214,
  2015.

\bibitem[Chen and Wainwright(2015)]{chen2015fast}
Y.~Chen and M.~J. Wainwright.
\newblock Fast low-rank estimation by projected gradient descent: General
  statistical and algorithmic guarantees.
\newblock \emph{arXiv preprint arXiv:1509.03025}, 2015.

\bibitem[Chen et~al.(2015)Chen, Li, and Xu]{chen2015convexified}
Y.~Chen, X.~Li, and J.~Xu.
\newblock Convexified modularity maximization for degree-corrected stochastic
  block models.
\newblock \emph{arXiv preprint arXiv:1512.08425}, 2015.

\bibitem[Cheng and Singer(2013)]{cheng2013spectrum}
X.~Cheng and A.~Singer.
\newblock The spectrum of random inner-product kernel matrices.
\newblock \emph{Random Matrices: Theory and Applications}, 2\penalty0
  (04):\penalty0 1350010, 2013.

\bibitem[Chung and Lu(2002)]{chung2002connected}
F.~Chung and L.~Lu.
\newblock Connected components in random graphs with given expected degree
  sequences.
\newblock \emph{Annals of combinatorics}, 6\penalty0 (2):\penalty0 125--145,
  2002.

\bibitem[Davenport et~al.(2014)Davenport, Plan, Van Den~Berg, and
  Wootters]{davenport20141}
M.~A. Davenport, Y.~Plan, E.~Van Den~Berg, and M.~Wootters.
\newblock 1-bit matrix completion.
\newblock \emph{Information and Inference: A Journal of the IMA}, 3\penalty0
  (3):\penalty0 189--223, 2014.

\bibitem[Dykstra(1983)]{dykstra1983algorithm}
R.~L. Dykstra.
\newblock An algorithm for restricted least squares regression.
\newblock \emph{Journal of the American Statistical Association}, 78\penalty0
  (384):\penalty0 837--842, 1983.

\bibitem[El~Karoui(2010)]{el2010spectrum}
N.~El~Karoui.
\newblock The spectrum of kernel random matrices.
\newblock \emph{The Annals of Statistics}, 38\penalty0 (1):\penalty0 1--50,
  2010.

\bibitem[Gao et~al.(2015{\natexlab{a}})Gao, Lu, Zhou, et~al.]{gao2015rate}
C.~Gao, Y.~Lu, H.~H. Zhou, et~al.
\newblock Rate-optimal graphon estimation.
\newblock \emph{The Annals of Statistics}, 43\penalty0 (6):\penalty0
  2624--2652, 2015{\natexlab{a}}.

\bibitem[Gao et~al.(2015{\natexlab{b}})Gao, Ma, Zhang, and
  Zhou]{gao2015achieving}
C.~Gao, Z.~Ma, A.~Y. Zhang, and H.~H. Zhou.
\newblock Achieving optimal misclassification proportion in stochastic block
  model.
\newblock \emph{arXiv preprint arXiv:1505.03772}, 2015{\natexlab{b}}.

\bibitem[Gao et~al.(2016)Gao, Lu, Ma, and Zhou]{gao2016optimal}
C.~Gao, Y.~Lu, Z.~Ma, and H.~H. Zhou.
\newblock Optimal estimation and completion of matrices with biclustering
  structures.
\newblock \emph{Journal of Machine Learning Research}, 17\penalty0
  (161):\penalty0 1--29, 2016.

\bibitem[Ge et~al.(2016)Ge, Lee, and Ma]{ge2016matrix}
R.~Ge, J.~D. Lee, and T.~Ma.
\newblock Matrix completion has no spurious local minimum.
\newblock In \emph{Advances in Neural Information Processing Systems}, pages
  2973--2981, 2016.

\bibitem[Goldenberg et~al.(2010)Goldenberg, Zheng, Fienberg, and
  Airoldi]{goldenberg2010survey}
A.~Goldenberg, A.~X. Zheng, S.~E. Fienberg, and E.~M. Airoldi.
\newblock A survey of statistical network models.
\newblock \emph{Foundations and Trends in Machine Learning}, 2\penalty0
  (2):\penalty0 129--233, 2010.

\bibitem[Handcock et~al.(2007)Handcock, Raftery, and
  Tantrum]{handcock2007model}
M.~S. Handcock, A.~E. Raftery, and J.~M. Tantrum.
\newblock Model-based clustering for social networks.
\newblock \emph{Journal of the Royal Statistical Society: Series A (Statistics
  in Society)}, 170\penalty0 (2):\penalty0 301--354, 2007.

\bibitem[Hoff(2003)]{hoff2003random}
P.~D. Hoff.
\newblock Random effects models for network data.
\newblock In \emph{Dynamic Social Network Modeling and Analysis: Workshop
  Summary and Papers}. Citeseer, 2003.

\bibitem[Hoff(2005)]{hoff2005bilinear}
P.~D. Hoff.
\newblock Bilinear mixed-effects models for dyadic data.
\newblock \emph{Journal of the american Statistical association}, 100\penalty0
  (469):\penalty0 286--295, 2005.

\bibitem[Hoff(2008)]{hoff2008modeling}
P.~D. Hoff.
\newblock Modeling homophily and stochastic equivalence in symmetric relational
  data.
\newblock In \emph{Advances in neural information processing systems}, pages
  657--664, 2008.

\bibitem[Hoff et~al.(2002)Hoff, Raftery, and Handcock]{hoff2002latent}
P.~D. Hoff, A.~E. Raftery, and M.~S. Handcock.
\newblock Latent space approaches to social network analysis.
\newblock \emph{Journal of the American Statistical Association}, 97\penalty0
  (460):\penalty0 1090--1098, 2002.

\bibitem[Jin(2015)]{jin2015fast}
J.~Jin.
\newblock Fast community detection by score.
\newblock \emph{The Annals of Statistics}, 43\penalty0 (1):\penalty0 57--89,
  2015.

\bibitem[Keshavan et~al.(2010{\natexlab{a}})Keshavan, Montanari, and
  Oh]{keshavan2010matrix}
R.~H. Keshavan, A.~Montanari, and S.~Oh.
\newblock Matrix completion from a few entries.
\newblock \emph{IEEE Transactions on Information Theory}, 56\penalty0
  (6):\penalty0 2980--2998, 2010{\natexlab{a}}.

\bibitem[Keshavan et~al.(2010{\natexlab{b}})Keshavan, Montanari, and
  Oh]{keshavan2010matrix1}
R.~H. Keshavan, A.~Montanari, and S.~Oh.
\newblock Matrix completion from noisy entries.
\newblock \emph{Journal of Machine Learning Research}, 11\penalty0
  (Jul):\penalty0 2057--2078, 2010{\natexlab{b}}.

\bibitem[Klopp et~al.(2015)Klopp, Tsybakov, and Verzelen]{klopp2015oracle}
O.~Klopp, A.~B. Tsybakov, and N.~Verzelen.
\newblock Oracle inequalities for network models and sparse graphon estimation.
\newblock \emph{arXiv preprint arXiv:1507.04118}, 2015.

\bibitem[Koltchinskii et~al.(2011)Koltchinskii, Lounici, and
  Tsybakov]{koltchinskii2011nuclear}
V.~Koltchinskii, K.~Lounici, and A.~B. Tsybakov.
\newblock Nuclear-norm penalization and optimal rates for noisy low-rank matrix
  completion.
\newblock \emph{The Annals of Statistics}, pages 2302--2329, 2011.

\bibitem[Krioukov et~al.(2010)Krioukov, Papadopoulos, Kitsak, Vahdat, and
  Bogun{\'a}]{krioukov2010hyperbolic}
D.~Krioukov, F.~Papadopoulos, M.~Kitsak, A.~Vahdat, and M.~Bogun{\'a}.
\newblock Hyperbolic geometry of complex networks.
\newblock \emph{Physical Review E}, 82\penalty0 (3):\penalty0 036106, 2010.

\bibitem[Krivitsky and Handcock(2008)]{krivitskyfitting}
P.~N. Krivitsky and M.~S. Handcock.
\newblock Fitting latent cluster models for networks with latentnet.
\newblock \emph{Journal of Statistical Software}, 24\penalty0 (i05), 2008.

\bibitem[Krivitsky et~al.(2009)Krivitsky, Handcock, Raftery, and
  Hoff]{krivitsky2009representing}
P.~N. Krivitsky, M.~S. Handcock, A.~E. Raftery, and P.~D. Hoff.
\newblock Representing degree distributions, clustering, and homophily in
  social networks with latent cluster random effects models.
\newblock \emph{Social Networks}, 31\penalty0 (3):\penalty0 204--213, 2009.

\bibitem[Lazega(2001)]{lazega2001collegial}
E.~Lazega.
\newblock \emph{The collegial phenomenon: The social mechanisms of cooperation
  among peers in a corporate law partnership}.
\newblock Oxford University Press on Demand, 2001.

\bibitem[Lei and Rinaldo(2015)]{lei2015consistency}
J.~Lei and A.~Rinaldo.
\newblock Consistency of spectral clustering in stochastic block models.
\newblock \emph{The Annals of Statistics}, 43\penalty0 (1):\penalty0 215--237,
  2015.

\bibitem[Ma(2013)]{ma2013sparse}
Z.~Ma.
\newblock Sparse principal component analysis and iterative thresholding.
\newblock \emph{The Annals of Statistics}, 41\penalty0 (2):\penalty0 772--801,
  2013.

\bibitem[Ma and Wu(2015)]{ma2015volume}
Z.~Ma and Y.~Wu.
\newblock Volume ratio, sparsity, and minimaxity under unitarily invariant
  norms.
\newblock \emph{IEEE Transactions on Information Theory}, 61\penalty0
  (12):\penalty0 6939--6956, 2015.

\bibitem[McCallum et~al.(2000)McCallum, Nigam, Rennie, and
  Seymore]{mccallum2000automating}
A.~K. McCallum, K.~Nigam, J.~Rennie, and K.~Seymore.
\newblock Automating the construction of internet portals with machine
  learning.
\newblock \emph{Information Retrieval}, 3\penalty0 (2):\penalty0 127--163,
  2000.

\bibitem[M{\'e}zard et~al.(1999)M{\'e}zard, Parisi, and Zee]{mezard1999spectra}
M.~M{\'e}zard, G.~Parisi, and A.~Zee.
\newblock {Spectra of Euclidean random matrices}.
\newblock \emph{Nuclear Physics B}, 559\penalty0 (3):\penalty0 689--701, 1999.

\bibitem[Micchelli et~al.(2006)Micchelli, Xu, and
  Zhang]{micchelli2006universal}
C.~A. Micchelli, Y.~Xu, and H.~Zhang.
\newblock Universal kernels.
\newblock \emph{Journal of Machine Learning Research}, 7\penalty0
  (Dec):\penalty0 2651--2667, 2006.

\bibitem[Nesterov(2004)]{nesterov2004introductory}
Y.~Nesterov.
\newblock \emph{Introductory lectures on convex optimization}, volume~87.
\newblock Springer Science \& Business Media, 2004.

\bibitem[Schoenberg(1937)]{schoenberg1937certain}
I.~J. Schoenberg.
\newblock {On certain metric spaces arising from Euclidean spaces by a change
  of metric and their imbedding in Hilbert space}.
\newblock \emph{Annals of Mathematics}, pages 787--793, 1937.

\bibitem[Schoenberg(1938)]{schoenberg1938metric}
I.~J. Schoenberg.
\newblock Metric spaces and positive definite functions.
\newblock \emph{Transactions of the American Mathematical Society}, 44\penalty0
  (3):\penalty0 522--536, 1938.

\bibitem[Sun and Luo(2016)]{sun2016guaranteed}
R.~Sun and Z.-Q. Luo.
\newblock Guaranteed matrix completion via non-convex factorization.
\newblock \emph{IEEE Transactions on Information Theory}, 62\penalty0
  (11):\penalty0 6535--6579, 2016.

\bibitem[Tang et~al.(2013)Tang, Sussman, and Priebe]{tang2013universally}
M.~Tang, D.~L. Sussman, and C.~E. Priebe.
\newblock Universally consistent vertex classification for latent positions
  graphs.
\newblock \emph{The Annals of Statistics}, 41\penalty0 (3):\penalty0
  1406--1430, 2013.

\bibitem[Traud et~al.(2011)Traud, Kelsic, Mucha, and
  Porter]{traud2011comparing}
A.~L. Traud, E.~D. Kelsic, P.~J. Mucha, and M.~A. Porter.
\newblock Comparing community structure to characteristics in online collegiate
  social networks.
\newblock \emph{SIAM review}, 53\penalty0 (3):\penalty0 526--543, 2011.

\bibitem[Traud et~al.(2012)Traud, Mucha, and Porter]{traud2012social}
A.~L. Traud, P.~J. Mucha, and M.~A. Porter.
\newblock Social structure of facebook networks.
\newblock \emph{Physica A: Statistical Mechanics and its Applications},
  391\penalty0 (16):\penalty0 4165--4180, 2012.

\bibitem[Tu et~al.(2015)Tu, Boczar, Soltanolkotabi, and Recht]{tu2015low}
S.~Tu, R.~Boczar, M.~Soltanolkotabi, and B.~Recht.
\newblock Low-rank solutions of linear matrix equations via procrustes flow.
\newblock \emph{arXiv preprint arXiv:1507.03566}, 2015.

\bibitem[Wolfe and Olhede(2013)]{wolfe2013nonparametric}
P.~J. Wolfe and S.~C. Olhede.
\newblock Nonparametric graphon estimation.
\newblock \emph{arXiv preprint arXiv:1309.5936}, 2013.

\bibitem[Wu et~al.(2017)Wu, Levina, and Zhu]{levina2017}
Y.-J. Wu, E.~Levina, and J.~Zhu.
\newblock Generalized linear models with low rank effects for network data.
\newblock \emph{arXiv preprint arXiv:1705.06772}, 2017.

\bibitem[Young and Scheinerman(2007)]{young2007random}
S.~J. Young and E.~R. Scheinerman.
\newblock Random dot product graph models for social networks.
\newblock In \emph{International Workshop on Algorithms and Models for the
  Web-Graph}, pages 138--149. Springer, 2007.

\bibitem[Zhang et~al.(2014)Zhang, Levina, and Zhu]{zhang2014detecting}
Y.~Zhang, E.~Levina, and J.~Zhu.
\newblock Detecting overlapping communities in networks using spectral methods.
\newblock \emph{arXiv preprint arXiv:1412.3432}, 2014.

\bibitem[Zhang et~al.(2015)Zhang, Levina, and Zhu]{zhang2015community}
Y.~Zhang, E.~Levina, and J.~Zhu.
\newblock Community detection in networks with node features.
\newblock \emph{arXiv preprint arXiv:1509.01173}, 2015.

\bibitem[Zheng and Lafferty(2016)]{lafferty2016convergence}
Q.~Zheng and J.~Lafferty.
\newblock Convergence analysis for rectangular matrix completion using
  burer-monteiro factorization and gradient descent.
\newblock \emph{stat}, 1050:\penalty0 23, 2016.

\end{thebibliography}
\bibliographystyle{abbrvnat}

\end{document}